\documentclass[letterpaper,11pt]{article}
\usepackage[margin=1in]{geometry}
\usepackage[utf8]{inputenc}
\usepackage[T1]{fontenc}
\usepackage[noadjust,compress]{cite}
\usepackage{a4wide}
\usepackage{amsmath, mathtools}
\usepackage[displaymath,mathlines,running]{lineno}
\usepackage{amssymb}
\usepackage{amsthm}   
\usepackage{enumerate}
\usepackage{thm-restate}
\usepackage{refcount}
\usepackage{caption}
\usepackage{subcaption}
\usepackage{xspace}
\usepackage{comment}
\usepackage{stmaryrd}
\usepackage{soul}
\usepackage{tikz}
\usetikzlibrary{fit,positioning, intersections, shapes.geometric,shapes.symbols,arrows,decorations.markings,decorations.pathreplacing,calc,ipe}
\usepackage{algorithm}
\usepackage[noend]{algpseudocode}

\usepackage{hyperref}
\hypersetup{colorlinks=true,citecolor=blue,linkcolor=blue,urlcolor=blue}
\usepackage{cleveref}
\renewcommand{\cref}{\Cref}
\theoremstyle{plain}
\newtheorem{thm}{Theorem}[section]
\newtheorem{cor}[thm]{Corollary}
\newtheorem{lem}[thm]{Lemma}
\newtheorem{obs}[thm]{Observation}

\newtheorem{clm}{Claim}[thm]

\newtheorem*{seth}{Strong Exponential-Time Hypothesis}
\newtheorem*{scc}{Set Cover Conjecture}
\newtheorem*{maxsathyp}{Max 3-Sat Hypothesis}

\crefname{thm}{Theorem}{Theorems}
\Crefname{thm}{Theorem}{Theorems}
\crefname{clm}{Claim}{Claims}
\crefname{clm}{Claim}{Claims}
\crefname{lem}{Lemma}{Lemmas}
\crefname{obs}{Observation}{Oberservations}

\Crefmultiformat{lem}{Lemmas~#2#1#3}%
{ and~#2#1#3}{, #2#1#3}{, and~#2#1#3}

\theoremstyle{definition}
\newtheorem{defn}[thm]{Definition}

\newenvironment{claimproof}[1][\unskip]{\noindent {\emph{Proof of Claim #1.\space}}}{\hfill$\triangleleft$ \smallskip}

\DeclareMathOperator*{\ar}{arity}

\newcommand{\N}{\mathbb{N}}
\newcommand{\bigO}{\mathcal{O}}

\newcommand{\calI}{\mathcal{I}}
\newcommand{\calJ}{\mathcal{J}}

\newcommand{\calS}{\mathcal{S}}
\newcommand{\calV}{\mathcal{V}}

\newcommand{\calT}{\mathcal{T}}

\newcommand{\calA}{\mathcal{A}}
\newcommand{\calB}{\mathcal{B}}
\newcommand{\calC}{\mathcal{C}}
\newcommand{\calD}{\mathcal{D}}
\newcommand{\calE}{\mathcal{E}}
\newcommand{\calF}{\mathcal{F}}
\newcommand{\calG}{\mathcal{G}}

\newcommand{\calZ}{\mathcal{Z}}

\newcommand{\boldq}{\mathbf{q}}
\newcommand{\boldr}{\mathbf{r}}
\newcommand{\bolds}{\mathbf{s}}
\newcommand{\boldx}{\mathbf{x}}

\newcommand{\boldf}{\mathbf{f}}
\newcommand{\boldc}{\mathbf{c}}
\newcommand{\boldd}{\mathbf{d}}
\newcommand{\boldz}{\mathbf{z}}
\newcommand{\boldp}{\mathbf{p}}
\newcommand{\boldI}{\mathbf{I}}

\newcommand{\rom}[1]{%
	\textup{\uppercase\expandafter{\romannumeral#1}}%
}

\newcommand{\NEQ}{\mathrm{NEQ}}
\newcommand{\OR}[1]{\mathrm{OR}_#1}

\newcommand{\NP}{\textsf{NP}}

\newcommand{\yes}{\texttt{YES}\xspace}
\newcommand{\no}{\texttt{NO}\xspace}

\newcommand{\coloring}[1]{\ensuremath{#1\textsc{-Coloring}}\xspace}
\newcommand{\listcoloring}[1]{\ensuremath{\textsc{List-}#1\textsc{-Coloring}}\xspace}
\newcommand{\coloringVD}[1]{\ensuremath{#1\textsc{-ColoringVD}}\xspace}
\newcommand{\coloringED}[1]{\ensuremath{#1\textsc{-ColoringED}}\xspace}
\newcommand{\listcoloringVD}[1]{\ensuremath{\textsc{List-}#1\textsc{-ColoringVD}}\xspace}
\newcommand{\listcoloringED}[1]{\ensuremath{\textsc{List-}#1\textsc{-ColoringED}}\xspace}

\newcommand{\maxcut}{\textsc{Max Cut}\xspace}
\newcommand{\maxsat}{\textsc{Max $3$-Sat}\xspace}

\newcommand{\MaxSat}{\maxsat}
\newcommand{\MaxCSP}[2]{\textsc{Max\,(#1,#2)-CSP}\xspace}

\newcommand{\SetCov}{\textsc{Set Cover}\xspace}
\newcommand{\msh}{M3SH\xspace}

\newcommand{\hard}{2^n\text{-hard}}
\newcommand{\hardness}{2^n\text{-hardness}}

\mathchardef\hyph="2D
\newcommand{\abs}[1]{|#1|}
\newcommand{\ceil}[1]{\lceil #1 \rceil}
\newcommand{\eps}{\varepsilon}
\renewcommand{\epsilon}{\varepsilon}
\renewcommand{\phi}{\varphi}
\renewcommand{\geq}{\geqslant}
\renewcommand{\leq}{\leqslant}
\renewcommand{\ge}{\geqslant}
\renewcommand{\le}{\leqslant}
\renewcommand{\tilde}{\widetilde}
\newcommand{\from}{\colon}

\newenvironment{myitemize}
{ \begin{itemize}
		\setlength{\itemsep}{0pt}
		\setlength{\parskip}{0pt}
		\setlength{\parsep}{0pt}     }
	{ \end{itemize}                  }

\newenvironment{myenumerate}[1][1.]
{  \begin{enumerate}[#1]
		\setlength{\itemsep}{0pt}
		\setlength{\parskip}{0pt}
		\setlength{\parsep}{0pt}     }
{ \end{enumerate} 				 }

\graphicspath{{figures/}}

\newcommand{\executeiffilenewer}[3]{%
\ifnum\pdfstrcmp{\pdffilemoddate{#1}}%
{\pdffilemoddate{#2}}>0%
{\immediate\write18{#3}}\fi%
} 

\usepackage{bm}
\newcommand{\del}{\bm{\times}}

\newcommand{\nh}{\Gamma}
\newcommand{\edcount}{\mathsf{cost}_\text{ed}}

\newcommand{\cost}{\mathsf{cost}}
\newcommand{\cdel}{\mathsf{c}_{\del}}

\newcommand{\cl}{\mathsf{clique}}
\newcommand{\gd}{\mathsf{gadget}}
\newcommand{\core}[2]{$(#1,#2)$-hub}

\newcommand{\sdhub}{(\sigma,\delta)\text{-hub}}

\newcommand{\cpar}{p}
\newcommand{\setf}[1]{\mathcal{F}_{#1, q}^{\del}}

\newcommand{\coreword}{hub\xspace}

\newcommand{\eqsetcover}[1]{\ensuremath{\mathord{=}{#1}}\textsc{-Set Cover}\xspace}
\newcommand{\leqsetcover}[1]{\ensuremath{\mathord{\leq}{#1}}\textsc{-Set Cover}\xspace}
\newcommand{\eqdsetcover}{\eqsetcover{d}}
\newcommand{\leqdsetcover}{\leqsetcover{d}}

\newcommand{\eqsetpartition}[1]{\ensuremath{\mathord{=}{#1}}\textsc{-Set Partition}\xspace}
\newcommand{\leqsetpartition}[1]{\ensuremath{\mathord{\leq}{#1}}\textsc{-Set Partition}\xspace}
\newcommand{\leqsetpartitionsets}[1]{\ensuremath{\mathord{\leq}{#1}}\textsc{-Set Partition (\#Sets)}\xspace}

\newcommand{\eqdsetpartition}{\eqsetpartition{d}}
\newcommand{\leqdsetpartition}{\leqsetpartition{d}}
\newcommand{\leqdsetpartitionsets}{\leqsetpartitionsets{d}}

\newcommand{\eqsetpackingsets}[1]{\ensuremath{{\mathord{=}}{#1}}\textsc{-Set Packing (\#Sets)}\xspace}
\newcommand{\leqsetpackingsets}[1]{\ensuremath{{\mathord{\leq}}{#1}\textsc{-Set Packing (\#Sets)}}\xspace}

\newcommand{\eqdsetpackingsets}{\eqsetpackingsets{d}}
\newcommand{\leqdsetpackingsets}{\leqsetpackingsets{d}}

\newcommand{\leqsetpackingunion}[1]{\ensuremath{\mathord{\leq}{#1}}\textsc{-Set Packing (Union)}\xspace}
\newcommand{\leqdsetpackingunion}{\leqsetpackingunion{d}}

\newcommand{\classymb}{\ast}

\newcommand{\eqsetcoverclass}{\eqsetcover{\classymb}}
\newcommand{\leqsetcoverclass}{\leqsetcover{\classymb}}
\newcommand{\eqsetpartitionclass}{\eqsetpartition{\classymb}}
\newcommand{\leqsetpartitionclass}{\leqsetpartition{\classymb}}
\newcommand{\leqsetpartitionsetsclass}{\leqsetpartitionsets{\classymb}}
\newcommand{\eqsetpackingsetsclass}{\eqsetpackingsets{\classymb}}
\newcommand{\leqsetpackingsetsclass}{\leqsetpackingsets{\classymb}}
\newcommand{\leqsetpackingunionclass}{\leqsetpackingunion{\classymb}}

\newcommand{\packing}[1]{\ensuremath{#1\textsc{-Packing}}\xspace}

\newcommand{\partition}[1]{\ensuremath{#1\textsc{-Partition}}\xspace}
\newcommand{\prepacking}{c\textsc{-Precolored-}\ensuremath{\triangle\textsc{-Packing}}\xspace}

\newcommand{\trieq}{$\triangle$\textsf{-eq}\xspace}
\newcommand{\trieqS}{$\triangle$\textsf{-eq}$(S)$\xspace}

\newcommand{\signature}[1]{\ensuremath{#1\textrm{-signature}}}
\newcommand{\signatures}[1]{\ensuremath{#1\textrm{-signatures}}}

\newcommand{\joininstance}[1]{\ensuremath{(#1)\textrm{-join}}}

\newcommand{\defproblem}[3]{
  \vspace{1mm}
  \noindent\fbox{
  \begin{minipage}{0.96\textwidth}
  \begin{tabular*}{\textwidth}{@{\extracolsep{\fill}}lr} #1 \\ \end{tabular*}
  {\bf{Input:}} #2  \\
  {\bf{Question:}} #3
  \end{minipage}
  }
  \vspace{1mm}
}

\newcounter{problemcounter}
\newcommand{\Oh}{\bigO}

\begin{document}
\title{\vspace*{-0.4cm}Fundamental Problems on Bounded-Treewidth Graphs:\\The Real Source of Hardness}

\date{}
\author{
Bar\i\c{s} Can Esmer\thanks{CISPA Helmholtz Center for Information Security. The first author is also affiliated with the Saarbr\"ucken Graduate School of Computer Science, Saarland Informatics Campus, Germany. Research of the third author is supported by the European Research Council (ERC) consolidator grant No.~725978 SYSTEMATICGRAPH.} \and 
Jacob Focke$^\ast$ 
\and 
D\'{a}niel Marx$^\ast$ \and
Paweł Rzążewski\thanks{Warsaw University of Technology, Faculty of Mathematics and Information Science and University of Warsaw, Institute of Informatics, \texttt{pawel.rzazewski@pw.edu.pl}.}
}

\maketitle

\begin{abstract}
It is known for many algorithmic problems that if a tree decomposition of width $t$ is given in the input, then the problem can be solved with exponential dependence on $t$. 
A line of research initiated by Lokshtanov, Marx, and Saurabh [SODA 2011] produced lower bounds showing that in many cases known algorithms already achieve the best possible exponential dependence on $t$, assuming the Strong Exponential-Time Hypothesis (SETH). The main message of this paper is showing that the same lower bounds can already be obtained in a much more restricted setting: informally, a graph consisting of a block of $t$ vertices connected to components of constant size already has the same hardness as a general tree decomposition of width $t$.

Formally, a {\em $(\sigma,\delta)$-hub} is a set $Q$ of vertices such that every component of $Q$ has size at most $\sigma$ and is adjacent to at most $\delta$ vertices of $Q$.
We explore if the known tight lower bounds parameterized by the width of the given tree decomposition remain valid if we parameterize by the size of the given hub.
\begin{itemize}
\item For every $\epsilon>0$, there are $\sigma,\delta>0$ such that \textsc{Independent Set} (equivalently \textsc{Vertex Cover}) cannot be solved in time $(2-\epsilon)^p\cdot n$, even if a $(\sigma, \delta)$-hub of size $p$ is given in the input, assuming the SETH.  This matches the earlier tight lower bounds parameterized by width of the tree decomposition. Similar tight bounds are obtained for \textsc{Odd Cycle Transversal}, \textsc{Max Cut}, \textsc{$q$-Coloring}, and edge/vertex deletions versions of \textsc{$q$-Coloring}. 
\item For every $\epsilon>0$, there are $\sigma,\delta>0$ such that \partition{\triangle} cannot be solved in time $(2-\epsilon)^p\cdot n$, even if a $(\sigma, \delta)$-hub of size $p$ is given in the input, assuming the Set Cover Conjecture (SCC). In fact, we prove that this statement is {\em equivalent} to the SCC, thus it is unlikely that this could be proved assuming the SETH.

\item For \textsc{Dominating Set}, we can prove a non-tight lower bound ruling out $(2-\epsilon)^p\cdot n^{\bigO(1)}$ algorithms, assuming \emph{either} the SETH or the SCC, but this does not match the $3^p\cdot n^{\bigO(1)}$ upper bound.
\end{itemize}
Thus our results reveal that, for many problems, the research on lower bounds on the dependence on tree width was never really about tree decompositions, but the real source of hardness comes from a much \emph{simpler} structure.

Additionally, we study if the same lower bounds can be obtained if $\sigma$ and $\delta$ are fixed universal constants (not depending on $\epsilon$). We show that lower bounds of this form are possible for \textsc{Max Cut} and the edge-deletion version of \textsc{$q$-Coloring}, under the Max 3-Sat Hypothesis (M3SH). However, no such lower bounds are possible for \textsc{Independent Set}, \textsc{Odd Cycle Transversal}, and the vertex-deletion version of \textsc{$q$-Coloring}: better than brute force algorithms are possible for every fixed $(\sigma,\delta)$.
\end{abstract}

\newpage
\pagestyle{empty}
\tableofcontents
\newpage
\setcounter{page}{1}
\pagestyle{plain}

\section{Introduction}\label{sec:intro}
Starting with the work of Lokshtanov, Marx, and Saurabh \cite{DBLP:journals/talg/LokshtanovMS18}, there is a line of research devoted to giving lower bounds on how the running time of parameterized algorithms can depend on treewidth (or more precisely, on the width of a given tree decomposition) \cite{DBLP:journals/siamcomp/OkrasaR21,DBLP:conf/esa/OkrasaPR20,DBLP:conf/stacs/EgriMR18,DBLP:conf/soda/CurticapeanLN18,
  DBLP:conf/iwpec/BorradaileL16,DBLP:journals/dam/KatsikarelisLP19, DBLP:conf/icalp/MarxSS21, DBLP:conf/soda/CurticapeanM16,DBLP:conf/soda/FockeMR22,DBLP:conf/iwpec/MarxSS22,DBLP:conf/soda/FockeMINSSW23}. The goal of this paper is to revisit the fundamental results from \cite{DBLP:journals/talg/LokshtanovMS18} to point out that previous work could have considered a \emph{simpler} parameter to obtain \emph{stronger} lower bounds in a more uniform way. Thus, in a sense, this line of research was never really about treewidth; a fact that future work should take into account.

Suppose we want to solve some algorithmic problem on a graph $G$ given with a tree decomposition of width $t$. For many NP-hard problems, standard dynamic program techniques or meta theorems such as Courcelle's Theorem~\cite{DBLP:journals/iandc/Courcelle90} show that the problem can be solved in time $f(t)\cdot n^{\bigO(1)}$ for some computable function $f$ \cite[Chapter 7]{DBLP:books/sp/CyganFKLMPPS15}. In many cases, the running time is actually $c^t\cdot n^{\bigO(1)}$ for some constant $c>1$, where it is an obvious goal to make the constant as small as possible. A line of work started by Lokshtanov, Marx, and Saurabh \cite{DBLP:journals/talg/LokshtanovMS18} provides tight conditional lower bounds for many problems with known $c^t\cdot n^{\bigO(1)}$-time algorithms. The lower bounds are based on the Strong Exponential-Time Hypothesis, formulated by Impagliazzo, Paturi, and Zane \cite{DBLP:journals/jcss/ImpagliazzoP01,DBLP:journals/jcss/ImpagliazzoPZ01}.

\begin{seth}[\textbf{SETH}]
There is no $\epsilon>0$ such that for every $k$, every $n$-variable instance of $k$-\textsc{Sat} can be solved in time $(2-\epsilon)^n \cdot n^{\bigO(1)}$.
\end{seth}
The goal of these results is to provide evidence that the base $c$ of the exponent in the best known $c^t\cdot n^{\bigO(1)}$-time algorithm is optimal: if a $(c-\epsilon)^t\cdot n^{\bigO(1)}$-time algorithm exists for any $\epsilon>0$, then SETH fails. The following theorem summarizes the basic results obtained by Lokshtanov, Marx, and Saurabh~\cite{DBLP:journals/talg/LokshtanovMS18}.

\begin{thm}[\cite{DBLP:journals/talg/LokshtanovMS18}]\label{LMS}
If there exists an $\epsilon > 0$ such that \vspace{-1mm}
\begin{enumerate}\setlength\itemsep{-.9mm}
\item {\sc Independent Set} can be solved in time $(2-\epsilon)^t\cdot n^{\bigO(1)}$, or 
\item {\sc Dominating Set} can be solved in time $(3-\epsilon)^t\cdot n^{\bigO(1)}$, or 
\item {\sc Max Cut} can be solved in time $(2-\epsilon)^t\cdot n^{\bigO(1)}$, or 
\item {\sc Odd Cycle Transversal} can be solved in time $(3-\epsilon)^t\cdot n^{\bigO(1)}$, or 
\item $q$-{\sc Coloring} can be solved in time $(q-\epsilon)^t\cdot n^{\bigO(1)}$ for some $q \geq 3$, or
\item \textsc{Triangle Partition} can be solved in time $(2-\epsilon)^t\cdot n^{\bigO(1)}$,
\end{enumerate}
on input an $n$-vertex graph $G$ together with a tree decomposition of width at most $t$, then the SETH fails.
\end{thm}

Already in \cite{DBLP:journals/talg/LokshtanovMS18} it is pointed out that many of the lower bounds remain true even in the more restricted setting where the input is not a tree decomposition, but a path decomposition. This raises the following natural questions:
\begin{itemize}
\item	How much further can we restrict the input and still obtain the same lower bounds?
\item	What is the \emph{real} structural source of hardness in these results?
\end{itemize}
In this paper, we show that many of these lower bounds remain true in a much more restricted setting where a block of $p$ vertices is connected to constant-size components. Additionally, we demonstrate that our results are very close to being best possible, as further restrictions of the structure of the graphs allow better algorithms.

We say that a set $Q$ of vertices is a \emph{\core
  {\sigma}{\delta}} of $G$ if every component of $G-Q$ has at most $\sigma$ vertices and each such component is adjacent to at most $\delta$ vertices of $Q$ in $G$\footnote{
  This notion is related to \textit{component order connectivity}, which is the size of the smallest set $Q$ of vertices such that deleting $Q$ leaves components of size not larger than some predefined constant $\sigma$ \cite{MR3308560,MR4589432,MR4415120,MR1676481,MR2704605,DBLP:conf/iwpec/KumarL16,DBLP:journals/algorithmica/Bang-JensenEGWY22,DBLP:journals/arscom/GrossKSSS13,TSUR2023112,CRESPELLE2023100556,DBLP:conf/soda/LokshtanovMRSZ21}. Our definition has the additional constraint on the neighborhood size of each component. As we often refer to the set $Q$ itself (not only its smallest possible size) and we want to make the constants $\sigma$, $\delta$ explicit, the terminology \core{\sigma}{\delta} is grammatically more convenient than trying to express the same using component order connectivity.}.
Our goal is to prove lower bounds parameterized by the size of 
a \core
  {\sigma}{\delta} given in the input, where $\sigma$ and $\delta$ are treated as constants. 
  One can observe that a \core{\sigma}{\delta} of size $p$ in $G$ can be easily turned into a tree decomposition of width less than $p+\sigma$, hence the treewidth  of $G$ is at most $p+\sigma$. Therefore, any lower bound parameterized by the size $p$ of \coreword\ immediately implies a lower bound parameterized by the width of the given tree decomposition. We systematically go through the list of problems investigated by Lokshtanov, Marx, and Saurabh \cite{DBLP:journals/talg/LokshtanovMS18}, to see if the same lower bound can be obtained with this parameterization. Our results show that, in most cases, the results remain valid under parameterization by hub size. However, new insights, techniques and arguments are needed; in particular, we require different complexity assumptions for some of the statements.

  \paragraph{Coloring problems and relatives.}
  Let us first consider the \coloring{q} problem: given a graph $G$, the task is to find a coloring of the vertices of $G$ with $q$ colors such that adjacent vertices receive different colors. Given a \core{\sigma}{\delta} $Q$ of size $p$, we can try all possible $q$-colorings on $Q$ and check if they can be extented to every component of $G-Q$. Assuming $\sigma$ and $\delta$ are constants, this leads to a $q^p\cdot n^{\bigO(1)}$ algorithm. Our first result shows that this is essentially best possible, assuming the SETH; note that this result immediately implies Theorem~\ref{LMS}(5).
  
\begin{thm}\label{thm:coloringcombined} Let $q\ge 3$ be an integer.
  \begin{enumerate}
\item For every $\sigma, \delta \geq 1$, \coloring{q} on $n$-vertex graphs can be solved in time $q^{p} \cdot n^{\bigO(1)}$ if a \core{\sigma}{\delta} of size $p$ is given in the input.

\item  For every $\epsilon > 0$, there exist integers $\sigma,\delta\ge 1$ such that if there is an algorithm solving in time $(q  - \epsilon)^{\cpar} \cdot n^{\bigO(1)}$ every $n$-vertex instance of \coloring{q} given with a \core{\sigma}{\delta} of size at most $\cpar$, then the SETH fails.
\end{enumerate}
\end{thm}

 The      \coloringED{q} problem is an edge-deletion optimization version of \coloring{q}: given a graph G, the task is to find a set $X$ of edges of minimum size such that $G\setminus X$ has a $q$-coloring. We show that $q^p\cdot n^{\bigO(1)}$ running time is essentially optimal for this problem as well.
\begin{thm}\label{thm:coloringEDcombined} Let $q\ge 2$ be an integer.
  \begin{enumerate}
\item For every $\sigma, \delta \geq 1$, \coloringED{q} on $n$-vertex graphs can be solved in time $q^{p} \cdot n^{\bigO(1)}$ if a \core{\sigma}{\delta} of size $p$ is given in the input.

\item  For every $\epsilon > 0$, there exist integers $\sigma,\delta\ge 1$ such that if there is an algorithm solving in time $(q  - \epsilon)^{\cpar} \cdot n^{\bigO(1)}$ every $n$-vertex instance of \coloringED{q} given with a \core{\sigma}{\delta} of size at most $\cpar$, then the SETH fails.
\end{enumerate}
\end{thm}
      For $q\ge 3$, the lower bound of Theorem~\ref{thm:coloringcombined} for \coloring{q} immediately implies the same lower bound for the more general problem \coloringED{q}. Observe that for $q=2$, the \coloringED{q} problem is equivalent to the \maxcut problem: deleting the minimum number of edges to make the graph bipartite is equivalent to finding a bipartition with the maximum number of edges going between the two classes. Thus the lower bound for \maxcut is needed to complete the proof of Theorem~\ref{thm:coloringEDcombined}.

      Let us consider now the vertex-deletion version \coloringVD{q}, where given a graph G, the task is to find a set $X$ of vertices of minimum size such that $G-X$ has a $q$-coloring (equivalently, we want to find a partial $q$-coloring on the maximum number of vertices). For this problem, a brute force approach would need to consider $(q+1)^p$ possibilities on a \core{\sigma}{\delta} of size $p$: each vertex can receive either one of the $q$ colors, or be deleted.
\begin{thm}\label{thm:coloringVDcombined} Let $q\ge 1$ be an integer.
  \begin{enumerate}
\item For every $\sigma, \delta \geq 1$, \coloringVD{q} on $n$-vertex graphs can be solved in time $(q+1)^{p} \cdot n^{\bigO(1)}$ if a \core{\sigma}{\delta} of size $p$ is given in the input.

\item  For every $\epsilon > 0$, there exist integers $\sigma,\delta\ge 1$ such that if there is an algorithm solving in time $(q  +1 - \epsilon)^{\cpar} \cdot n^{\bigO(1)}$ every $n$-vertex instance of \coloringVD{q} given with a \core{\sigma}{\delta} of size at most $\cpar$, then the SETH fails.
\end{enumerate}
\end{thm}
Observe that \textsc{Vertex Cover} is equivalent to \coloringVD{1} and \textsc{Odd Cycle Transversal} is equivalent to \coloringVD{2}. Furthermore, \textsc{Independent Set} and \textsc{Vertex Cover} have the same time complexity (due to the well-known fact that minimum size of a vertex cover plus the maximum size of an independent set is always equal to the number of vertices). Thus the definition of \coloringVD{q} gives a convenient unified formulation that includes these fundamental problems.
    
\paragraph{Packing problems.} Given a graph $G$, the \textsc{Triangle Partition} (denoted by \partition{\triangle} for short) problem asks for a partition of the vertex set into triangles. \textsc{Triangle Packing} (denoted by \packing{\triangle}) is the more general problem where the task is to find a maximum-size collection of vertex-disjoint triangles. Given a tree decomposition of width $t$, Theorem~\ref{LMS}(6) shows that $2^t\cdot n^{\bigO(1)}$ is essentially the best possible running time. It seems that the same lower bound holds when parameterizing by the size of a \coreword, but the source of hardness is somehow different. Instead of assuming the SETH, we prove this lower bound under the \emph{Set Cover Conjecture (SCC)} \cite{cyganProblemsHardCNFSAT2016,DBLP:books/sp/CyganFKLMPPS15}. In the  \textsc{$d$-Set Cover} problem, we are given a universe $U$ of size $n$ and a collection $\calF$ of subsets of $U$, each with size at most $d$. The task is to find a minimum-size collection of sets whose union covers the universe. 
\begin{scc}[\textbf{SCC}]
	For all $\varepsilon > 0$, there exists $d \geq 1$ such that there is no algorithm that solves every $\leqdsetcover$ instance $(U,\mathcal{F})$ in time $(2 - \varepsilon)^{n} \cdot n^{\bigO(1)}$ where $n = \abs{U}$.
\end{scc}

We actually show that the lower bounds for \partition{\triangle}/\packing{\triangle} are \emph{equivalent} to the SCC.
\begin{thm}\label{thm:mainpacking}
  The following three statements are equivalent:
  \begin{itemize}
  \item The SCC is true.
    \item For every $\epsilon>0$, there are $\sigma,\delta>0$ such that \partition{\triangle} on an $n$-vertex graph cannot be solved in time $(2-\epsilon)^p\cdot n^{\bigO(1)}$, even if the input contains a \core{\sigma}{\delta} of size $p$.
    \item For every $\epsilon>0$, there are $\sigma,\delta>0$ such that \packing{\triangle} on an $n$-vertex graph cannot be solved in time $(2-\epsilon)^p\cdot n^{\bigO(1)}$, even if the input contains a \core{\sigma}{\delta} of size $p$.
  \end{itemize}
\end{thm}
Ideally, one would like to prove lower bounds under the more established conjecture: the SETH. However, Theorem~\ref{thm:mainpacking} shows that it is no shortcoming of our technique that we prove the lower bound based on the SCC instead. If we proved statement 2 or 3 under the SETH, then this would prove that the SETH implies the SCC, resolving a longstanding open question. 

\paragraph{Dominating Set.}
Given an $n$-vertex graph with a tree decomposition of width $t$, a minimum dominating set can be computed in time $3^t\cdot n^{\bigO(1)}$ using an algorithm based on fast subset convolution \cite{DBLP:conf/birthday/Rooij20,DBLP:conf/esa/RooijBR09}. By Theorem~\ref{LMS}\,(2), this running time cannot be improved to $(3-\epsilon)^t\cdot n^{\bigO(1)}$ for any $\epsilon>0$, assuming the SETH. Can we get a $(3-\epsilon)^p\cdot n^{\bigO(1)}$ algorithm if a \coreword of size $p$ is given in the input? We currently have no answer to this question. In fact, we do not even have a good guess whether or not such algorithms should be possible. What we do have are two very simple weaker results that rule out $(2-\epsilon)^p\cdot n^{\bigO(1)}$ algorithms, with one proof based on the SETH and the other proof based on the SCC.

\begin{thm}\label{thm:domsetcombined}
 For every $\epsilon>0$, there are $\sigma,\delta>0$ such that \textsc{Dominating Set} on an $n$-vertex graph with a \core{\sigma}{\delta} of size $p$ given in the input cannot be solved in time $(2-\epsilon)^p\cdot n^{\bigO(1)}$, unless \emph{both} the SETH and the SCC fail.
\end{thm}  

Theorem~\ref{thm:domsetcombined} suggests that, if there is no $(3-\epsilon)^p\cdot n^{\bigO(1)}$ time algorithm for \textsc{Dominating Set}, then perhaps the matching lower bound needs a complexity assumption that is stronger than both the SETH and the SCC.

\paragraph{\boldmath Universal constants for $\sigma$ and $\delta$?}
The lower bounds in Theorems~\ref{thm:coloringcombined}--\ref{thm:domsetcombined} are stated in a somewhat technical form: ``for every $\epsilon>0$, there are $\sigma$ and $\delta$ such that$\ldots$''. The statements would be simpler and more intuitive if they were formulated in a setting where $\sigma$ and $\delta$ are universal constants, say, 100. Can we prove statements that show, for example, that there is  no $(2-\epsilon)^p\cdot n^{\bigO(1)}$ algorithm, where $p$ is the size of a \core{100}{100} given in the input?

The answer to this question is complicated. For the vertex-deletion problem \coloringVD{q} (which includes \textsc{Vertex Cover} and \textsc{Odd Cycle Transversal}) there are actually better than brute force algorithms for fixed constant values of $\sigma$ and $\delta$.
\begin{thm} \label{thm:nonlist-vd-better-alg}
For every $q \geq 3$ and $\sigma,\delta>0$, there exists $\epsilon>0$ with the following property: 
 every instance $(G,L)$ of \coloringVD{q} with $n$ vertices, given with a \core{\sigma}{\delta} of size $\cpar$, can be solved in time $(q+1-\epsilon)^\cpar \cdot n^{\bigO(1)}$.
\end{thm}
Thus Theorem~\ref{thm:nonlist-vd-better-alg} explains why the formulation of Theorem~\ref{thm:coloringVDcombined} needs to quantify over $\sigma$ and $\delta$, and cannot be stated for a fixed pair $(\sigma,\delta)$.

On the other hand, for the edge-deletion problem \coloringED{q} (which includes \textsc{Max Cut}), we can prove stronger lower bounds where $\sigma$ and $\delta$ are universal constants. However, we need a complexity assumption different from the SETH. 

An instance of \MaxSat is a CNF formula $\phi$ with at most three literals in each clause.
We ask for the minimum number of clauses that need to be deleted in order to obtain a satisfiable formula.
Equivalently, we look for a valuation of the variables which violates the minimum number of clauses.
Clearly, an instance of \MaxSat with $n$ variables can be solved in time $2^n \cdot n^{\bigO(1)}$ by exhaustive search.
It is a notorious problem whether this running time can be significantly improved, i.e., whether there exists an $\epsilon>0$
such that every $n$-variable instance of \MaxSat can be solved in time $(2-\epsilon)^n$.

\begin{maxsathyp}[\textbf{M3SH}]
There is no $\epsilon>0$ such that every $n$-variable instance of \MaxSat can be solved in time $(2-\epsilon)^n \cdot n^{\bigO(1)}$.
\end{maxsathyp}

Under this assumption, we can prove a lower bound where $\delta=6$ and $\sigma$ is a constant (depending only on $q$). 

\begin{restatable}{thm}{LBcolEDMSH}\label{thm:LBcolEDMSHintro}
	For every $q\ge 2$ there is an integer $\sigma$ such that the following holds.
	For every $\eps>0$, no algorithm solves every $n$-vertex instance of \coloringED{q} that is given with a \core{\sigma}{6} of size $p$, in time $(q-\eps)^p\cdot n^{\bigO(1)}$, unless the \msh fails.
\end{restatable}
  
For the case $q=2$ we even show a slight improvement over \cref{thm:LBcolEDMSH} --- in this case it suffices to consider instances with a constant $\sigma$ and $\delta=4$.

For \partition{\triangle}, we do not know if the lower bound of Theorem~\ref{thm:mainpacking}, ruling out $(2-\epsilon)^p\cdot n^{O(1)}$ running time under the SCC, remains valid for some fixed universal $\sigma$ and $\delta$ independent of $\epsilon$. Note that the proof of Theorem~\ref{thm:mainpacking} provides a reduction from \partition{\triangle} to \textsc{$d$-Set Packing} for some $d$. It is known that \textsc{$d$-Set Packing} over a universe of size $n$ can be solved in time $(2-\epsilon)^n\cdot (n+m)^{\bigO(1)}$ with some $\epsilon>0$ depending on $d$ \cite{nederlof:LIPIcs.ESA.2016.69,DBLP:conf/iwpec/Koivisto09,bjorklund:LIPIcs.STACS.2010.2447}. However, our reduction from \partition{\triangle} to \textsc{$d$-Set Packing} chooses $d$ in a way that it cannot be used to reduce the case of a fixed $\sigma$ and $\delta$ to a \textsc{$d$-Set Packing} problem with fixed $d$. It seems that we would need to understand if certain generalizations of 
\textsc{$d$-Set Packing} can also be solved in time $(2-\epsilon)^n\cdot (n+m)^{\bigO(1)}$ for fixed $d$. The simplest such problem would be the generalization of \textsc{$d$-Set Packing} where the sets in the input are partitioned into pairs and the solution is allowed to use at most one set from each pair.

\paragraph{Discussion.}

  Given the amount of attention to algorithms on tree decompositions and the number of nontrivial techniques that were developed to achieve the best known algorithms, it is a natural question to ask if these algorithms are optimal. Even though understanding treewidth is a very natural motivation for this line of research, the actual results turned out to be less related to treewidth than one would assume initially: the lower bounds remain valid even under more restricted conditions. Already the first paper on this topic \cite{DBLP:journals/talg/LokshtanovMS18} states the lower bounds in a stronger form, as parameterized by pathwidth or by feedback vertex set number (both of which are bounded below by treewidth). Some other results considered parameters such as the size of a set $Q$ where every component of $G-Q$ is a path \cite{DBLP:conf/ciac/JaffkeJ17} or has bounded treewidth \cite{DBLP:conf/iwpec/HegerfeldK22}. However, our results show that none of these lower bounds got to the fundamental reason why known algorithms on bounded-treewidth graphs cannot be improved: Theorems~\ref{thm:coloringcombined}--\ref{thm:mainpacking} highlight that these algorithms are best possible already if we consider a much more restricted problem setting where constant-sized gadgets are attached to a set of \coreword\ vertices. Moreover, Theorems~\ref{thm:coloringcombined} and \ref{thm:coloringVDcombined} are likely to be best possible: as we have seen, for coloring and its vertex-deletion generalizations, $\sigma$ and $\delta$ cannot be made a constant independent from $\eps$ (Theorem~\ref{thm:nonlist-vd-better-alg}). Therefore, one additional conceptual message of our results is understanding where the hardness of solving problems on bounded-treewidth graphs really stems from, by reaching the arguably most restricted setting in which the lower bounds hold.%
  
  The success of Theorems~\ref{thm:coloringcombined}--\ref{thm:coloringVDcombined} (for coloring problems and relatives) suggests that possibly all the treewidth optimality results could be revisited and the same methodology could be used to strengthen to parameterization by \coreword\ size. But the story is more complicated than that. For example, for \packing{\triangle}, the ground truth appears to be that the lower bound parameterized by width of the tree decomposition can be strengthened to a lower bound parameterized by \coreword\ size. However, proving the lower bound parameterized by \coreword\ size requires a different proof technique, and we can do it only by assuming the SCC --- and for all we know this is an assumption orthogonal to the SETH. In fact, we showed that the lower bound for \packing{\triangle} is equivalent to the SCC, making it unlikely that a simple proof based on the SETH exists. For \textsc{Dominating Set}, we currently do not know how to obtain tight bounds, highlighting that it is far from granted that all results parameterized by width of tree decomposition can be easily turned into lower bounds parameterized by \coreword\ size.

  Another important aspect of our results is the delicate way they have to be formulated, with the values of $\sigma$ and $\delta$ depending on $\epsilon$. Theorem~\ref{thm:LBcolEDMSHintro} shows that in some cases it is possible to prove a stronger bound where $\sigma$ and $\delta$ are universal constants, but this comes at a cost of choosing a different complexity assumption (M3SH). Thus, there is a tradeoff between the choice of the complexity assumption and the strength of the lower bound. In general, it seems that the choice of complexity assumption can play a crucial role in these kind of lower bounds parameterized by \coreword\ size. This has to be contrasted with the case of parameterization by the width of the tree decomposition, where the known lower bounds are obtained from the SETH (or its counting version).

  It would be natural to try to obtain lower bounds parameterized by \coreword\ size for other algorithmic problems as well. The lower bounds obtained in this paper for various fundamental problems can serve as a starting point for such further results. Concerning the problems studied in this paper, we leave two main open questions:
  \begin{itemize}
  \item For \textsc{Dominating Set}, can we improve the lower bound of Theorem~\ref{thm:domsetcombined} to rule out $(3-\epsilon)^p\cdot n^{\bigO(1)}$ algorithms, under some reasonable assumption? Or is there perhaps an algorithm beating this bound?
    
    \item For \partition{\triangle}/\packing{\triangle}, can we improve the lower bound of Theorem~\ref{thm:mainpacking} such that $\sigma$ and $\delta$ are universal constants? Or is it true perhaps that for every fixed $\sigma$ and $\delta$, there is an algorithm solving these problems in time $(2-\epsilon)^p\cdot n^{\bigO(1)}$ for some $\epsilon>0$? 
\end{itemize}

\section{Technical Overview}\label{sec:overview}

In this section, we overview some of the most important technical ideas in our results. 

\subsection{$q$-Coloring}

The algorithmic statement in \cref{thm:coloringcombined} is easily obtained via a simple branching procedure. For the hardness part, we use a lower bound of Lampis~\cite{DBLP:journals/siamdm/Lampis20} for constraint satisfaction problems (CSP) as a starting point: for any $\epsilon>0$ and integer $d$, there is an integer $r$ such that there is no algorithm solving
CSP on $n$ variables of domain size $d$ and $r$-ary constraints in time $(d-\epsilon)^n$.
Therefore, to prove \cref{thm:coloringcombined}\,(2), we give a reduction that, given an $n$-variable CSP instance where the variables are over $[q]$ and the arity of constraints is some constant $r$, creates an instance of \coloring{q} having a \coreword{} of size roughly $n$.

First, we introduce a set of $n$ main vertices in the \coreword, representing the variables of the CSP instance. We would like to represent each $r$-ary constraint with a gadget that is attached to a set $S$ of $r$ vertices. We will first allow our gadgets to use lists that specify to which colors certain vertices are allowed to be mapped. In a second step we then remove these lists. 

A bit more formally, we say that an \emph{$r$-ary $q$-gadget} is a graph $J$ together with a list assignment $L\from V(J)\to 2^{[q]}$ and $r$ distinguished vertices $\boldx=(z_1,\ldots, z_r)$ from $J$. The vertices $z_1,\ldots,z_r$ are called \emph{portals}. A \emph{list coloring} of $(J,L)$ is an assignment $\phi\from V(J)\to [q]$ that respects the lists $L$, i.e., with $\phi(v) \in L(v)$ for all $v \in V(J)$.

A construction by Jaffke and Jansen~\cite{DBLP:conf/ciac/JaffkeJ17} gives a gadget that enforces that a set of vertices forbids one prescribed coloring. We use this statement to construct a gadget extends precisely the set of colorings that are allowed according to some relation. 

\begin{restatable}{prop}{realizecoloring}\label{prop:realize-coloring}
	Let $q \geq 3$ and $r \geq 1$ be integers, and let $R \subseteq [q]^r$ be a relation.
	Then there exists an $r$-ary $q$-gadget $\calF=(F,L,(z_1,\ldots,z_r))$ such that
	\begin{myitemize}
		\item the list of every vertex is contained in $[q]$,
		\item for each $i \in r$, it holds that $L(z_i)=[q]$,
		\item $\{z_1,\ldots,z_r\}$ is an independent set,
		\item for any $\psi : \{z_1,\ldots,z_r\} \to [q]$, coloring vertices $z_1,\ldots,z_r$ according to $\psi$ can be extended to a list coloring of $(F,L)$ if and only if $(\psi(z_1),\ldots,\psi(z_r)) \in R$.
	\end{myitemize}
\end{restatable}

Then, by introducing one gadget per constraint and attaching it to the vertices of the \coreword, from the $q^n$ possible behaviors of the \coreword{} vertices, only those can be extended to the gadgets that correspond to a satisfying assignment of the CSP instance.
Note that gadgets are allowed to use lists and they model the relational constraints using list colorings. So the final step to obtain a reduction to \coloring{q} is to remove these lists. This can be done using a standard construction, where a central clique of size $q$ is used to model the $q$ colors, and a vertex $v$ of the graph is adjacent to the $i$th vertex of the clique, whenever $i\notin L(v)$.

\subsection{Vertex Deletion to $q$-Coloring.}

Similarly to \coloring{q}, the algorithmic statement in \cref{thm:coloringVDcombined} is easily obtained via a simple branching procedure. However, for \coloringVD{q}, we need to consider $q+1$ possibilities at each vertex: assigning to it one of the $q$ colors, or deleting it. This leads to the running time $(q+1)^p\cdot n^{\bigO(1)}$. The hardness proof is also similar, but this time we have to give a reduction that, given an $n$-variable CSP instance where the variables are over $[q+1]$ and the arity of constraints is some constant $r$, creates an instance of \coloringVD{q} having a \coreword\ of size roughly $n$.  Intuitively, we are using deletion as the $(q+1)$-st color: the $(q+1)^n$ possibilities for these vertices in the \coloringVD{q} problem (coloring with $q$ colors $+$ deletion) correspond to the $(q+1)^n$ possible assignments of the CSP instance. To enforce this interpretation, we attach to these vertices small gadgets representing each constraint. We attach a large number of copies of each such gadget, which means that it makes no sense for an optimum solution to delete vertices from these gadgets and hence deletions occur only in the \coreword. This means that we can treat the vertices of the gadgets as ``undeletable''.

We would like to use again the construction from \cref{prop:realize-coloring} to create gadgets that enforce that a set of vertices has one of the prescribed colorings/deletions. A gadget can force the deletion of a vertex if its neighbors are colored using all $q$ colors. However, there is a fundamental limitation of this technique: deleting a vertex is always better than coloring it. That is, a gadget cannot really force a set $S$ of vertices to the color ``red'': from the viewpoint of the gadget, deleting some of them and coloring the rest red is equally good. In other words, it is not true that every relation $R\subseteq [q+1]^r$ can be represented by a gadget that allows only these combinations of $q$ colors + deletion on a set $S$ of $r$ vertices.

To get around this limitation, we use a grouping technique to have control over how many vertices are deleted. Let us divide the $n$ variables into $M=n/b$ blocks $B_1$, $\dots$, $B_M$ of size $b$ each. Let us guess the number $f_i$ of variables in $B_i$ that receive the value $q+1$ in a hypothetical solution; that is, we expect $f_i$ deletions in block $B_i$ of central vertices. Instead of just attaching a gadget to a set $S$ of at most $r$ vertices, now each gadget is attached to the at most $r$ blocks containing $S$. Besides ensuring a combination of values on $S$ that satisfies the constraint, the gadget also ensures that each block $B_i$ it is attached to has at least the guessed number $f_i$ of deletions. This way, if we have a solution with exactly $\sum_{i=1}^Mf_i$ deletions, then we know that it has exactly $f_i$ deletions in the $i$-th block. Therefore, if a gadget forces the deletion of $f_i$ vertices of $B_i$ and forces a coloring on the remaining vertices of $B_i$, then we know that that block has exactly this behavior in the solution.

\subsection{Edge Deletion to $q$-Coloring}
Let us turn our attention to the edge-deletion version now. Similarly to the vertex-deletion version, the algorithmic results are simple, thus we discuss only the hardness proofs here.
As starting point for all our reductions, we use a CSP problem with domain size $q$ that naturally generalizes \maxsat: the task is to find an assignment of variables that satisfies the maximum number of constraints. For $q=2$, the hardness of this problem follows from the SETH and the \msh. For $q\ge 3$, we prove a new tight lower bound based on \msh. 

For $q\ge 3$, the lower bound of \cref{thm:coloringEDcombined} (hardness of \coloringED{q} under the SETH) already follows from our result for \emph{finding} a $q$-coloring \emph{without deletions} (\cref{thm:coloringcombined}).
So, in order to complete the proof of \cref{thm:coloringEDcombined}, we give a reduction from the CSP problem with $q=2$ to \coloringED{2} (i.e, \maxcut), which shows hardness under SETH. As the gadgets of \cref{prop:realize-coloring} work only for $q\ge 3$, we need to design new gadgets using only 2 colors for this case.

The same reduction can be used to establish the lower bound from \cref{thm:LBcolEDMSHintro} (hardness of \coloringED{q} under the \msh) in the $q=2$ case.
For the $q\ge 3$ case, we present a reduction from the CSP problem with domain size $q$ to \coloringED{q}. Here we can once again use the gadgets from \cref{prop:realize-coloring}.

In all cases, as the gadgets we design may use lists, we establish respective lower bounds for the list coloring problem on the way. In a second step, we then show how to remove the lists.

\paragraph*{Max\,CSP --- Hardness under the \msh}
For some positive integers $d$ and $r$, we define \MaxCSP{$d$}{$r$}: Given $v$ variables over a $q$-element domain and a set of $n$ relational constraints of arity $3$, the task is to find an assignment of the variables such that the maximum number of constraints are satisfied. The problem can be solved in time $q^v\cdot n^{\bigO(1)}$ by brute force. For $q=2$, the problem is clearly a generalization of \maxsat, hence the M3SH immediately implies that there is no $(q-\eps)^v\cdot n^{\bigO(1)}$ algorithm for any $\eps>0$. We show that the M3SH actually implies this for any $q\ge 2$. This might also be a helpful tool for future work.

\begin{restatable}{thm}{MTSHtoMaxCSP}\label{thm:M3SHtoMaxCSP}
	For $d \geq 2$ and any $r \geq 3$, there is no algorithm solving every $n$-variable instance of \MaxCSP{$d$}{$r$} in time $(d-\epsilon)^n\cdot n^{\bigO(1)}$ for $\epsilon>0$, unless the M3SH fails.
\end{restatable}

In order to show \cref{thm:M3SHtoMaxCSP},
if $q$ is a power of 2, then a simple grouping argument works: for example, if $q=2^4=16$, then each variable of the CSP instance can represent 4 variables of the \maxsat instance, and hence it is clear that a $(q-\eps)^v$ algorithm would imply a $(2-\eps)^{4v}$ algorithm for a \maxsat instance with $4v$ variables.

      The argument is not that simple if $q$ is not a power of 2, say $q=15$. Then a variable of that CSP instance cannot represent all 16 possibilities of 4 variables of the \maxsat instance, and using it to represent only 3 variables would be wasteful. We cannot use the usual trick of grouping the CSP variables such that each group together represents a group of \maxsat variables: then each constraint representing a clause would need to involve not only 3 variables, but 3 blocks of variables, making $\delta$ larger than 3. Instead, for each block of 4 variables of the \maxsat instance, we randomly choose 15 out of the 16 possible assignments, and use a single variable of the CSP instance to represent these possibilities. An optimum solution of a $4v$-variable \maxsat instance
      ``survives'' this random selection with probability $(15/16)^v$. Thus a  $(15-\eps)^v\cdot n^{\bigO(1)}$ time algorithm for the CSP problem would give a randomized $(16/15)^v\cdot (15-\eps)^v\cdot n^{\bigO(1)}= (16-\eps')^v\cdot n^{\bigO(1)}= (2-\eps'')^{4v}\cdot n^{\bigO(1)}$ time algorithm for \maxsat, violating (a randomized version of) the M3SH. Furthermore, we show in Section~\ref{sec:maxCSPMSH} that the argument can be derandomized using the logarithmic integrality gap between integer and fractional covers in hypergraphs.
      
 \paragraph*{Realizing Relations using Lists}
   Recall that an $r$-ary $q$-gadget is a graph with lists in $[q]$ and $r$ specified portal vertices. For our hardness proofs, we reduce from \MaxCSP{$q$}{$r$}, and we use gadgets to ``model'' the relations in $[q]^r$. We say that an $r$-ary $q$-gadget \emph{realizes} a relation $R\in [q]^r$ if there is an integer $k$ such that (1) for each $\boldd\in R$, if the portals are colored according to $\boldd$, then it requires precisely $k$ edge deletions to extend this to a full list coloring of the gadget, and (2) extending a state that is \emph{not} in $R$ requires strictly more than $k$ edge deletions.
  We say that such a gadget $1$-realizes $R$, if for each state outside of $R$ it takes precisely $k+1$ edge deletions to extend this state. So, this is a stronger notion in the sense that now the violation cost is the same for all tuples outside of $R$. Moreover, with a $1$-realizer in hand, by identifying copies of this gadget with the same portal vertices one can freely adjust the precise violation cost --- this works as long as the portals form an independent set and therefore no multiedges are introduced in the copying process.
  
  For our treatment of the case $q\ge3$, we again use \cref{prop:realize-coloring} to show that arbitrary relations over a domain of size $q$ can be realized. As \cref{prop:realize-coloring} is for the decision problem without deletions, it does not help for the case $q=2$, i.e., for \maxcut/\coloring{2}.
  In this case, we need a different approach to show that every relation over a domain of size $2$ can be realized. For \coloring{2}, a single edge is essentially a ``Not Equals''-gadget as the endpoints have to take different colors or otherwise the edge needs to be deleted. Starting from this, we show how to model OR-relations of any arity. With these building blocks we then obtain the following result.
 
  \begin{restatable}{thm}{maxcutRelations}\label{lem:NEQtoRelstrong}\label{cor:K2relations}
  	For each $r\ge 1$, and $R\subseteq [2]^r$, there is an $r$-ary $2$-gadget that $1$-realizes $R$.
  \end{restatable}

\paragraph*{Removing the Lists} 
Note that gadgets may use lists and therefore, on the way, we first obtain the following lower bounds for the respective list coloring problems.

\begin{restatable}{thm}{LBLHomEDKqSETH}\label{thm:LBLHomEDKqSETH}
	For every $q\ge 2$ and $\eps>0$, there are integers $\sigma$ and $\delta$ such that if an algorithm solves in time $(q-\eps)^p\cdot n^{\bigO(1)}$ every $n$-vertex instance of \listcoloringED{q} that is given with a \core{\sigma}{\delta} of size $p$, then the SETH fails.
\end{restatable}

\begin{restatable}{thm}{LBLHomEDKqMSH}\label{thm:LBLHomEDKqMSH}
	For every $q\ge 2$, there is a constant $\sigma_q$ such that, for every $\eps>0$, if an algorithm solves in time $(q-\eps)^p\cdot n^{\bigO(1)}$ every $n$-vertex instance of \listcoloringED{q} that is given with a \core{\sigma_q}{3} of size $p$, then the \msh fails.
\end{restatable}

In a second step, we show how to remove the lists by adding some additional object of size roughly $q$ (a central vertex or a $q$-clique for $q=2$ or $q\ge3$, respectively). This addition is then considered to be part of the \coreword{}, thereby increasing the size of the \coreword{} by some constant. However, this modification means that for the other gadgets the number of neighbors in the \coreword{} increases slightly. This is irrelevant for the SETH-based lower bound, but it leads to a slight increase in the universal constant $\delta$ that we obtain for our \msh-based lower bounds for the coloring problems without lists.

\subsection{Covering, Packing, and Partitioning}\label{sec:overview_SCCequiv}
      Theorem~\ref{thm:mainpacking} gives lower bounds for \partition{\triangle} and \packing{\triangle} based on the Set Cover Conjecture. This hypothesis was formulated in terms of the \textsc{$d$-Set Cover} problem. For our purposes, it is convenient to consider slightly different covering/partitioning problems. To facilitate our reductions and as a tool for future reductions of this type, we establish equivalences between eight different covering type problems. 
      Before we make this more formal in \cref{thm:SCCequivalent}, let us briefly introduce the corresponding problems.
      
      First, we use \eqdsetcover and \leqdsetcover to distinguish between the problem for which the sets have size exactly $d$ or at most $d$, respectively. For \partition{\triangle}, it is more natural to start a reduction from the partitioning problems \eqdsetpartition or \leqdsetpartition, in which the task is to find pairwise disjoint sets that cover the universe.
      The \leqdsetpartition problem can be considered as a decision problem. However, we can also consider the corresponding optimization problem in which the task is to minimize the number of selected sets, and we use \leqdsetpartitionsets to denote this problem. Further variants are the optimization problems \eqdsetpackingsets and \leqdsetpackingsets, in which we need to select the maximum number of pairwise disjoint sets. For \textsc{$\le d$-Set Packing}, an equally natural goal is to maximize the total size of the selected sets (for \textsc{$=d$-Set Packing}, this is of course equivalent to maximizing the number of selected sets). So we use \leqdsetpackingunion to denote the packing problem in which the union/total size of the selected sets is maximized.

      Given the large number of variants of \textsc{$d$-Set Cover}, one may wonder how they are related to each other. In particular, does the SCC imply lower bounds for these variants? There are obvious reductions between some of these problems (e.g., from \textsc{$=d$-Set Cover} to \textsc{$\le d$-Set Cover}) and there are also reductions that are not so straightforward. We fully clarify this question by showing that choosing \emph{any} of these problems in the definition of the SCC leads to an equivalent statement. Thus in our proofs to follow we can choose whichever form is most convenient for us. Knowing this equivalence could prove useful for future work as well.
      
      \begin{restatable}{thm}{SCCequivalent}\label{thm:SCCequivalent}
        Suppose that for one of the problems below, it is true that for every $\epsilon>0$, there is an integer $d$ such that the problem cannot be solved in time $(2-\epsilon)^n\cdot n^{\bigO(1)}$, where $n$ is the size of the universe. Then this holds for all the other problems as well. In particular, any of these statements is equivalent to the SCC.
        \begin{myenumerate}
        \item \eqdsetcover
        \item \eqdsetpartition
        \item \eqdsetpackingsets
        \item \leqdsetcover
        \item \leqdsetpartition
        \item \leqdsetpartitionsets
        \item \leqdsetpackingsets
        \item \leqdsetpackingunion
        \end{myenumerate}
      \end{restatable}        

	To make the statements about relationships between the problems from the list in \cref{thm:SCCequivalent} more concise, it will be convenient to introduce some shorthand notation.
	Let $\mathbb{A} = \{A_d\}_{d \geq 1}$ and $\mathbb{B} =
	\{B_d\}_{d \geq 1}$ be two families of problems where $A_d$ and
	$B_d$ belong to the list in \cref{thm:SCCequivalent}. To shorten notation, we speak of an \emph{$n$-element instance} if the universe $U$ of an instance has size $n$.
	We say that $\mathbb{A}$ is $\hard$ if the following lower bound holds
	\begin{quote}
		For each $\eps > 0$ there is some $d\geq 1$ such that no algorithm solves $A_d$ on all $n$-element instances in time $(2 - \eps)^{n} \cdot n^{\Oh(1)}$.
	\end{quote}

	Using this language, the SCC states that $\{\leqdsetcover\}_{d\ge1}$ is $\hard$.
	To establish \cref{thm:SCCequivalent}, we show reductions, stating that if $\calA$ is $\hard$ then $\calB$ is $\hard$ as well. Spelled out this means:
	
	\begin{quote}
		Suppose for each $\eps > 0$ there is some $d\geq 1$ such that no algorithm solves $A_d$ on all $n$-element instances in time $(2 - \eps)^{n} \cdot n^{\Oh(1)}$.\\[.5em]
		Then, for each $\eps > 0$ there is some $d'\geq 1$ such that no algorithm solves $B_{d'}$ on all $n$-element instances in time $(2 - \eps)^{n} \cdot n^{\Oh(1)}$. 
	\end{quote}
	
	This shows that this is really a relationship between two classes of problems, and not necessarily a relationship between $A_d$ and $B_d$ for the same value $d$.
	To make this distinction explicit, we write \eqsetcoverclass if we refer to the class of problems $\{\eqdsetcover\}_{d\ge 1}$. We use analogous notation for the other problems on the list.
	For example, a simple observation is that if \eqsetcoverclass is $\hard$ then so is \leqsetcoverclass as the latter is a generalization of the former.
	The reductions we use to prove \cref{thm:SCCequivalent} are illustrated in \cref{fig:reductionoverview}.
	
	\begin{figure}[t]
		\centering
		\begin{tikzpicture}[ipe import]
  \filldraw[-latex, fill=white]
    (374.2074, 743.4706)
     -- (374.2074, 695.4706);
  \filldraw[-latex, fill=white]
    (68.3115, 695.4706)
     -- (68.3115, 743.4706);
  \node[ipe node]
     at (80.284, 790.974) {Obs.~\ref{obs:eqPCtoleqPC2}};
  \node[ipe node]
     at (9.347, 790.974) {Lem.~\ref{lem:leqPCtoeqPC}};
  \node[ipe node]
     at (203.586, 742.536) {Obs.~\ref{obs:eqPCtoleqPC1}};
  \node[ipe node]
     at (192.731, 605.936) {Obs.~\ref{obs:eqSCtoleqSC}};
  \node[ipe node]
     at (103.63, 643.05) {Obs.~\ref{lem:eqPTtoeqSC}};
  \node[ipe node]
     at (289.401, 643.05) {Lem.~\ref{lem:leqSCtoleqPTs}};
  \node[ipe node]
     at (18.249, 713.02) {Obs.~\ref{obs:eqPTtoeqPC}};
  \node[ipe node]
     at (386.417, 713.02) {Lem.~\ref{lem:leqPCtoleqPTs}};
  \node[ipe node]
     at (257.199, 690.154) {Lem.~\ref{lem:leqPTstoPT}};
  \node[ipe node]
     at (115.705, 690.154) {Lem.~\ref{lem:leqPTtoeqPT}};
  \draw[shift={(59.955, 814.833)}, xscale=0.0087, yscale=-1.0165, rotate=30.4186, -latex]
    (0, 0)
     -- (80, 0);
  \draw[shift={(74.993, 815.083)}, xscale=0.0076, yscale=-1.021, rotate=-149.4649, latex-]
    (0, 0)
     -- (-80, 0);
  \begin{scope}[shift={(-176.704, -237.513)}, xscale=1.0742, yscale=1.3403]
    \node[ipe node, anchor=center]
       at (512.871, 683.019) {$\leqsetpartitionsetsclass$};
    \draw
      (579.674, 676.566) rectangle (446.024, 689.825);
  \end{scope}
  \begin{scope}[shift={(-266.284, -338.87)}, xscale=1.1339, yscale=1.3686]
    \node[ipe node, anchor=center]
       at (482.803, 683.863) {$\leqsetcoverclass$};
    \draw
      (519.581, 677.902) rectangle (446.024, 689.825);
  \end{scope}
  \begin{scope}[shift={(-174.489, -80.5331)}, xscale=1.0798, yscale=1.2264]
    \node[ipe node, anchor=center]
       at (508.128, 683.019) {$\leqsetpackingunionclass$};
    \draw
      (570.232, 676.212) rectangle (446.024, 689.825);
  \end{scope}
  \begin{scope}[shift={(-416.42, -268.319)}, xscale=1.1466, yscale=1.2642]
    \node[ipe node, anchor=center]
       at (482.802, 684.541) {$\eqsetcoverclass$};
    \draw
      (519.581, 679.257) rectangle (446.024, 689.825);
  \end{scope}
  \begin{scope}[shift={(-477.403, -233.319)}, xscale=1.1088, yscale=1.3302]
    \node[ipe node, anchor=center]
       at (491.1, 684.541) {$\eqsetpartitionclass$};
    \draw
      (536.177, 679.257) rectangle (446.024, 689.825);
  \end{scope}
  \begin{scope}[shift={(-475.846, -100.425)}, xscale=1.0671, yscale=1.2555]
    \node[ipe node, anchor=center]
       at (508.836, 683.018) {$\eqsetpackingsetsclass$};
    \draw
      (571.649, 676.212) rectangle (446.024, 689.825);
  \end{scope}
  \begin{scope}[shift={(-476.201, -3.9739)}, xscale=1.0677, yscale=1.2264]
    \node[ipe node, anchor=center]
       at (508.837, 683.017) {$\leqsetpackingsetsclass$};
    \draw
      (571.649, 676.212) rectangle (446.024, 689.825);
  \end{scope}
  \draw[shift={(67.177, 661.094)}, xscale=1.1046, yscale=1.4751, rotate=16.2086, -latex]
    (0, 0)
     -- (51.1158, -51.116);
  \draw[shift={(279.783, 612.866)}, xscale=0.8504, yscale=0.975, -latex]
    (0, 0)
     -- (111.4197, 49.413);
  \begin{scope}[shift={(-334.286, -207.862)}, xscale=1.1078, yscale=1.2943]
    \node[ipe node, anchor=center]
       at (491.101, 683.863) {$\leqsetpartitionclass$};
    \draw
      (536.177, 677.902) rectangle (446.024, 689.825);
  \end{scope}
  \filldraw[shift={(146.227, 677.27)}, xscale=-0.194, yscale=-1, -latex, fill=white]
    (0, 0)
     -- (80, 0);
  \filldraw[shift={(288.823, 677.27)}, xscale=-0.194, yscale=-1, -latex, fill=white]
    (0, 0)
     -- (80, 0);
  \filldraw[shift={(186.224, 596.263)}, xscale=0.5794, yscale=1.2811, -latex, fill=white]
    (0, 0)
     -- (80, 0);
  \draw[shift={(298.474, 757.013)}, xscale=4.8642, yscale=0, latex-]
    (0, 0)
     -- (-32, 32);
\end{tikzpicture}
		\caption{An overview of the proof of \cref{thm:SCCequivalent}. An arrow from $\calA$ to $\calB$ indicates an implication stating that if $A$ is $\hard$ then $\calB$ is $\hard$ as well.}
		\label{fig:reductionoverview}
	\end{figure}

	\subsection{Triangle Partition and Triangle Packing}

	\begin{figure}
	    \centering
	    \includegraphics{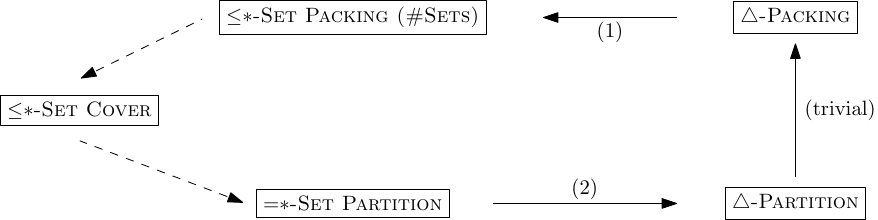}
	\caption{An overview of the reductions in the proof of
		\cref{thm:mainpacking}. The two dashed arrows refer to $\hardness$ reductions
		from \cref{thm:SCCequivalent}. To establish these two connections, note that we actually utilize all reductions shown in \cref{fig:reductionoverview}, except for
		the simple \cref{obs:eqSCtoleqSC,obs:eqPTtoeqPC}. The arrows
		annotated with (1) and (2) denote the reductions proved in
	\cref{sec:clique-packing-partitioning}.}
	    \label{fig:packing_partitioning_lb}
	\end{figure}
	
	Now let us discuss the proof of \cref{thm:mainpacking} that can be found in 	\cref{sec:clique-packing-partitioning}. 
	The proof consists of two main steps: (1) a reduction from \eqsetpartitionclass to $\partition{\triangle}$, and (2) a reduction from $\packing{\triangle}$ to \leqsetpackingsetsclass (see \cref{fig:packing_partitioning_lb}).
	Recall that by \cref{thm:SCCequivalent}, assuming the SCC, all of \eqsetpartitionclass, \leqsetpackingsetsclass, and \leqsetcoverclass are $\hard$. Finally, $\partition{\triangle}$ trivially reduces to $\packing{\triangle}$, so indeed,  the statements in  \cref{thm:mainpacking} are equivalent.

	\paragraph{Reducing \textsc{Set Partition} to $\partition{\triangle}$.}
	We start with step (1), i.e., reducing an instance $(U,\calF)$ of \eqdsetpartition to an equivalent instance $G$ of $\partition{\triangle}$.
	With a simple technical trick we can ensure that $d$ is divisible by 3.
	
The main building block used in the reduction is the so-called \trieq gadget.
For fixed $d$, it is a graph with $d$ designated vertices called \emph{portals}. The gadget essentially has exactly two triangle packings that cover all non-portal vertices:
\begin{itemize}
\item one that also covers \emph{all portals} (i.e., is actually a triangle partition), and
\item one that covers \emph{no portal}.
\end{itemize}
Now the construction of $G$ is simple: we introduce the set $Q$ containing one vertex for each element of $U$, and for each set $S \in \mathcal{F}$ we introduce a copy of the \trieq gadget whose portals represent elements of $S$ and are identified with corresponding vertices from $Q$.
It is straightforward to verify that there is $\mathcal{F}' \subseteq \mathcal{F}$ that partitions $U$ if and only if $G$ has a triangle partition: the sets from $\mathcal{F}'$ correspond to \trieq gadgets whose non-portal vertices are covered in the first way. Note that $Q$ is a \core{\sigma}{d} of $G$, where $\sigma$ is the number of vertices of the \trieq gadget, i.e., is a constant that depends only on $d$.
	
\paragraph{Reducing $\packing{\triangle}$ to \textsc{Set Packing}.}
Now let us consider a graph $G$ given with a \core{\sigma}{\delta} $Q$ of size $p$, and an integer $t$. We will show that a hypothetical fast algorithm for \leqdsetpackingsets can be used to determine whether $G$ has a triangle packing of size at least $t$.

For simplicity of exposition, assume that $G$ has no triangles contained in $Q$; dealing with such triangles is not difficult but would complicate the notation. We say that a component $C$ of $G - Q$ is \emph{active} in some triangle packing $\Pi$ if there is a triangle in $\Pi$ that intersects both $C$ and $Q$. Note that for any triangle packing there are at most $p$ active components.

We would like to guess components that are active for some (unknown) solution $\Pi$. However, this results in too many branches.
We deal with it by employing color-coding and reducing the problem to its auxiliary \emph{precolored variant}.
Suppose for a moment that we are given a coloring $\psi$ of components of $G-Q$ into $p/c$ colors, where $c$ is a large constant, with a promise that at most $c$ components in each color are active in $\Pi$.

For a color $i \in [p/c]$, let $\calC_i$ denote the set of components of $G-Q$ colored $i$ by $\psi$. The \emph{contribution} of the color $i$ to $\Pi$ is the number of triangles that intersect vertices in components from $\calC_i$.
Note that the size of $\Pi$ is the sum of contributions of all color (since we assumed that there are no triangles contained in $Q$).
What can be said about the contribution of $i$?
Certainly picking a maximum triangle packing in the graph consisting only of components from $\calC_i$ is a lower bound. Let $X_i$ denote the number of triangles in such a triangle packing and note that $X_i$ can be computed in polynomial time as each component of $G-Q$ is of constant size.
Moreover, for each active component $C \in \calC_i$, there are at most $\sigma$ triangles that intersect both $C$ and $Q$ (as each of them has to use a distinct vertex from $C$). As, by the promise on $\psi$, there are at most $c$ active components in $\calC_i$, we observe that the contribution of $i$ is at most $X_i + c\sigma$.
We exhaustively guess the contribution of each color by guessing the offset $q_i$ against $X_i$; it gives a constant number of options per color.
We reject guesses where the total contribution of all colors, i.e., the number of all triangles packed, is less than $t$.

For each color $i$, we enumerate all sets $S \subseteq Q$ that are candidates for these vertices of $Q$ that form triangles with vertices from components of $\calC_i$; call such sets \emph{$i$-valid}.
An $i$-valid set $S$ must satisfy the following two conditions.
First, the size of $S$ is at most $2c\sigma$, as there are at most $c\sigma$ vertices in active components from $\calC_i$ and each such vertex belongs to a triangle with at most two vertices from $Q$. Second, there exists a triangle packing $\Pi_S$ in the graph induced by $S$ together with components of $\calC_i$ such that
\begin{itemize}
\item at most $c$ elements from $\calC_i$ are active in $\Pi_S$ (this follows from the promise on $\psi$), and
\item the number of triangles in $\Pi_S$ is at least $X_i + q_i$ (by our guess of $q_i$).
\end{itemize}
It is not difficult to verify that $i$-valid sets can be enumerated in polynomial time, where the degree of the polynomial depends on $c$ and $\sigma$.

Now we are ready to construct an instance $(U, \calF, p/c)$ of \leqdsetpackingsets.
The universe $U$ is $Q \cup \{a_i ~|~ i\in [p/c]\}$, i.e., it consists of the hub of $G$ and one extra vertex per color.
For each $i$-valid set $S$, we include in $\calF$ the set $S \cup \{a_i\}$.
Again, one can verify that $\calF$ contains $p/c$ pairwise disjoint sets if and only if $G$ has a packing of $t$ triangles that agree both with $\psi$ and with the guessed values of $q_i$'s.

By adjusting $c$, we can ensure that the whole algorithm works in time $(2-\epsilon')^p \cdot |V(G)|^{\bigO(1)}$, for some $\epsilon' >0$, provided that we have a fast algorithm for \leqdsetpackingsets.

The only thing left is to argue how we obtain the coloring $\psi$ satisfying the promise. Here we use \emph{splitters} introduced by Naor, Schulman, and Srinivasan~\cite{NaorSS95}.
Informally, a splitter is a family of colorings of a ``large set'' $\mathcal{X}$, such that for each ``small subset'' $\mathcal{Y} \subseteq \mathcal{X}$ there is a coloring that splits $\mathcal{Y}$ evenly.
In our setting, the ``large set'' $\mathcal{X}$ is the set of all components of $G-Q$ and the ``small subset'' $\mathcal{Y}$ is the set of all active components with respect to some fixed (but unknown) solution; recall that there are at most $p$ such active components.
Since our colorings use $p/c$ colors, we are sure that there is some $\psi$ for which at most $\frac{p}{p/c}=c$ components in each color are active.
Calling the result of Naor, Schulman, and Srinivasan~\cite{NaorSS95}, we can find a small splitter $\Psi$, and then just exhaustively try every coloring $\psi \in \Psi$.
Again, carefully adjusting the constants, we can ensure that the overall running time is  $(2-\epsilon)^p \cdot |V(G)|^{\bigO(1)}$, for some $\epsilon >0$.

\section{Preliminaries}\label{sec:prelims} \label{sec:problem-definitions}
For an integer $k$, by $[k]$ we denote $\{1,\ldots,k\}$. For a set $X$, by $2^X$ we denote the family of all subsets of $X$.

\paragraph{Graph theory.}
Let $G$ be a graph. By $V(G)$ and $E(G)$ we denote, respectively, the vertex set and the edge set of $G$.
Let $X \subseteq V(G)$, by $G[X]$ we denote the graph induced by $X$. By $G-X$ we denote the graph obtained form $G$ by removing all vertices in $X$ along with incident edges, i.e., $G[V(G) \setminus X]$. For a set $X \subseteq E(G)$, by $G \setminus X$ we denote the graph obtained by removing all edges in $X$, i.e., $(V(G),E(G) \setminus X)$.
A vertex is \emph{isolated} if its neighborhood is empty.

\paragraph{\boldmath Treewidth and \core{\sigma}{\delta}s.}
Consider a graph with a \core{\sigma}{\delta} $Q$ of size $p$. Introducing a bag that contains $Q$ that is the center of a star whose leaves are $Q\cup C_i$ for each connected component $C_i$ of $G-Q$. Then this is a tree decomposition of $G$ of width at most $p+\sigma-1$. We state this observation formally.
\begin{obs}\label{obs:coretotw}
	For some $\sigma, \delta\ge 1$, let $G$ be a graph given with a \core{\sigma}{\delta} of size $p$. One can obtain a tree decomposition of width less than $p+\sigma$ in time polynomial in the size of $G$.
\end{obs}

\paragraph{Variants of (list) colorings.}
In \cref{sec:intro}, we defined the \coloring{q} problem as well as its vertex and edge deletion variant.
We also use a generalization of vertex colorings that include lists.
Formally, the \listcoloring{q} problem takes as input a graph $G$ together with a \emph{list} function $L : V(G) \to  2^{[q]}$. The task is then to compute a $q$-coloring $\phi\from V(G)\to [q]$ that respects lists $L$, i.e., with $\phi(v) \in L(v)$ for all $v \in V(G)$.
We say that such an assignment $\phi$ is a \emph{proper list coloring} of $(G,L)$.

We also use the corresponding vertex and edge deletion variant.
In the \listcoloringVD{q} (resp. \listcoloringED{q}) problem we ask for a smallest set $X$ of \emph{vertices} (resp. \emph{edges}) such that $G-X$ (resp. $G \setminus X$) admits a proper list coloring that respects the lists $L$.

Note that \listcoloringVD{q} and \listcoloringED{q} are optimization problems. Sometimes it will be convenient to consider their corresponding decisions versions, when we are additionally given an integer $k$ and we ask whether the instance graph can be modified into a yes-instance of by removing at most $k$ vertices/edges. In general, we will show algorithms for the optimization version, and lower bounds for the decision version.

\paragraph{Gadgets.}
Let $J$ be a graph together with a list assignment $L\from V(J)\to 2^{[q]}$ and $r$ distinguished vertices $\boldx=(x_1,\ldots, x_r)$ from $J$. Then we refer to the tuple $\calJ=(J, L, \boldx)$ as an \emph{$r$-ary $q$-gadget}.
We might not specify $r$ nor $q$ in case they are clear from the context.
The vertices $x_1,\ldots,x_r$ are called \emph{portals}.

\paragraph{Definitions for Set Covering, Partitioning, and Packing Problems.}
For reference and to be clear about the distinctions between the different problems that appear in \cref{sec:equivHypo}, we now list all their definitions.

\defproblem{\eqdsetcover}
{A set $U$ of elements, a set family $\calF \subseteq 2^{U}$ where each set has size exactly $d$ such that $U = \bigcup \calF$, an integer $t$.}
{Is there a collection of at most $t$ sets from $\calF$ whose union is $U$?}

\defproblem{\leqdsetcover}
{A set $U$ of elements, a set family $\calF \subseteq 2^{U}$ where each set has size at most $d$ such that $U = \bigcup \calF$, an integer $t$.}
{Is there a collection of at most $t$ sets from $\calF$ whose union is $U$?}

\defproblem{\eqdsetpartition}
{A set $U$ of elements, a set family $\calF \subseteq 2^{U}$ where each set has size exactly $d$ such that $U = \bigcup \calF$.}
{Is there a collection of ($\abs{U}/d$) sets from $\calF$ that is a partition of $U$?}

\defproblem{\leqdsetpartition}
{A set $U$ of elements, a set family $\calF \subseteq 2^{U}$ where each set has size at most $d$ such that $U = \bigcup \calF$.}
{Is there a collection of sets from $\calF$ that is a partition of $U$?}

\defproblem{\leqdsetpartitionsets}
{A set $U$ of elements, a set family $\calF \subseteq 2^{U}$ where each set has size at most $d$ such that $U = \bigcup \calF$, an integer $t$.}
{Is there a collection of at most $t$ sets from $\calF$ that is a partition of $U$?}

\defproblem{\eqdsetpackingsets}
{A set $U$ of elements, a set family $\calF \subseteq 2^{U}$ where each set has size exactly $d$, an integer $t$}
{Is there a collection of at least $t$ sets from $\calF$ that are pairwise disjoint?}

\defproblem{\leqdsetpackingsets}
{A set $U$ of elements, a set family $\calF \subseteq 2^{U}$ where each set has size at most $d$, an integer $t$.}
{Is there a collection of at least $t$ sets from $\calF$ that are pairwise disjoint?}

\defproblem{\leqdsetpackingunion}
{A set $U$ of elements, a set family $\calF \subseteq 2^{U}$ where each set has size at most $d$, an integer $t$.}
{Is there a collection of sets from $\calF$ that are pairwise disjoint and whose union has size at least $t$?}

\section{\boldmath  \coloring{q}}\label{sec:qColSETH}

The goal of this section is to show \cref{thm:coloringcombined}. Its algorithmic part is very simple:
Exhaustively enumerate all possible $q$-colorings of the \coreword\ $Q$; this results in $q^p$ branches.
Then, for each component $C$ of $G-Q$, we check whether the coloring of the neighbors in $C$ can be extended to a proper coloring of $C$. Note that this can be performed by brute-force, as the number of vertices in $C$ is at most $\sigma$, i.e., a constant. As the number of such components $C$ is at most $n$, the overall running time is $q^p \cdot n^{\bigO(1)}$.

The hardness part of \cref{thm:coloringcombined} is restated in the following theorem.
\begin{thm}\label{thm:LBqColSETH}
	For every $q\ge 3$ and $\eps>0$, there are integers $\sigma$ and $\delta$ such that no algorithm solves every $n$-vertex instance of \coloring{q}, given with a \core{\sigma}{\delta} of size $p$, in time $(q-\eps)^p\cdot n^{\bigO(1)}$, unless the SETH fails.
\end{thm}
The proof uses the following constraint satisfaction problem as a starting point. 
For integers $d$ and $r$, an instance of $(d,r)$-\textsc{CSP} is a pair $(\calV,\calC)$, where $\calV$ is the set of \emph{variables} and $\calC$ is the set of \emph{constraints}.
Each constraint $C$ involves a sequence of variables, whose length, called the \emph{arity} of $C$, is at most $r$.
The constraint $C$ enforces some relation $R_C \subseteq [d]^{\ar(C)}$ on the variables involved.
We ask for a mapping $\calV \to [d]$ which \emph{satisfies} every constraint, i.e., the sequence of images of variables involved in $C \in \calC$
belongs to $R_C$.

Lampis showed the following lower bound for  $(d,r)$-\textsc{CSP}.

\begin{thm}[Lampis~\cite{DBLP:journals/siamdm/Lampis20}]\label{thm:csp}
For every $d \geq 2$ and $\epsilon$ there exists $r$, such that there is no algorithm solving every $N$-variable instance of $(d,r)$-\textsc{CSP} in time $(d-\epsilon)^N \cdot N^{\bigO(1)}$, unless the SETH fails.
\end{thm}

Before we proceed to the proof of \cref{thm:LBqColSETH}, let us recall following lemma by Jaffke and Jansen~\cite[Lemma 14]{DBLP:conf/ciac/JaffkeJ17} which will be a crucial tool used to build gadgets (defined in \cref{sec:prelims}).

\begin{lem}[Jaffke, Jansen~\cite{DBLP:conf/ciac/JaffkeJ17}]\label{lem:jaffke}
Let $q \geq 3$ and $r \geq 1$ be integers,
For any $\boldc \in [q]^r$ there exists an $r$-ary gadget $\calF_{\boldc}=(F_{\boldc},L,(x_1,x_2,\ldots,x_r))$, where lists $L$ are contained in $[q]$, such that for every $\boldd = (d_1,\ldots,d_r) \in [q]^r$ there exists a proper list coloring $\phi$ of $F_{\boldc}$ in which for all $i$ it holds that $\phi(x_i) \neq d_i $ if and only if $\boldd \neq \boldc$.
\end{lem}

The following statement is a consequence of \cref{lem:jaffke}.

\realizecoloring*
\begin{proof}
We start the construction with introducing an independent set $Z = \{z_1,\ldots,z_r\}$ and we set $L(z_i):=[q]$ for every $i \in [r]$.

Let $\bar R := [q]^r \setminus R$ and consider $\bar \boldc \in \bar R$.
We call \cref{lem:jaffke} to obtain a gadget $\calF_{\bar \boldc}=(F_{\bar \boldc},L,(x_1\ldots,x_r))$.
For each $i \in [r]$, we make $x_i$ adjacent to $z_i$.
We repeat this for every $\bar \boldc \in \bar R$. This concludes the construction of $\calF$.

Consider a coloring $\psi: Z \to [q]$ such that $(\psi(z_1),\ldots,\psi(z_r)) \in R$.
Consider $\bar \boldc \in \bar R$. We need to show that  $\psi$ can be extended to a coloring of $F_{\bar \boldc}$.
Note that $(\psi(z_1),\ldots,\psi(z_r)) \neq \bar \boldc$ as $\bar \boldc  \notin R$.
Thus by \cref{lem:jaffke} called for $(d_1\ldots,d_r) = (\psi(z_1),\ldots,\psi(z_r))$ we note that that there exists a proper list coloring of $F_{\bar \boldc}$ such that the color of each $x_i$ is different than $\psi(z_i)$. This is the sought-for coloring of $F_{\bar \boldc}$.

Consider a coloring $\psi: Z \to [q]$ such that $\bar \boldc := (\psi(z_1),\ldots,\psi(z_r)) \notin R$.
For contradiction, suppose that there exists a proper list coloring $\phi$ of $F$ that extends $\psi$.
Consider the gadget $\calF_{\bar \boldc}=(F_{\bar \boldc},L,(x_1,\ldots,x_r))$.
Clearly, for all $i \in [r]$, it holds that $\phi(x_i) \neq \psi(z_i)$.
However, the existence of such a coloring contradicts \cref{lem:jaffke}.
\end{proof}

Now we are ready to prove \cref{thm:LBqColSETH}.

\begin{proof}[Proof of \cref{thm:LBqColSETH}.]
Let $q \geq 3$ and $\epsilon >0$.
For contradiction, suppose that for every $\sigma,\delta$ there is an algorithm that solves every $n$-vertex instance of \coloring{q}, given with a \core{\sigma}{\delta} of size $p$, in time $(q-\eps)^p\cdot n^{\bigO(1)}$.

Let $r$ be the value given by \cref{thm:csp} for $d=q$ and $\epsilon$.
Consider an instance $\calI = (\calV,\calC)$ of $(q,r)$-\textsc{CSP} and let $n \coloneqq |\calV|$.
Note that without loss of generality we may assume that every constraint in $\calC$ has at least one satisfying assignment.
Furthermore, we may assume that there are no two constraints with exactly the same set of vertices: otherwise we can 
replace all such constraints $C_1,C_2,\ldots,C_m$ with a single one whose set of satisfying assignments is the intersection of the sets  of satisfying assignments of $C_1,C_2,\ldots,C_m$. Note that this means that $|\calC| \leq n^r$.

\paragraph{Defining an equivalent instance of \listcoloring{q}.}
As the first step, we will define an instance $(G',L)$ of \listcoloring{q}.

We start with introducing an independent set $Y = \{y(v) ~|~ v \in \calV\}$ and each vertex from $Y$ has list $[q]$.
The vertex $y(v)$ is meant to represent $v$, i.e., the coloring of $y(v)$ corresponds to the valuation of $v$.

Now consider a constraint $C \in \calC$.
Let $v_1,v_2,\ldots,v_\ell$ be the variables of $C$, where $\ell \coloneqq \ar(C) \leq r$.
Let $R_C$ be the relation enforced by $C$ on $V_C$.
We call \cref{prop:realize-coloring} for $R_C$ to obtain an $\ell$-ary gadget 
\[
\calF(C) = (F(C), L, (z_1,\ldots,z_{\ell})).
\]
The list of every vertex is contained in $[q]$, and the interface vertices are pairwise non-adjacent and have lists $[q]$.
For each $i \in [\ell]$, we identify $z_i$ with $y(v_i)$.
We repeat this for every $C \in \calC$. This completes the construction of $(G',L)$.

It is straightforward to verify that $(G',L)$ is a yes-instance of \listcoloring{q} if and only if $\calI$ is satisfiable.
Indeed, the coloring of vertices from $Y$ exactly corresponds to the valuation of variables of $\calI$.
For each $C \in \calC$, a coloring of the interfaces of the gadget $\calF(C)$ introduced for $C$ can be extended to a list coloring of the whole gadget if and only if the coloring of interfaces corresponds to an assignment of variables that satisfies $C$.

\paragraph{Defining an equivalent instance of \coloring{q}.}
The modification of $(G',L)$ into an equivalent instance of \coloring{q} is standard.
We introduce a clique on vertices $a_1,a_2,\ldots,a_q$.
For each $u \in V(G')$ and each $i \in [q]$, we add an edge $ua_i$ if and only if $i \notin L(u)$.

Again, it is straightforward to verify that this construction simulates the lists in $L$.
Indeed, if $\phi$ is a proper $q$-coloring of $G$, then by symmetry we may assume that $\phi(a_i)=i$ for every $i \in [q]$.
Thus, for any vertex $u \in V(G') \cap V(G)$ such that $i \notin L(i)$, we have that $\phi(u) \neq i$.

\paragraph{Structure of $G$.}
The graph $G'$ has $n + q + |\calC| \cdot h(q,r) = \bigO(n^r)$ vertices, where $h(q,r)$ is the size of the largest gadget given by \cref{prop:realize-coloring}. Note that $q$ is a constant and $r$ is also a constant (depending on $q$ and $\epsilon$).

Now consider the set $Q = Y \cup \{a_1,a_2,\ldots,a_q\}$, clearly $|Y| = n+q$.
Note that every component of $G-Y$ is (a subgraph of) one of the gadgets given by \cref{prop:realize-coloring},
i.e., its size is at most $h(q,r)$. Furthermore, each such component is adjacent to at most $r+q$ vertices from $Q$.
Thus $Q$ is a \core{\sigma}{\delta}, where $\sigma \coloneqq h(q,r)$ and $\delta \coloneqq r+q$ are constants depending only on $q$ and $\epsilon$.

\paragraph{Running time.}
Note that calling our hypothetical algorithm for \coloring{q} on $G'$ we can solve $\calI$ in time
\[
(q-\epsilon)^{|Q|} \cdot |V(G')|^{\bigO(1)} = (q-\epsilon)^n \cdot n^{\bigO(1)},
\]
which, by \cref{thm:csp}, contradicts the SETH.
\end{proof}

\section{\boldmath Vertex Deletion to \coloring{q}}\label{sec:vd-basic}
In this section we prove \cref{thm:coloringVDcombined}. The algorithmic part is simple and it is covered in \cref{sec:vd-basic-algo}. The main work is showing the lower bound, and this is done in form of \cref{thm:vd-coloring-intro} in \cref{sec:vd-basic-hardness}.

\subsection{\boldmath Hardness for for \coloringVD{q}}\label{sec:vd-basic-hardness}
We prove the following lower bound for \coloringVD{q}, which is identical to \cref{thm:coloringVDcombined}\,(2)

\begin{thm}
	\label{thm:vd-coloring-intro}
	For every $q\ge 1$ and $\epsilon > 0$, there exist integers $\sigma,\delta\ge 1$ such that if there is an algorithm solving in time $(q + 1 - \epsilon)^{\cpar} \cdot n^{\bigO(1)}$ every $n$-vertex instance of \coloringVD{q} given with a \core{\sigma}{\delta} of size at most $\cpar$, then the SETH fails.
\end{thm}

To that end, we start by showing that solving $(d,r)$-\textsc{CSP} for certain structured instances is still hard. Let us define what we mean by structured.

\begin{defn}
	For all $d,r,b,N \in \mathbb{N}$ and $\boldf \colon \big[\frac{N}{b}\big] \to \{0, \ldots,b\}$, an $N$-variable instance $(\calV,\calC)$ of $(d,r)$-\textsc{CSP} is called \emph{($b, \boldf$)-structured} if the following holds
	\begin{myitemize}
		\item $N$ is divisible by $b$ and $\calV$ is partitioned into $\frac{N}{b}$ \emph{blocks} $V_1,\ldots,V_{N/b}$, each of size $b$.
		\item the scope of each constraint $C \in \calC$ includes, for each $i$, either all or none of the variables from the block $V_i$. 
		
		\item There are two types of constraints, i.e., $\calC = \calC^{\textrm{sign}} \cup \calC^{\textrm{sat}}$ such that
		\begin{myitemize}
			\item $\calC^{\textrm{sign}}$ contains a constraint $C_i$ for each $i \in \big[\frac{N}{b}\big]$. The constraint 				$C_i$ makes sure that exactly $f_i$ variables from $V_i$ are set to $d$.
			\item For each $C \in \calC^{\textrm{sat}}$, if the scope of $C$ contains $V_i$, then exactly $f_i$					variables of $V_i$ are set to $d$. Furthermore, no two constraints $C,C' \in \calC^{\textrm{sat}}$ have exactly the same set of variables.
		\end{myitemize}
	\end{myitemize}
\end{defn}
Note that the last property of $\calC^{\textrm{sat}}$ implies that $|\calC^{\textrm{sat}}| \leq N^{r}$ and thus $|\calC| \leq N/b + N^r$.
The astute reader might wonder why we need $\calC^{\textrm{sign}}$, as these constraints are already implied by $\calC^{\textrm{sat}}$. However, their special structure will be exploited in our reduction.

\begin{lem}\label{lem:csp-structured}
	For all $N,d \in \mathbb{N}$ and $\epsilon > 0$, there exists $b,r \in \mathbb{N}$ such that there is no algorithm solving every ($b, \boldf$)-structured $N$-variable $(d,b \cdot r)$-\textsc{CSP} instance where $\boldf \colon \big[\frac{N}{b}\big] \to \{0, \ldots,b\}$, in time $(d-\epsilon)^{N}\cdot N^{\bigO(1)}$, unless the SETH fails.
\end{lem}

\begin{proof}
	Let $d \in \mathbb{N}$ and $\epsilon \in \mathbb{R}_{>0}$. Let $r$ be the value in the Theorem~\ref{thm:csp} for $d$ and $\epsilon$. Let $b$ be the smallest integer such that $(b+1)^{1/b} < 1 + \epsilon/d$; note that $b$ is a constant (again, depending on $d$ and $\epsilon$).

	Let $\calI = \left( \calV, \calC \right)$ be an $N$-variable instance of $(d,r)$-\textsc{CSP}.
	We assume that there is some fixed total order on the set $\calV$.
	We can further assume that the number $N$ of variables is divisible by $b$, as otherwise we can add at most $b-1$ dummy variables.
	Then we partition $\calV$ into $M= N/b$ blocks $V_1,\ldots,V_M$, each of size $b$.
	
	Let $r' = b \cdot r$; note that $r'$ depends on $d$ and $\epsilon$.
	Now we will construct a family $\boldI$ of instances of $(d,r')$-\textsc{CSP},  
	such that $\calI$ is satisfiable if and only if at least one instance in $\boldI$ is satisfiable.
	
	\paragraph{Construction of $\boldI$.}
	For an assignment $\varphi :\calV \to [d]$, its \emph{signature} is the vector $\boldf = (f_1\ldots,f_M) \in \{0,\ldots,b\}^M$, where for each $i \in [M]$ the value of $f_i$ is the number of variables from $V_i$ mapped by $\varphi$ to $d$.
	Note that the number of possible signatures is at most 
	\begin{equation}
		(b+1)^M = (b+1)^{N/b} = \left( (b+1)^{1/b} \right )^N < (1+\epsilon/d)^N. \label{eq:branches_new}
	\end{equation}
	
	Consider one such signature  $\boldf = (f_1,\ldots,f_M)$.
	We include into $\boldI$ the instance $\calI_{\boldf} = (\calV, \calC_{\boldf})$ of $(d,r')$-\textsc{CSP} defined as follows. The set of constrains $\calC_{\boldf}$ will be partitioned into two subsets, denoted by $\calC^{\textrm{sign}}_{\boldf}$ and $\calC^{\textrm{sat}}_{\boldf}$.
	
	First, for each block $V_i$, we add to $\calC_{\boldf}^{\textrm{sign}}$ the constraint $C^{i}_{\boldf}$ including all variables from $V_i$, which is satisfied if and only if exactly $f_i$ variables from $V_i$ have value $d$.
	
	Then, for each constraint $C \in \calC$, let $C_{\boldf}$ be the constraint involving all variables in all blocks intersected by $C$ such that $C_{\boldf}$ is satisfied by some assignment if an only if:
	\begin{myitemize}
		\item its projection to the variables from $C$ satisfies $C$,
		\item for each block $V_i$ involved in $C_{\boldf}$, exactly $f_i$ variables from $V_i$ are mapped to $d$.
	\end{myitemize}
	We include each such constraint $C_{\boldf}$ in $\calC_{\boldf}^{\textrm{sat}}$.
	Now, if $\calC_{\boldf}^{\textrm{sat}}$ contains several constraints $C_1,\ldots,C_p$ with exactly the same set of variables, we replace them with a single constraint whose corresponding relation is the intersection of the relations forced by $C_1,\ldots,C_p$.
	This completes the definition of $\calI_{\boldf}$.
	
	Note that each constraint from $\calC_{\boldf}$ has arity at most $r'=b \cdot r$ and thus $\calI_{\boldf}$ is an $N$-variable instance of $(d,r')$-\textsc{CSP}.
	Moreover, since $b$ divides $N$ and by the defnition of $\calC_{\boldf}$, we observe that $\calI_{\boldf}$ is ($b, \boldf$)-structured.
	Intuitively, constrains in $\calC^{\textrm{sign}}_{\boldf}$ make sure that a solution has the correct signature,
	and constraints in $\calC^{\textrm{sat}}_{\boldf}$ make sure that it satisfies $\calI$.

	\paragraph{Equivalence of instances.}
	First, observe that $\calI_{\boldf}$ is satisfiable if and only if $\calI$ is satisfied by some assignment with signature $\boldf$. Consequently, since $\boldI$ contains the instance $\calI_{\boldf}$ for every possible signature $\boldf$,
	we conclude that $\calI$ is satisfiable if and only if $\boldI$ contains a satisfiable instance.
	
	\paragraph{Running time.}
	Now let us estimate the running time.
	For each signature $\boldf$, the construction of $\calI_{\boldf}$ takes polynomial time.
	Suppose there exists an algorithm solving every ($b, \boldf$)-structured $N$-variable instance of $(d,b \cdot r)$-\textsc{CSP}, where $\boldf \colon \big[\frac{N}{b}\big] \to \{0, \ldots,b\}$, in time $(d-\epsilon)^{N}\cdot N^{\bigO(1)}$. Then, given any $N$-variable instance $\calI$ of $(d,r)$-\textsc{CSP}, we can construct $\calI_\boldf$ for all possible signatures $\boldf$, where each $\calI_\boldf$ can be solved in time $(d-\epsilon)^{N}\cdot N^{\bigO(1)}$.
	By \eqref{eq:branches_new}, the number of possible signatures is at most $(1+\epsilon/d)^N$. Consequently, $\calI$ can be solved in time 
	\[
	\left (1 + \frac{\epsilon}{d} \right )^N \cdot (d-\epsilon)^{N} \cdot N^{\bigO(1)} = \left( d - \frac{\epsilon^2}{d} \right)^N \cdot N^{\bigO(1)}= \left( d - \epsilon^{\prime} \right)^N \cdot N^{\bigO(1)}.
	\]
	where $\epsilon' \coloneqq \epsilon^{2} / d$. By \cref{thm:csp}, this is not possible unless the SETH fails.
\end{proof}

In the following, we will show lower bounds for the \emph{decision} variant of the \coloringVD{q} problem.
We assume that, along with the instance, we are given an integer $k$ and we ask whether the optimum solution (i.e., the set of vertices to delete) has size at most $k$.
Clearly, such a lower bound implies the lower bound for the optimization version,
as $k$ can be assumed to be bounded by the number of vertices of the instance.

It turns out that the characteristics of the problem differ between the cases $q \in \{1,2\}$ and $q \geq 3$:
this is because $q$-\textsc{Coloring} is polynomial-time solvable in the former case and \NP-hard in the latter one.
The proof of our lower bound is split into \cref{lem:qcolVD-1,lem:qcolVD-2}.

\begin{lem}\label{lem:qcolVD-1}
	For $q \in \{1,2\} $ and $\epsilon > 0$, there exist $\sigma, \delta > 0$ depending on $q$ and $\epsilon$
	such that there is no algorithm solving all $n$-vertex instances of \coloringVD{q} given with a
	\core{\sigma}{\delta} of size at most $\cpar$, in time $(q + 1 - \epsilon)^{\cpar} \cdot n^{\bigO(1)}$, unless the SETH fails.
\end{lem}
\begin{proof}
	Let $q \in \{1,2\}$ and $\epsilon > 0$. Let $b$ and $r$ be the values given by \cref{lem:csp-structured} for $d \coloneqq q + 1$ and $\epsilon$. 
	
	Suppose for a contradiction that for all constant $\sigma, \delta  \in \mathbb{N}$, there exists an algorithm solving all $n$-vertex instances of \coloringVD{q}, given with a \core{\sigma}{\delta} of size at most $\cpar$, in time $(q + 1 - \epsilon)^{\cpar} \cdot n^{\bigO(1)}$.
	
	Consider a ($b, \boldf$)-structured $N$-variable $(d,b \cdot r)$-\textsc{CSP} instance $\calI = (\calV,\calC)$, where $\boldf \colon \big[\frac{N}{b}\big] \to \{0, \ldots,b\}$.
	We will define an instance $(G,k)$ of \coloringVD{q} that is equivalent to $\calI$,
	and solving it with the hypothetical algorithm mentioned before,
	would allow us to solve $\calI$ in time $(d-\epsilon)^N \cdot N^{\bigO(1)}$.
	By \cref{lem:csp-structured} this contradicts the SETH.
	
	Note that without loss of generality we may assume that every constraint in $\calC$ has at least one satisfying assignment, as otherwise $\calI$ is clearly a no-instance.
	
	\paragraph{Construction of $(G,k)$.}
	We start by introducing the set $Y = \bigcup_{v \in \calV} \{y(v) \}$ of \textit{variable vertices}.
	For each $v \in \calV$, the vertex $y(v)$ represents the variable $v$.
	The intended meaning is that coloring $y(v)$ with a color $c \in [q]$ corresponds to assigning the value $c$ to $v$.
	Deleting $y(v)$ then corresponds to assigning the value $q+1$ to $v$.
	
	Consider a constraint $C \in \calC$.
	Let $V_C$ be the set of variables involved in $C$ and let $R_C$ be the relation enforced by $C$ on $V_C$.
	Let $\ell_C$ be the size of $R_C$, i.e., the number of valuations of $V_C$ that satisfy $C$; recall that $\ell_C \geq 1$.
	We will create a gadget $\gd_C$ representing $C$.
	
	We introduce a clique $\cl_C$ with vertices $\bigcup_{\boldr \in R}\{ x(\boldr)\}$,
	where $x(\boldr)$ corresponds to the assignment $\boldr$.
	
	In the case of $q = 2$, for each edge $x(\boldr)x(\boldr')$ of $\cl_C$,
	we add two additional vertices and connect them both to $x(\boldr)$ and $x(\boldr')$.
	
	Next, we add edges between the vertices of $\cl_C$ and $Y$.
	Consider $\boldr \in R$.
	For $q = 1$ and all $v \in V_C$, we connect $x(\boldr)$ and $y(v)$ if $\boldr$ sets $v$ to 2.
	For $q = 2$ and for all $v \in V_C$ we proceed as follows:
	\begin{myitemize}
		\item if $\boldr$ sets $v$ to 1, we connect $x(\boldr)$ and $y(v)$,
		\item if $\boldr$ sets $v$ to 2, then we introduce a two-edge path with endvertices $x(\boldr)$ and $y(v)$,
		\item and finally, if $\boldr$ sets $v$ to 3, we introduce both an edge and a two-edge path between $x(\boldr)$ and $y(v)$.
	\end{myitemize}
	
	This completes the description of $\gd_C$. 
	Now, for each $C \in \calC$, we introduce $N+1$ copies of the gadget.
	The $j$-th copy will be denoted by $\gd_C^{(j)}$, and we will use the same convention when refering to the vertices of that copy.
	Finally, if $q=2$, we introduce a single ``global'' vertex $h$ and make it adjacent to every vertex $\cl_C^{(i)}$, for all $C \in \calC$ and $i \in [N+1]$.
	Let $G$ be the resulting graph.
	By $Y'$ we denote the set of vertices that are \emph{not} in any $\gd_C^{(i)}$ for $C \in \calC$ and $i \in [N+1]$, i.e.,
	$Y'=Y$ if $q=1$, and $Y' = Y \cup \{h\}$ if $q=2$. 
	
	Let us define
	\[
	k' \coloneqq  \sum_{i=1}^{N / b} f_i \qquad \text{ and } \qquad k  \coloneqq  (N+1) \cdot \left( \sum_{C \in \calC} (\ell_C - 1) \right) + k'.
	\]
	
	This completes the definition of the instance $(G,k)$ of \coloringVD{q}.
	
	\paragraph{Equivalence of instances.} 
	We will now show that $\calI$ is satisfiable if and only if $G$ becomes $q$-colorable after deleting at most $k$ vertices.
	
	\medskip
	Suppose that there exists a satisfying assignment $\psi \colon \calV \to [q+1]$.
	We will find a set $X$ of size at most $k$ such that $G - X$ is $q$-colorable.
	Note that for $q=1$ this means that $X$ is a vertex cover, while for $q=2$ this means that $X$ is an odd cycle transversal.
	
	For all $v \in \calV$, the vertex $y(v)$ is included in $X$ if and only if $\psi(v) = q+1$.
	Furthermore, for $q=2$, the vertex $h$ is not included in $X$.
	Since $\calI$ is ($b$,$\boldf$)-structured, this means that $|X \cap Y'| = k'$.
	
	Consider a constraint $C \in \calC$.
	As $\psi$ satisfies $C$, we obtain that $\psi$ restricted to $V_C$ corresponds to some element of $R_C$, say $\boldr$.
	Now consider $j \in [N+1]$ and the gadget $\gd_C^{(j)}$.
	We include in $X$ all vertices from $\cl_{C}^{(j)}$, except the vertex $x(\boldr)^{(j)}$.
	In total, we obtain that
	\[
	|X \setminus Y'| = (N+1) \cdot \left( \sum_{C \in \calC} (\ell_C - 1) \right) = k - k'.
	\]
	Summing up, we have $|X| = k$.
	It is straightforward to verify that for $q=1$, the set $X$ is a vertex cover of $G$, i.e., the graph $G-X$ has a proper $1$-coloring.
	Similarly, for $q=2$, the set $X$ is an odd cycle transversal: To see this, we can define a proper 2-coloring of $G-X$ as follows. For each $v \in \calV$, if $\psi(v) \in \{1,2\}$, then we color $y(v)$ with the color $\psi(v)$.
	The vertex $h$ receives color 1, and the remaining vertices from $\cl_C^{(i)}$, for $C \in \calC$ and $i \in [N+1]$, receive color 2. The vertices introduced for edges of $\cl_C^{(i)}$ receive color 1. Finally, the endvertices of two-edge paths joining vertices from $Y$ with the vertices from $\cl_C^{(i)}$ are all colored 2, so we can color the middle vertices with color 1.
	
	Summing up, there exists a set of vertices of size $k$ whose removal from $G$ results in a graph that is $q$-colorable. Thus $(G,k)$ is a yes instance of \coloringVD{q}.
	\medskip
	
	Now suppose that there exists a set $X \subseteq V(G)$ of size at most $k$ and a proper $q$-coloring $\phi$ of $G - X$.
	We claim that for each constraint $C \in \calC$ and $j \in [N+1]$ the following two properties hold:
	\begin{myenumerate}[(P1)]
		\item At least $\ell_C - 1$ vertices from $\gd_{C}^{(j)}$ belong to $X$.
		\item If exactly $\ell_C-1$ vertices from $\gd_{C}^{(j)}$ belong to $X$, then all of them belong to $\cl_{C}^{(j)}$.
	\end{myenumerate}
	
	For $q = 1$, the properties are clear: any vertex cover of a clique with $\ell_C$ vertices must contain at least $\ell_C-1$ vertices.
	Furthermore, property (P2) follows trivially, as $\gd_C^{(j)}$ contains no vertices that are not in $\cl_C^{(j)}$.
	
	Now consider the case $q=2$ and let us first argue about property (P1).
	For contradiction suppose that $X$ contains at most $\ell_C-2$ vertices from $\gd_{C}^{(j)}$; in particular there are two vertices from  $\cl_{C}^{(j)}$ that are not in $X$.
	On the other hand, $X$ must contain at least $\ell_C - 2$ vertices from $\cl_{C}^{(j)}$ as otherwise $G - X$ has a triangle and thus is not bipartite.
	Therefore no other vertices from $\gd_C^{(j)}$ are included in $X$.
	However, recall that for every edge of  $\cl_{C}^{(j)}$ we introduced two vertices adjacent with both endpoints of this edge.
	Thus $G-X$ contains a triangle, a contradiction.
	
	Now let us argue about (P2). For contradiction suppose that $X$ contains $\ell_C-1$ vertices from  $\gd_{C}^{(j)}$,
	but only at most $\ell_C - 2$ of them are in $\cl_{C}^{(j)}$.
	Repeating the argument from the previous paragraph, this means that exactly $\ell_C - 2$ vertices from $\cl_{C}^{(j)}$ are in $X$.
	Since $X$ might contain at most one of the vertices adjacent to both vertices from  $\cl_{C}^{(j)} -X$, we conclude that $G-X$ contains a triangle, a contradiction.

	Note that property (P1) implies that $|X \setminus Y'| \geq  (N+1) \cdot \left( \sum_{C \in \calC} (\ell_C - 1) \right) = k - k'$.

	Without loss of generality we can assume that for each $C \in \calC$ and every vertex $x$ of $\gd_C$, either the copy of $x$ in each $\gd_C^{(j)}$ is included in $X$, or none of them. Furthermore, for $q=2$, the coloring $\phi$ restricted to each copy of $\gd_C$ (after removing $X$) is exactly the same.
	Indeed, otherwise we can pick a copy whose intersection with $X$ is the smallest and use this solution (i.e., vertices in $X$ and the $q$-coloring of the remaining vertices) on all other copies, obtaining another solution with at most the same number of deleted vertices.
	Notice that if for some $C' \in \calC$, the set $X$ contains more than $\ell_{C'}-1$ vertices from each copy of $\gd_{C'}$,
	then
	\[
	|X| = |X \setminus Y'| + |X \cap Y'| \geq (N+1) \cdot \left( \sum_{C \in \calC} (\ell_C - 1) \right) + (N+1) + |X \cap Y'| = k - k' + (N+1) > k,
	\]
	a contradiction. Thus, for each $C \in \calC$ and each $j \in [N+1]$, the set $X$ contains exactly $\ell_C-1$ vertices from $\gd_C^{(j)}$ and, by property (P2), all these vertices belong to $\cl_C^{(j)}$.
	In particular this implies that $|X \setminus Y'| = k-k'$.

	\medskip
	Finally, we set the valuation $\psi \from \calV \to [q+1]$ according to the coloring of vertices in $Y$:
	$v \in \calV$ is assigned the value $i \in [q]$ if $y(v) \notin X$ and $\phi(v)=i$,
	and it is assigned the value $q + 1$ if $y(v) \in X$.
	
	We claim that $\psi$ satisfies all constraints in $\calC$.
	First, let us consider a constraint $C \in \calC^{\textrm{sign}}$ corresponding to block $V_i$ for $i \in [N/b]$.
	Consider a copy $\gd_C^{(j)}$ of $\gd_C$;
	recall that the solution (i.e., the intersection with $X$ and the coloring of the remaining vertices) is the same for each copy.
	Since $\ell_C - 1$ vertices were deleted from $\cl_{C}^{(j)}$, there is a unique vertex $x$ in $\cl_{C}^{(j)}-X$.
	Let $\boldr \in R_C$ be the assignment corresponding to $x$. 
	Since $C \in \calC^{\textrm{sign}}$, exactly $f_i$ variables have the value $q + 1$ in $\boldr$.
	Note that for each such variable $v \in V_i$, the vertex $y(v)$ must be included in $X$.
	Indeed, if $q=1$, then $y(v)x$ is an edge of $G$.
	If $q=2$, then $y(v)$ and $x$ are contained in a triangle whose third vertex is not in $X$.
	Thus, for each $i \in [N/b]$, at least $f_i$ vertices from $\{x(v) ~|~ v \in V|_i\}$ must be included in $X$.
	Since $\sum_{i=1}^{N/b} f_i = k'$ and $|X \setminus Y'| = k-k'$, we conclude that $|X \cap Y| = k'$.
	Consequently, for each $i \in [N/b]$ exactly $f_i$ vertices from $\{x(v) ~|~ v \in V_i\}$ are in $X$.
	This implies that exactly $f_i$ variables from $V_i$ are mapped to $q+1$ and thus $C$ is satisfied.
	Furthermore, we have that in case $q=2$ the vertex $h$ does not belong to $X$.
	Without loss of generality we may assume that in this case we have $\phi(h)=1$ (otherwise we can switch colors $1$ and $2$ in $\phi$).
	
	Now consider $C \in \calC^{\textrm{sat}}$ and again let $\boldr \in R_C$ be the assignment corresponding to the unique vertex $x$ of $\cl_C^{(1)} - X$.
	By the argument in the previous paragraph we observe that for each $v \in V_C$ we have  $\psi(v) = q+1$ whenever $\boldr$ maps $v$ to $q+1$. In particular, if $q=1$, this shows that $C$ is satisfied.
	So consider the case that $q=2$, we need to argue that values $1$ and $2$ are assigned according to $\boldr$.
	As $\phi(h)=1$, we have that $\phi(x)=2$.
	Note that for every $v \in V_C$ such that $\boldr$ maps $v$ to 1 we have an edge joining $x$ and $x$ and $y(v)$, and thus $\phi(y(v))=1$.
	Simlarly, every $v \in V_C$ such that $\boldr$ maps $v$ to 2 is joined with $x$ with a two-edge path whose middle vertex is not in $X$, so we have $\phi(y(v))=2$.
	Consequently, all variables from $V_C$ are assigned according to $\boldr$ and thus $C$ is satisfied.
	
	Summing up, indeed $(G,k)$ is a yes-instance of \coloringVD{q} if and only if $\calI$ is satisfiable.
	
	\paragraph{Structure of the instance.}
	Let us bound the number of vertices of $G$. We will only consider the case that $q=2$, the case of $q=1$ is similar but  simpler. We have
	\begin{align*}
		|V(G)| \leq \ & |Y'| + (N+1) \cdot \sum_{C \in \calC} \left( (\ell_C - 1) \cdot 2 + \binom{\ell_C - 1}{2} \cdot 2 \right)\\ 
		&\leq (N+1) + (N+1) \cdot (N/b + N^{br}) (2 d^{br} + d^{2br}) = N^{\bigO(1)},
	\end{align*}
	as $|\calC|=\abs{\calC^{\textrm{sat}}}+\abs{\calC^{\textrm{sign}}} \leq N/b + N^{br}$, $\ell_C \leq d^{br}$, and $d,b,r$ are constants.
	
	Define $Q \coloneqq Y'$, we clearly have $Q = N+1$. Note that connected components of $G - Q$ are precisely copies of $\gd_C$ for $C \in \calC$. The number of vertices of each such component is at most $\sigma \coloneqq (\ell_C - 1) \cdot 2 + \binom{\ell_C - 1}{2} \cdot 2 \leq d^{2br}$, and each such gadget is connected to at most $\delta \coloneqq br+1$ vertices of $Q$.
	Thus $Q$ is a \core{\sigma}{\delta}, where $\sigma$ and $\delta$ depend only on $d,b,r$, which in turn depend only on $q$ and $\epsilon$.
	
	\paragraph{Running time.} The construction of the instance
	$(G,k)$ takes time  polynomial in $N$. Also, $G$ has a
	\core{\sigma}{\delta} of size at most $p \coloneqq N + 1$. Using the
	hypothetical algorithm that solves all $n$-vertex instances of \coloringVD{q}
	instances with a \core{\sigma}{\delta} of size at most $p$ in time $(q
	+ 1 - \epsilon)^{p} \cdot  n^{\bigO(1)}$, we get an algorithm that
	solves $\calI$ in time
	\begin{equation*}
		(q + 1 - \epsilon)^{p} \cdot n^{\bigO(1)} = (q + 1 - \epsilon)^{N} \cdot N^{\bigO(1)}
	\end{equation*}
	which, by \cref{lem:csp-structured}, contradicts the SETH.
\end{proof}

Now let us move to the proof of our lower bound for \coloringVD{q} for $q \geq 3$.

\begin{lem}\label{lem:qcolVD-2}
	For $q \geq 3$ and $\epsilon > 0$, there exist $\sigma, \delta > 0$ depending on $q$ and $\epsilon$
	such that there is no algorithm solving all $n$-vertex instances of \coloringVD{q} given with a
	\core{\sigma}{\delta} of size at most $\cpar$, in time $(q + 1 - \epsilon)^{\cpar} \cdot n^{\bigO(1)}$, unless the SETH fails.
\end{lem}

\begin{proof}
	Let $b$ and $r$ be the values given by \cref{lem:csp-structured} for $d=q+1$ and $\eps$.
	Consider an $N$-variable ($b,\boldf$)-structured $(q+1,b \cdot r)$-\textsc{CSP} instance $\calI$ with variables $\calV$ and constraints $\calC$, for some $\boldf \colon \big[\frac{N}{b}\big] \to \{0, \ldots,b\}$ and $N \in \mathbb{N}$.
	Again we may assume that each constraint has at least one satisfying assignment.
	
	\paragraph{Auxiliary relations.}
	For a tuple 
	\[
	\boldz = (z_{1,1},\ldots,z_{1,q},z_{2,1},\ldots,z_{2,q},\ldots,z_{b,1},\ldots,z_{b,q}) \in [q]^{bq},
	\]
	by $\textsf{Rainbow}(\boldz)$ we denote the set of those $i$
	for which $\{z_{i,1},\ldots,z_{i,q}\} = [q]$.
	For $j \in \{0,\ldots,b\}$, we define $\textsf{Rb}_j \subseteq [q]^{bq}$ consisting of all tuples $\boldz$ such that $|\textsf{Rainbow}(\boldz)|=j$.
	The important property is that $\textsf{Rb}_j$ is invariant under permutations of the domain.

	For each $i \in [q+1]$ we define a $q$-element vector $\gamma(i)$:
	\begin{myitemize}
		\item if $i =1$, then $\gamma(i) = (2,2,3\ldots,q)$,
		\item if $i \in [2,q]$, then $\gamma(i) = (1,1,2,\ldots,i-1,i+1,\ldots,q)$,
		\item if $i = q+1$, then $\gamma(i) = (1,2,\ldots,q)$.
	\end{myitemize}
	The crux is that if $i \leq q$, then the set of values appearing in $\gamma(i)$ is exactly $[q] \setminus \{i\}$,
	and if $i = q+1$, then the set of values appearing in $\gamma(i)$ is exactly $[q]$.
	
	Now consider a relation $R \subseteq [q+1]^\ell$.
	Let $\textsf{Force}(R) \subseteq [q]^{q\ell}$ be the relation defined as follows. 
	For each $(a_1,a_2,\ldots,a_{\ell}) \in R$, we add to $\textsf{Force}(R)$ the $ \ell q$-element sequence $(\gamma(a_1),\gamma(a_2),\ldots,\gamma(a_{\ell}))$ over $[q]$.
	
	Now we can proceed to our reduction.
	
	\paragraph{Intermediate step: an instance of the list variant.}
	Let us first construct an instance of the \emph{list} variant of \coloringVD{q}, i.e., a triple $(G',L,k)$,
	where $G'$ is a graph whose every vertex $u$ is equipped with a list $L(u) \subseteq [q]$ of admissible colors (the option of deleting a vertex is always available).
	We want to decide whether one can obtain a yes-instance of list $q$-coloring by deleting at most $k := \sum_{i=1}^{N/b} f_i$ vertices from $G'$.
	
	We start by introducing the set $Y = \bigcup_{v \in \calV} \{y(v) \}$ of vertices, each with list $[q]$.
	For each $v \in \calV$, the vertex $y(v)$ represents the variable $v$.
	The intended meaning is that coloring $y(v)$ with a color $c \in [q]$ corresponds to assigning the value $c$ to $v$.
	Deleting $y(v)$ then corresponds to assigning the value $q+1$ to $v$.
	
	Consider a constraint $C \in \calC^{\textrm{sign}}$ related to the block $V_i$.
	We call \cref{prop:realize-coloring} for $\textsf{Rb}_{f_i}$ --- here is where we use the fact that $q\ge 3$ --- to obtain a $bq$-ary gadget
	\[
	\calF(C) = (F(C), L, (z_{1,1},\ldots,z_{1,q},z_{2,1},\ldots,z_{2,q},\ldots,z_{b,1},\ldots,z_{b,q})),
	\]
	with all lists contained in $[q]$ and interface vertices forming an independent set.
	If $v$ is the $j$-th variable in $V_i$, we make vertices $z_{j,1},\ldots,z_{j,q}$ adjacent to $y(v)$.
	We repeat the above step $N+1$ times, i.e., for each $i \in [N/b]$ we introduce $N+1$ copies of the gadget $\calF(C)$.
	
	Now consider a constraint $C \in \calC^{\textrm{sat}}$. We proceed similarly as in the previous case.
	Denote the arity of $C_{\boldf}$ by $\ell$ (where $\ell \leq r'$) and let $R \subseteq [q+1]^{r'}$ be the relation enforced by $C$ (i.e., the set of satisfying assignments). 
	
	We call \cref{prop:realize-coloring} for the relation $\textsf{Force}(R)$ to obtain an $\ell q$-ary gadget
	\[
	\calF(C) = (F(C),L,(z_{1,1},z_{1,2},\ldots,z_{1,q},z_{2,1},\ldots,z_{2,q},\ldots,z_{\ell,1},\ldots,z_{\ell,q})),\]
	with all lists contained in $[q]$. 
	If $v$ is the $j$-th variable of $C$, we make the vertices $z_{i,1},\ldots,z_{i,q}$ adjacent to $y(v)$. 
	
	This completes the construction of $(G',L,k)$.

	\paragraph{Equivalence of instances.} 
	We claim that $(G',L,k)$ is a yes-instance of the list variant of \coloringVD{q} if and only if $\calI$ is satisfiable. 
	
	First, suppose that there exists a satisfying assignment $\varphi : \calV \to [q+1]$ of $\calI$.
	We assign colors to vertices of $Y$ as described above: for $v \in \calV$, if $\varphi(v) \leq q$, then we color $y(v)$ with color $\varphi(v)$. If $\varphi(v) = q+1$, then the vertex $y(v)$ is deleted.
	
	Since all constraints in $\calC^{\textrm{sign}}$ are satisfied, we note that we delete exactly $f_i$ variables from each block $V_i$, and thus in total we delete $k = \sum_{i=1}^{N/b} f_i$ vertices from $Y$. Thus we need to argue that the coloring of the remaining vertices of $Y$ can be extended to the coloring of all gadgets, without deleting any further vertices.
	
	Consider a constraint $C \in \calC^{\textrm{sign}}$ related to the block $V_i$ and a copy of the gadget $\calF(C)$ with the portals $z_{1,1},\ldots,z_{1,q},\ldots,z_{b,1},\ldots,z_{b,q}$.
	Let $j \in [b]$ and $v$ be the $j$-th variable of $C$.
	We color the vertices $z_{j,1},\ldots,z_{j,q}$ so that the vector of colors appearing on these vertices is $\gamma(\varphi(v))$. Note that by the definition of $\gamma$, the color of $y(v)$ does not appear on its neighbors.
	Finally, by the definition of $\textsf{Rb}_{f_i}$, we can extend this coloring of the portals of the gadget to the coloring of the whole gadget, without deleting any vertices from it.
	
	For contraints $C \in \calC^{\textrm{sat}}$ we proceed in a similar way.
	Let $R$ be the relation enforced by $C$ on its variables.
	Consider the gadget $\calF(C)$ introduced for $C$ and let $z_{1,1},\ldots,z_{\ell,q}$ be its portals, where $\ell=\ar(C)$.
	For each $j$, the vertices $z_{j,1},\ldots,z_{j,q}$ receive the vector of colors $\gamma(\varphi(v))$, where $v$ is the $j$-th variable involved in $C$.
	Now the coloring of portals of the gadget is in $\textsf{Force}(R)$, so it can be extended to the coloring of the whole gadget without removing any additional vertices. 
	
	\medskip
	
	Now suppose that there is a set $X \subseteq V(G)$ of size at most $k$ and a proper list coloring $\psi$ of $G'- X$.
	Note that we can safely assume that, for any $C \in \calC^{\textrm{sign}}$,
	all $N+1$ copies of $F(C)$ receive exactly the same coloring.
	Indeed, if this is not the case, we can recolor all copies in the same way as the one in which there are the fewest deleted vertices, obtaining another solution to the problem.
	As $N+1 > N \geq \sum_{i=1}^{N/b} f_i = k$, we conclude that for all $i \in [N/b]$, no vertex from $F(C)$ is deleted.
	
	We set the valuation of variables of $\calI$ according to the coloring of $Y$.
	Consider a variable $v \in \calV$.
	If $y(v) \notin X$, we set the value of $v$ to $\psi(y(v))$, and otherwise we set the value of $v$ to $q+1$.
	Note that at most $|X|\leq k$ variables were mapped to $q+1$.
	We claim that this valuation satisfies all constraints of $\calI$.
	
	First, consider a constraint $C \in \calC^{\textrm{sign}}$ related to the block $V_i$, and a copy of $\calF(C)$ introduced for $C$.
	Recall that no vertex from the gadget was deleted, thus the coloring of the portals of the gadget must respect the relation $\textsf{Rb}_{f_i}$.
	In other words, for $f_i$ variables $v$ of $C$ it holds that all colors from $[q]$ appear in the neighborhood of $y(v)$.
	This means that such a vertex $y(v)$ must be deleted.
	Consequently, for each $i \in [N/b]$, the set $X$ contains at least $f_i$ vertices $y(v)$ for $v \in V_i$.
	
	Summing up, we conclude that $|X| = k$ and that all the deleted vertices are in $Y$.
	Furthermore, exactly $f_i$ vertices corresponding to variables from $V_i$ are deleted, which means that $C$ is satisfied by our valuation.
	
	Now consider a constraint $C \in \calC^{\textrm{sat}}$ of arity $\ell$.
	Consider the gadget $\calF(C)$ introduced for $C$ and denote its portals by $\boldz=(z_{1,1},\ldots,z_{\ell,q})$.
	Recall that by the previous argument no vertex from the gadget is deleted.
	Thus the coloring of the portals of the gadget belongs to $\textsf{Force}(R)$, where $R$ is the relation enforced by $C$.
	Let $(a_1,\ldots,a_p) \in R$ be such that $(\gamma(a_1),\ldots,\gamma(a_p))$ is the coloring of the portals of the gadget.
	
	Let $v$ be the $j$-th variable of $C$ and let $i$ be such that $v \in V_i$.
	Recall that all variables from $V_i$ appear in $C$.
	
	By the definition of $\gamma$, if $a_j = q+1$, then $j \in \mathrm{Rainbow}(\boldz)$ and so $y(v)$ is certainly in $X$.		Consequently, the value of the variable $v$ was set to $q+1$.
	Now consider the case that $a_j \leq q$, say $a_j = c$. Then either $\psi(y(v))=c$ (and thus $v$ gets the value $c$) or $y(v) \in X$.
	However, recall that by the definition of $C$, for exactly $f_i$ variables from $V_i$ their index appears in $\mathrm{Rainbow}(\boldz)$. Note that this is \emph{not} ensured by any relation in $\calC^{\textrm{sign}}$ and this is why $C \in \calC^{\textrm{sat}}$ still needs to check consistency with $\boldf$.

	On the other hand, exactly $f_i$ vertices from $\bigcup_{v' \in V_i} \{y(v')\}$ are in $X$.
	Thus $y(v)$ is not in $X$ and therefore the value of $v$ is set to $c$.
	Summing up, the valuation of the variables of $C$ is exactly $(a_1,\ldots,a_p)$, which means that $C$ is satisfied.

	Let us point out the special structure of the constraints in $\calC^{\textrm{sign}}$ was not used yet and so far we 	 could have obtained the same outcome by introducing $N+1$ copies of each gadget related to a contraint in $\calC^{\textrm{sat}}$.
	
	\paragraph{\boldmath Construction of $(G,k)$.}
	Now let us modify the instance $(G',L,k)$ into an equivalent instance $(G,k)$ of \coloringVD{q}.
	
	For a graph $\widetilde{G}$ with lists $\widetilde{L} : V(\widetilde{G}) \to 2^{[q]}$ and a set $V' \subseteq V(\widetilde{G})$,
	by \emph{simulating lists of $V'$ by $K$} we denote the following operation.
	We introduce a clique $K$ with vertices $c_1,\ldots,c_q$, and for every vertex $u$ of $\widetilde{G}$ and every $i \in [q]$,
	we make $u$ adjacent to $c_i$ if and only if $i \notin \widetilde{L}(u)$.
	
	The typical way of turning an instance of list $q$-coloring to an equivalent instance of $q$-coloring is to introduce a ``global'' $q$-clique $K$ and simulate lists of all vertices of the graph with $K$; recall e.g., the proof of \cref{thm:LBqColSETH}.
	However, in the vertex deletion variant this is no longer that simple, as we have to prevent deleting vertices from this central clique.
	Note that we cannot make them undeletable by introducing (say roughly the size of $G$) many copies of the clique and connecting them in an appropriate way, as this would create an instance without a \core{\sigma}{\delta} for constant $\sigma$, $\delta$.
	
	\medskip
	Let us discuss how to modify $(G',L,k)$ into $(G,k)$.
	Note that the vertices from $Y$ already have lists $[q]$ so they do not have to be simulated.
	For each constraint $C \in \calC^{\textrm{sign}}$ and for each copy of the gadget $\calF(C)$ introduced for $C$,
	we introduce a private $q$-clique and simulate the lists of all vertices of the gadget by this clique.
	
	Next we introduce a single $q$-clique $K$, and the lists of all the remaining vertices, i.e., the vertices from the gadgets  $\calF(C)$ introduced for $C \in \calC^{\textrm{sat}}$, are simulated by $K$. This completes the construction of $G$.
	
	Observe that $K$ is meant to ``synchronize'' the colorings of the gadgets $\calF(C)$ for $C \in \calC^{\textrm{sat}}$.
	However, the colorings of the gadgets $\calF(C)$ for $C \in \calC^{\textrm{sign}}$ are not synchronized, i.e., each clique introduced for these gadgets might receive a different permutation of colors.
	
	Let us argue that $(G',L,k)$ is a yes-instance of the list variant of \coloringVD{q} if and only if $(G,k)$ is a yes-instance of \coloringVD{q}.
	
	First suppose that there is a set $X \subseteq V(G')$ of size at most $k$ and a proper $q$-coloring $\psi$ of $G' - X$ that respects lists $L$. Note that we can delete from $G$ exactly the same vertices $X$,
	and use exactly the same coloring $\psi$ on the remaining vertices from $V(G) \cap V(G')$, and then extend this partial coloring to the additionally introduced $q$-cliques according to the description of the operation of simulating lists. 
	
	On the other hand, suppose that there is a set $X \subseteq V(G)$ of size at most $k$ and a proper $q$-coloring $\psi$ of $G - X$.
	We aim to show that $\psi$ respects lists $L$ on vertices from $V(G) \cap V(G')$.
	
	Similarly as we did in the previous paragraph when analyzing the properties of $(G',L,k)$ we observe that no vertex from any gadget introduced for $C \in \calC^{\textrm{sign}}$, along with its private $q$-clique, was deleted. Indeed, this is because there are $N+1$ copies of each gadget.
	Consider $i \in [N/b]$ and the constraint $C \in \calC^{\textrm{sign}}$ introduced for the block $V_i$. 
	Next, recall that the relation $\textsf{Rb}_{f_i}$ enforced by $C$ is invariant to permuting the set of colors.
	Consequently, simulating lists by a local copy of $K_q$ is sufficient to ensure that the coloring $\psi$ of the vertices of the gadget satisfies lists $L$. This is exactly the point where we use the special structure of constrains in $\calC^{\textrm{sign}}$.
	
	Now, similarly as in the previous case, we can argue that $X \subseteq Y$.
	Consequently, the vertices from $K$, i.e., the common copy of $K_q$ introduced for all gadgets $\calF(C)$ for $C \in \calC^{\textrm{sat}}$,
	are not deleted.
	Thus they simulate the lists of all vertices within the gadgets, i.e., the coloring $\psi$ respects lists $L$.
	
	\paragraph{\boldmath Structure of $G$.} 
	The number of vertices of $G$ is at most (below $h_1(\cdot), h_2(\cdot)$ are some functions)
	\begin{align*}
		n = & |Y| + |K| + \sum_{C \in \calC^{\textrm{sign}} } (|V(F(C)|+q) \cdot (N+1) + \sum_{C \in \calC^{\textrm{sat}}} |V(F(C))| \\
		\leq & N + q + N \cdot (h_1(b,q)+q) \cdot (N+1) + |\calC| \cdot h_2(b,q,r) \\
		\leq & N + q + N(N+1) \cdot (h_1(b,q)+q) + N^r \cdot h_2(b,q,r) \\
		= &  \bigO(N^{r+1}) = N^{\bigO(1)},
	\end{align*}
	as $q,b,r$ are constants.
	
	Set $Q= Y \cup K$, note that $|Q| = N+q$.
	There are two types of components of $G-Q$:
	(i) (subgraphs of) copies of the gadgets $F(C)$  introduced for $C \in \calC^{\textrm{sign}}$,
	including the local $q$-clique, or
	(ii) (subgraphs of) copies of the gadgets $F(C)$  introduced for $C \in \calC^{\textrm{sat}}$.
	Using the notation from the formula above, the components of type (i) have size at most $h_1(b,q)+q$ and each of them attaches to $b$ vertices from $Q$, and the components of type (ii) have size at most $h_2(b,q,r)$ and each of them attaches to at most $br+q$ vertices from $Q$.
	Since $r$ depends on $q$ and $\epsilon$, we conclude that $Q$ is a \core{\sigma}{\delta} for some $\sigma$, $\delta$ depending only on $q$ and $\epsilon$.
	
	\paragraph{Running time.}
	Now let us estimate the running time.
	The construction of $(G,k)$ takes polynomial time. Furthermore, $G$ has a \core{\sigma}{\delta} of size at most $\cpar \coloneqq N + \bigO(1)$. This implies that if there is an algorithm solving all $n$-vertex instances of \coloringVD{q} with a
	\core{\sigma}{\delta} of size at most $\cpar$ in time $(q + 1 - \epsilon)^{\cpar} \cdot n^{\bigO(1)}$, then we get an algorithm solving all $N$-variable ($b,\boldf$)-structured instances of $(q+1,b \cdot r)$-\textsc{CSP} in time
	\begin{equation*}
		(q+1-\epsilon)^{p} \cdot n^{\bigO(1)}=(q+1-\epsilon)^{N} \cdot N^{\bigO(1)},
	\end{equation*}
	which, by \cref{lem:csp-structured}, contradicts the SETH.
\end{proof}

Now, \cref{thm:vd-coloring-intro} follows from combining \cref{lem:qcolVD-1} and \cref{lem:qcolVD-2}.

\subsection{\boldmath Simple Algorithm for \coloringVD{q}}\label{sec:vd-basic-algo}
The following algorithm solves the \coloringVD{q} problem for $q \geq 1$, which also includes the \textsc{Vertex Cover} problem for $q = 1$ and \textsc{Odd Cycle Transversal} for $q = 2$. It is worth noting that this algorithm can be modified to use treewidth as the parameter.

\begin{thm}
	For all integers $q, \sigma, \delta \geq 1$, every $n$-vertex instance of \coloringVD{q} can be solved in time $(q+1)^{p} \cdot n^{\bigO(1)}$ if a \core{\sigma}{\delta} of size $p$ is given in the input.
\end{thm}

\begin{proof}
As input to \coloringVD{q}, consider a graph $G$ with $n$ vertices, given along with a \core{\sigma}{\delta} $Q$ of size $p$.
We exhaustively guess the set $X \subseteq Q$ of \coreword\ vertices that are deleted and a $q$-coloring $f$ of $Q-X$.
This results in at most $(q+1)^p$ branches.

Now, for each component $C$ of $G-X$, we compute the minimum number of vertices to be deleted from $C$,
so that the remaining vertices admit a proper $q$-coloring that extends $f$.
We output the coloring with the fewest vertices deleted in total (i.e., from $Q$ and from the components of $G-Q$).

Recall that each component $C$ of $G-Q$ is of size at most $\sigma$, i.e., a constant.
Thus, the optimum extension of the coloring $Q-X$ to $C$ can be computed in constant time.
Since the number of components $C$ is at most $n$, we conclude that the running time of each branch is bounded by polynomial.
Summing up, the overall running time is $(q+1)^p \cdot n^{\bigO(1)}$.
\end{proof}

\section{\boldmath Edge Deletion to \coloring{q}}\label{sec:ed-basic}

Now let us move to the edge-deletion counterpart of \coloringVD{q}.
The goal of this section is to prove \cref{thm:coloringEDcombined,thm:LBcolEDMSHintro}.

The algorithmic part \cref{thm:coloringEDcombined}\,(1) is again simple and stated in \cref{sec:ed-basic-algo} for completeness.
For the hardness result \cref{thm:coloringEDcombined}\,(2), the case $q\ge3$ already follows from our result \cref{thm:coloringcombined} for the decision problem. However, the case $q=2$ (i.e., the result for \maxcut) is the missing piece, and so it remains to show the following hardness result.

\begin{thm}\label{thm:LBmaxcutSETH}
	For every $\eps>0$, there are integers $\sigma$ and $\delta$ such that no algorithm solves every $n$-vertex instance of \maxcut that is given with a \core{\sigma}{\delta} of size $p$, in time $(2-\eps)^p\cdot n^{\bigO(1)}$, unless the SETH fails.
\end{thm}

Our second main result covered in this section, \cref{thm:LBcolEDMSHintro}, states that
under the \msh the lower bound from \cref{thm:coloringEDcombined} holds even for universal constants $\sigma$ and $\delta$, i.e., for $\sigma$ and $\delta$ that no longer depend on $\eps$. 
For the case of \maxcut, we show that constant $\sigma$ and $\delta=4$ suffice. 

\begin{thm}\label{thm:LBmaxcutMSH}
	There is an integer $\sigma$ such that, for every $\eps>0$, no algorithm solves every $n$-vertex instance of \maxcut that is given with a \core{\sigma}{4} of size $p$, in time $(2-\eps)^p\cdot n^{\bigO(1)}$, unless the \msh fails. 
\end{thm}

For $q\ge3$, the respective lower bound for \coloringED{q} can be obtained for constant $\sigma$ and $\delta=6$.

\begin{thm}\label{thm:LBcolEDMSH}
	For every $q\ge 3$ there is an integer $\sigma$ such that the following holds.
	For every $\eps>0$, no algorithm solves every $n$-vertex instance of \coloringED{q} that is given with a \core{\sigma}{6} of size $p$, in time $(q-\eps)^p\cdot n^{\bigO(1)}$, unless the \msh fails.
\end{thm}

The combination of \cref{thm:LBmaxcutMSH,thm:LBcolEDMSH} gives \cref{thm:LBcolEDMSHintro}.

This section is structured as follows. 
First, in \cref{sec:maxCSPMSH}, we establish a lower bound for \textsc{Max-CSP} (\cref{thm:M3SHtoMaxCSP}) based on the \msh. This will be the starting point for our \msh-based lower bounds. 
Second, in order to bridge from the coloring problems to \textsc{Max-CSP} we show how to model arbitrary relations with the coloring problems via gadget constructions. For $q\ge 3$ we can rely on the previously established \cref{prop:realize-coloring}. For $q=2$,we have to do some additional work, and this is done in \cref{sec:ed-relations}. 
Third, in \cref{sec:ed-maxcut}, we consider the case $q=2$ --- here we show \cref{thm:LBmaxcutSETH,thm:LBmaxcutMSH} in a unified proof. Finally, in \cref{sec:ed-col}, we treat the case $q\ge 3$ and show \cref{thm:LBcolEDMSH}.
In both \cref{sec:ed-maxcut} and \cref{sec:ed-col}, we show intermediate results for the \listcoloringED{q} problem. In fact, we first show the following results for list colorings, and subsequently show how to remove the lists in the reductions.

\LBLHomEDKqSETH*

\LBLHomEDKqMSH*

\subsection{\textsc{Max\,CSP} --- Hardness under \msh}\label{sec:maxCSPMSH}
Let us introduce a constraint-deletion variant of $(d,r)$-\textsc{CSP}.
An instance of \MaxCSP{$d$}{$r$} is the same as in $(d,r)$-\textsc{CSP} and we ask for a smallest possible number of constraints that need to be deleted in order to obtain a yes-instance of $(d,r)$-\textsc{CSP}.

Clearly, using a brute-force approach we can solve an $n$-variable instance of  \MaxCSP{$d$}{$r$} in time $d^n \cdot n^{\bigO(1)}$.
Let us explore the consequences of M3SH to the complexity of \MaxCSP{$d$}{$r$}.
Clearly, \MaxSat is a special case of \MaxCSP{$2$}{$3$}, which means that there is no $\epsilon>0$ such that for any $r \geq 3$,
\MaxCSP{$2$}{$r$} with $n$ variables can be solved in time $(2-\epsilon)^n \cdot n^{\bigO(1)}$, unless the M3SH fails.

\begin{thm}\label{thm:powerof2}
	For $p \geq 1$ and any $r \geq 3$, there is no algorithm solving every $n$-variable instance of \MaxCSP{$2^p$}{$r$} in time $(2^p-\epsilon)^n\cdot n^{\bigO(1)}$, for $\epsilon>0$, unless the M3SH fails.
\end{thm}
\begin{proof}
	Note that it is sufficient to prove the lower bound for $r=3$.
	Consider an instance $I_1$ of \MaxSat with variables $V$ and the set of clauses $C$, where $|V|=n$. By introducing fewer than $p$ dummy variables,
	we can assume that $n$ is divisible by $p$.
	Let $N = n/p$. We partition $V$ into $N$ blocks of size $p$.
	Denote these blocks by $V_1,V_2,\ldots,V_N$.
	
	Now we are going to create an instance $I_2$ of \MaxCSP{$2^p$}{$3$}. For each $i \in [N]$, we introduce a new variable $x_i$.
	We bijectively map each possible valuation of variables in $V_i$ to one of $2^p$ possible values of $x_i$.
	
	Now consider a clause $C$ of the \MaxSat instance, involving variables $v,v',v''$, such that $v \in V_i$, $v' \in V_{i'}$, and $v'' \in V_{i''}$.
	We introduce a constraint involving variables $x_i, x_{i'}$, and $x_{i''}$, which is satisfied by exactly those values that correspond to valuations that
	satisfy $C$. We repeat this step for every clause $C$; note that the arity of each constraint is at most 3.
	
	It is clear that the created instance $I_2$ of \MaxCSP{$2^p$}{3} is equivalent to the original instance $I_1$ of \MaxSat.
	Now let us argue about the running time.
	For contradiction, suppose we can solve $I_2$ in time $(2^p-\epsilon)^N \cdot N^{\bigO(1)}$ for some $\epsilon > 0$,
	i.e., $(2^p)^{\delta N}\cdot N^{\bigO(1)}$ for some $\delta <1$ depending on $p$ and $\epsilon$.
	This means that $I_1$ can be solved in time $(2^p)^{\delta N} \cdot N^{\bigO(1)} = 2^{\delta n} \cdot n^{\bigO(1)}= (2-\epsilon')^n\cdot n^{\bigO(1)}$ for some $\epsilon'>0$. This contradicts the M3SH.
\end{proof}

\begin{lem}\label{lem:randomized-reduction}
	Let $d \geq 3$ and let $d'$ and $p$ be integers such that $d'=2^p$ and $d\le d'$.
	If there is $\eps>0$ such that every $n$-variable instance of \MaxCSP{$d$}{$3$} can be solved in time $(d-\epsilon)^n\cdot n^{\bigO(1)}$,
	then there is $\epsilon' > 0$ depending on $d$ and $\epsilon$ such that every $n$-variable instance of \MaxCSP{$d'$}{$3$} can be solved in randomized time $(d'-\epsilon')^n\cdot n^{\bigO(1)}$.
\end{lem}
\begin{proof}
	Consider an instance $I_1$ of \MaxCSP{$d'$}{$3$} with $n$ variables.
	Let $\varphi$ be the (unknown) optimum solution, i.e., the one that violates the minimum number of constraints;
	let this number be $k$.
	
	For each variable of $I_1$ independently, uniformly at random we discard $d'-d$ possible valuations.
	More precisely, for each variable $v$ we select a subset $D(v) \subset [d']$ of size $d$ and set it as the domain of $v$.
	We modify each constraint to be satisfied by only these valuations that for each variable $v$ pick a value from $D(v)$;
	note that it is possible that some constraints become not satisfiable at all.
	This way we have created an instance $I_2$ of \MaxCSP{$d$}{$3$}.
	
	Note that, while in $I_2$ we only consider a subset of the valuations of $I_1$ the fact whether some clause is violated or not is preserved.
	In particular, if $\varphi$ satisfies the domains of all variables in $I_2$, then it is an optimal solution of $I_2$.
	
	Note that the probability that $\varphi$ satisfies all domains is $\left( \frac{d}{d'} \right)^n$.
	Thus, using standard calculations we observe that if we repeat our random experiment $(\frac{d'}{d})^n$ times, then the success probability is constant.
	
	So the algorithm is as follows: we randomly create $(\frac{d'}{d})^n$ instances of \MaxCSP{$d$}{$3$} by discarding some valuations as previously described, and each of them is solved with our hypothetical algorithm. We return the best of found solutions; with constant probability this solution will be an optimal solution for $I_1$.
	
	Let us compute the running time (below $\delta, \delta' < 1$ and $\epsilon'>0$ are some constants depending on $d$ and $\epsilon$):
	\[
	\left(\frac{d'}{d} \right)^n \cdot (d-\epsilon)^n \cdot n^{\bigO(1)}= \left(\frac{d'}{d} \right)^n \cdot d^{\delta n} \cdot n^{\bigO(1)}= d'^{\delta' n} \cdot n^{\bigO(1)}= (d' - \epsilon')^n\cdot n^{\bigO(1)}.
	\]
	We remark that we can make the success probability arbitrarily close to 1 by adding more iterations of sampling, which results in a blow-up in the running time by a constant factor.
\end{proof}

Now let us discuss how to derandomize the argument in Lemma~\ref{lem:randomized-reduction}.
Instead of repeated random sampling, we want to enumerate a family $\calF$, such that
\begin{myitemize}
	\item each set in $\calF$ is of the form  $D_1 \times \ldots \times D_n$, where each $D_i$ is a $d$-element subset of $[d']$,
	\item the union of all sets in $\calF$ covers $[d']^n$, and
	\item $\calF$ can be enumerated in time roughly $\left(\frac{d'}{d}\right)^n$.
\end{myitemize}
Then we could replace the random sampling step by testing each member of $\calF$ separately.

Notice that the idea above is actually about finding a small (approximate) solution to a certain \SetCov instance.
Recall that an instance of \SetCov consists of a universe $U$ and a family $\calS$ of subsets of $U$, for which we can assume without loss of generality that $\bigcup_{S \in \calS} S = U$.
The task is to find a minimum-sized family $\calS' \subseteq \calS$ such that $\bigcup_{S \in \calS'} S = U$.

For integers $d' \geq d$ and $n$, by $\calS(d',d,n)$ we denote the family consisting of sets of the form $D_1 \times D_2 \times \ldots \times D_n$, for all possible choices of $D_1,\ldots,D_n \in \binom{[d']}{d}$. Thus, in this notation, our task is to look for a small subfamily of $\calS(d',d,n)$ which covers $[d']^n$.

\begin{lem}\label{lem:setcovergap}
	For every $n>1$ and $d' \geq d$, the \SetCov instance $([d']^n, \calS(d',d,n))$ has a solution of size $\left( \frac{d'}{d} \right)^n \cdot n^{\bigO(1)}$.
\end{lem}
\begin{proof}
	For an instance $(U,\calS)$ of \SetCov, a \emph{fractional solution}  is a function $f : \calS \to [0,1]$, such that for every $u \in U$ we have $\sum_{\substack{S \in \calS\\ u \in S}} f(S) \geq 1$. The weight of a fractional solution is $\sum_{S \in \calS} f(S)$.
	
	We claim that $([d']^n, \calS(d',d,n))$ admits a fractional solution of weight at most $\left( \frac{d'}{d} \right)^n$.
	Indeed, set $f(S) = \binom{d'-1}{d-1}^{-n}$ for every $S \in \calS(d',d,n)$.
	First, note that each element of $[d']^n$ is covered by exactly $\binom{d'-1}{d-1}^n$ sets in $\calS(d',d,n)$, so indeed $f$ is a fractional solution.
	Its weight is 
	\[
	\sum_{S \in \calS(d',d,n)} f(S) = \left( \frac{\binom{d'}{d}} { \binom{d'-1}{d-1}} \right)^n = \left( \frac{d'}{d} \right)^n.
	\]
	
	Next, we recall a well-known fact that the integrality gap for \SetCov is bounded by a logarithmic function of the size of the universe~\cite{LOVASZ1975383}.
	This, combined with the observation from the previous paragraph, implies that  $([d']^n, \calS(d',d,n))$ has an integral solution of size  $\left( \frac{d'}{d} \right)^n \cdot n^{\bigO(1)}$.
\end{proof}

Note that \cref{lem:setcovergap} is not sufficient for our derandomization procedure, as we need to be able to find this solution efficiently.

\begin{lem}\label{lem:setcoverefficient}
	Let $\delta >0$ and $d' \geq d \geq 2$ be constants.
	For every $n$, there is a family $\calF \subseteq \calS(d',d,n)$ of size at most $\left( \frac{d'}{d} + \delta \right)^n$ that covers $[d']^n$ and can be enumerated in time $\bigO\left(|\calF|\right)$.
\end{lem}
\begin{proof}
	In what follows we assume that $n$ is sufficiently large, as otherwise we can find the optimum solution using brute force.
	The exact lower bound on $n$ depends on $d,d'$, and $\delta$ and follows from the reasoning below.
	
	Let $c$ be the constant hidden in the $\bigO(\cdot)$-notation in \cref{lem:setcovergap}, i.e., 
	such that for every $n$, the instance $([d']^n, \calS(d',d,n))$ has an integral solution of size at most $\left( \frac{d'}{d} \right)^n \cdot n^c$.
	Let $m$ be the minimum integer such that $m^{c/m} < 1+\delta d/d'$.
	Note that $m$ depends only on $\delta,d$, and $d'$, as $c$ is an absolute constant.
	
	For simpilicity let us assume that $m$ divides $n$; the argument below can be easily modified to the general case by making the last block smaller.
	Let $\calF'$ be the family given by \cref{lem:setcovergap} for $([d']^m, \calS(d',d,m))$; note that we can compute it by brute force as $m$ is a constant.
	We observe that the family
	\[
	\calF = \underbrace{\calF' \times \calF' \times \ldots \times \calF'}_{n/m \text{ times}}
	\]
	covers $[d']^n$.
	
	Let us estimate the size of $\calF$. It is upper-bounded by
	\[
	|\calF'|^{n/m} \leq \left( \left( \frac{d'}{d} \right)^m \cdot m^c \right)^{n/m} = \left( \frac{d'}{d} \right)^n \cdot (m^{c/m})^n  <  \left( \frac{d'}{d} \right)^n \cdot \left( 1 + \frac{d \delta}{d'} \right)^n =  \left( \frac{d'}{d} + \delta \right)^n.
	\]
	Clearly the family $\calF$ can be enumerated in time proportional to its size.
\end{proof}

We can now prove \cref{thm:M3SHtoMaxCSP}, which we restate for convenience.
\MTSHtoMaxCSP*
\begin{proof}
	If $d$ is a power of 2, then the result follows already from \cref{thm:powerof2}.
	So assume that $d$ is not a power of 2, and let $d'$ be the smallest power of 2 which is greater than $d$.
	We reduce from \MaxCSP{$d'$}{$r$} consider an instance with $n$ variables.
	
	For contradiction, suppose that there is some $\epsilon>0$ and some algorithm solving every $n$-variable instance of \MaxCSP{$d$}{$r$} in time $(d-\epsilon)^n\cdot n^{\bigO(1)}$.
	Let $\delta := \frac{\epsilon d'}{d^2}$ and let $\calF$ be the family given by \cref{lem:setcoverefficient} for this value of $\delta$.
	We proceed similarly as in the proof of \cref{lem:randomized-reduction}: for each set in $\calF$ we construct a corresponding instance of \MaxCSP{$d$}{$r$} and solve it using our hypothetical algorithm.
	As $\calF$ covers $[d']^n$, we know that for some set in $\calF$ we will find the optimum solution for the original instance.
	
	Let us estimate the running time:
	\[
	\bigO\left(\left( \frac{d'}{d} + \delta \right)^n\right) + \left( \frac{d'}{d} + \delta \right)^n \cdot (d-\epsilon)^n \cdot n^{\bigO(1)}= \left (d'  + \delta d - \epsilon \frac{d'}{d} -  \epsilon \delta  \right)^n \cdot n^{\bigO(1)} =  \left (d'  - \frac{\epsilon^2 d'}{d^2} \right)^n \cdot n^{\bigO(1)}.
	\]
	By \cref{thm:powerof2}, this contradicts the M3SH.
\end{proof}

\subsection{Realizing Relations}\label{sec:ed-relations}

We start by defining what it means for a gadget to realize a relation. Recall from the definition in \cref{sec:prelims} that a gadget may use lists.

\begin{defn}[Realizing a relation]\label{def:realizing}
	For some positive integer $q$ and $r$, let $R\subseteq [q]^r$. Consider some $r$-ary $q$-gadget $\calJ=(J,L,\boldx)$, and some $\boldd\in [q]^r$. By $\edcount(\calJ, \boldd)$ we denote the size of a minimum set of edges $X$ that ensure that there is a proper list $q$-coloring $\phi$ of $(J\setminus X,L)$ with $\phi(\boldx)=\boldd$.  Then $\calJ$ \emph{realizes} the relation $R$ if there is an integer $k$ such that, for each $\boldd\in [q]^r$, $\edcount(\calJ, \boldd)=k$ if $\boldd\in R$, and $\edcount(\calJ, \boldd)>k$, otherwise.
	We say that $\calJ$ \emph{$\omega$-realizes} $R$ for some integer $\omega\ge 1$ if additionally, for each $\boldd\notin R$ we have $\edcount(\calJ, \boldd)=k+\omega$.
\end{defn}

\begin{lem}\label{lem:NEQtoRel}\label{lem:ed-allrelations}
	For each $r\ge 1$ and $R\subseteq [2]^r$ there is an $r$-ary $2$-gadget that realizes $R$.
\end{lem}
\begin{proof}
First note that a single edge $x_1x_2$ whose endpoints have each have a list $L(x_1)=L(x_2)=\{1,2\}$ enforces that its endpoints are mapped to different colors. So essentially this is a gadget that realizes the relation $\NEQ=\{(1,2),(2,1)\}$.

For an integer $p \geq 2$, by $\OR{p}$ we denote the relation $[2]^p \setminus \{2^p\}$ (the intuition behind the notation is clear when we interpret $1$ as \emph{true} and $2$ as \emph{false}).

\begin{clm}\label{clm:or2}
There exists a gadget $\calJ_{\OR{2}}$ whose portals are pairwise non-adjacent
that $2$-realizes $\OR{2}$.
\end{clm}
\begin{claimproof}
We start the construction of the gadget with a 5-cycle with consecutive vertices $v_1,v_2,v_3,v_4,v_5$.
We set the list of $v_5$ to $\{2\}$; the lists of $v_i$ for $i \in [4]$ are $\{1,2\}$.
The portals of the created gadget $\calJ_{\OR{2}}$ are $v_1$ and $v_4$, see also \cref{fig:or2}.

It is straightforward to verify that for $\boldd \in \OR{2}$ precisely one edge has to be deleted and hence $\edcount(\calJ_{\OR{2}}, \boldd)= 5\alpha+1$,
while for $(2,2)$ we have to delete precisely three edges and hence $\edcount(\calJ_{\OR{2}}, (2,2))=5\alpha + 3$.
\end{claimproof}

\begin{clm}\label{clm:or2tau}
For every fixed integer $\omega \geq 1$ there exists a gadget $\calJ^\omega_{\OR{2}}$  whose portals are pairwise non-adjacent that $2\omega$-realizes $\OR{2}$.
\end{clm}
\begin{claimproof}
We introduce $\omega$ copies of the gadget $\calJ_{\OR{2}}$ from \cref{clm:or2}, let the portals of the $i$-th copy be $x_i,y_i$.
We identify all $x_i$ into a new vertex $x$ and all $y_i$ into a new vertex $y$; note that the portals of $\calJ_{\OR{2}}$ are non-adjacent so this step does not introduce multiple edges.
The obtained graph with $x$ and $y$ as portals is our gadget $\calJ^\omega_{\OR{2}}$.

For $\boldd\in \OR{2}$, let $\edcount(\calJ_{\OR{2}}, \boldd)= \alpha'$.
It is straighforward to observe that for $\boldd \in \OR{2}$ we have $\edcount(\calJ^\omega_{\OR{2}}, \boldd)= \omega \alpha'$, and for $\bar \boldd \notin \OR{2}$ we have $\edcount(\calJ^\omega_{\OR{2}}, \bar \boldd)= \omega \cdot (\alpha'+2) = \omega \alpha' + 2\omega$.
\end{claimproof}

\begin{clm}\label{clm:orp}
For every $p \geq 3$, there exists a gadget $\calJ_{\OR{p}}$ realizing $\OR{p}$ whose portals are pairwise non-adjacent.
\end{clm}
\begin{claimproof}
Let $\omega$ be an integer whose value will be specified later.

We introduce three pairwise disjoint sets of vertices $X := \bigcup_{i=1}^p \{x_i\}$, $Y := \bigcup_{i=1}^p \{y_i\}$,
and $Z := \bigcup_{i=1}^p \{z_i\}$
For each $i \in [p]$ we add a copy of $\calJ^\omega_{\OR{2}}$ with portals identified with $x_i$ and $y_i$,
and another copy with portals identified with $y_i$ and $z_i$.
Next, for all distinct $i,j$, we add a copy of $\calJ^\omega_{\OR{2}}$ with portals identified with $z_i,z_j$.

Next, we add a new vertex $u$ with list $\{1\}$ that is adjacent to every vertex in $X \cup Y \cup Z$. 
This completes the construction of the gadget. Its portals are $y_1,\ldots,y_p$.
A schematic picture of the construction for $p=3$ is depicted in \cref{fig:or3}.

We say that a copy of $\calJ^\omega_{\OR{2}}$ is satisfied if the mapping of its portals satisfies $\OR{2}$.
Imagine that $\omega$ is sufficiently large so that every copy of $\calJ^\omega_{\OR{2}}$ must be satisfied.
Because of the universal vertex $u$, we have to delete one additional edge for every vertex in $X \cup Y \cup Z$ that is mapped to $1$.
Intuitively, we want to map as many of these vertices as possible to $2$.
Since all copies of $\calJ^{\omega}_{\OR{2}}$ must be satisfied, at most $p+1$ vertices from $X \cup Y \cup Z$ can be mapped to $2$: precisely one vertex from each pair $\{x_i, y_i\}$ is mapped to $2$, and potentially one vertex from the set $Z$. The maximum number of $2$s is only possible if at least one vertex from $Y$ is mapped to $1$.
In other words, this happens if and only if the mapping of portals of $\calJ_{\OR{p}}$ satisfies $\OR{p}$.

For $\boldd\in \OR{2}$, let $\edcount(\calJ_{\OR{2}}, \boldd)= \alpha'$.
If $\boldd \in \OR{p}$, then $\edcount(\calJ_{\OR{p}}, \boldd) = \alpha''$, where $\alpha'' \coloneqq (2p-1) + \alpha' \cdot \left( \binom{p}{2} + 2p \right)$.
By setting $\omega = \alpha''+1$ we ensure that we cannot afford making even one copy of $\calJ^\omega_{\OR{2}}$ unsatisfied.
\end{claimproof}

\begin{figure}
\begin{subfigure}[b]{0.45\textwidth}
\centering
\includegraphics[scale=1,page=1]{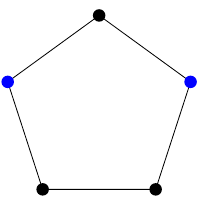}
\caption{The gadget realizing $\OR{2}$ constructed in \cref{clm:or2}.}
\label{fig:or2}
\end{subfigure}
\hfill
\begin{subfigure}[b]{0.45\textwidth}
\centering
\includegraphics[scale=1,page=2]{figures/ed-gadgets.pdf}
\caption{The gadget realizing $\OR{3}$ constructed in \cref{clm:orp}.}
\label{fig:or3}
\end{subfigure}
\caption{Gadgets constructed in the proof of \cref{lem:ed-allrelations}.
Black lines correspond to edges, and red lines denote copies of $\calJ^\omega_{\OR{2}}$. Blue vertices are portals.
}
\end{figure}

For an integer $p \geq 2$, and $\boldd \in [2]^p$,
by $R_{\neq \boldd}$ we denote the relation $[2]^p \setminus \{\boldd\}$.

\begin{clm}\label{clm:notr}
For every $p \geq 2$ and every $\boldd \in [2]^p$ there exists a gadget $\calJ_{R_{\neq \boldd}}$ realizing $R_{\neq \boldd}$
whose portals are non-adjacent.
\end{clm}
\begin{claimproof}
Let $\boldd = (d_1,\ldots,d_p)$.
Introduce a copy of $\calJ_{\OR{p}}$ given by \cref{clm:orp}.
Let its portals be $x_1,\ldots,x_p$.
Let $I =\{ i \in [p] \mid d_i = 1\}$.
For each $i \in I$, introduce a vertex $y_i$ that is adjacent only to $x_i$.

This completes the construction of $\calJ_{R_{\neq \boldd}}$. For $i \in [p]$, its $i$-th portal is $x_i$ (if $i \notin I$) or $y_i$ (if $i \in I$). It is straightforward to verify that $\calJ_{R_{\neq \boldd}}$ satisfies all required properties.
\end{claimproof}

Finally, we show that using previously obtained gadgets we can realize any relation $R$.

\begin{clm}\label{clm:allR}
For every $p \geq 2$ and every $R \subseteq  [2]^p$, there exists a gadget $\calJ_{R}$ realizing $R$ whose portals are non-adjacent.
\end{clm}
\begin{claimproof}
We start with introducing vertices $x_1,\ldots,x_p$, which will be the portals of our gadget.
For each $\bar \boldd \in \{a,b\}^p \setminus R$ we call \cref{clm:notr} to obtain $\calJ_{R_{\neq \bar \boldd}}$ and identify the respective portals of the gadget with $x_1,x_2,\ldots,x_p$.
It is straightforward to verify that such a gadget indeed realizes $R$.
\end{claimproof}

This completes the proof.
\end{proof}

We strengthen the previous result by balancing out the violation costs. 

\maxcutRelations*
\begin{proof}
	Let $\bar R\coloneqq [2]^r\setminus R$, and let $\bar \boldd_1, \ldots, \bar \boldd_{\abs{\bar R}}$ be an enumeration of the elements of $\bar R$.
	
	We define a relation $R'$ in $[2]^r\times [2]^{\abs{\bar R}}$ as follows:
	A tuple $(x_1, \ldots, x_r, y_1, \ldots, y_{\abs{\bar R}})$ is in $R'$ if and only if one of the following holds:
	\begin{enumerate}
		\item $(x_1, \ldots, x_r)\in R$ and $y_i=1$ for all $i\in [\abs{\bar R}]$, or
		\item for some $i\in [\abs{\bar R}]$, $(x_1, \ldots, x_r)=\bar \boldd_i$, $y_i=2$, and $y_j=1$ for all $j\neq i$.
	\end{enumerate}
	According to \cref{lem:NEQtoRel}, there is a gadget $\calJ$ that realizes $R'$.
	Let $Z_1, \ldots, Z_r$ be the portals of $\calJ$ that belong to the respective entries $x_1\ldots, x_r$, and let $z_1, \ldots, z_{\abs{\bar R}}$ be those that belong to $y_1, \ldots, y_{\abs{\bar R}}$, respectively.
	By definition of a realization, there is some integer $\alpha$ such that for each $\boldd\in R'$, we have $\edcount(\calJ,\boldd)=\alpha$.
	
	We modify $\calJ$ to form a new gadget $\calJ'$ by attaching to each portal $z_i$ a new vertex $v_i$ with list $L(v_i)=\{2\}$.
	Now interpret this modified gadget $\calJ'$ as a gadget with portals $Z_1, \ldots, Z_r$. We claim that $\calJ'$ $1$-realizes $R$:
	Whenever $Z_1, \ldots, Z_r$ are mapped to a state in $R$, mapping all of the $z_i$s to $1$ requires only $\alpha$ edge deletions in $\calJ$ and no further edge deletions outside as the $v_i$'s can be mapped to $2$.
	However, if $Z_1, \ldots, Z_r$ are in a state $\bar \boldd_i\in \bar R$ then $\alpha$ edge deletions are required in $\calJ$, and they are sufficient if $z_i$ is mapped to $2$. However, this requires one more edge deletion between $z_i$ and $v_i$, both of which are mapped to $2$. At the same time $\alpha +1$ edge deletions are sufficient to extend a state $\bar \boldd_i$, so all violations of $R$ have the same cost $\alpha +1$.
	
	So we have shown that $\calJ'$ $1$-realizes $R$. Note that, according to \cref{lem:NEQtoRel}, the portals of $\calJ$ are non-adjacent. Consequently, so are the portals of $\calJ'$. Thus, we can introduce $\omega$ copies of $\calJ'$ on the same set of portals $Z_1, \ldots, Z_r$ to obtain a gadget that $\omega$-realizes $R$.
\end{proof}

\subsection{\boldmath Hardness for \maxcut}\label{sec:maxcut}\label{sec:ed-maxcut}
We now consider the problem \maxcut, and our goal is to prove \cref{thm:LBmaxcutSETH,thm:LBmaxcutMSH}.
To this end we first show respective results for the list version of \maxcut, that is, we show \cref{thm:LBLHomEDKqSETH,thm:LBLHomEDKqMSH} for the case $q=2$.
We do this in a unified proof that only in the end plugs in two different assumptions based on the SETH or the \msh, respectively.
Afterwards we show how to model the lists using the nonlist version to obtain \cref{thm:LBmaxcutSETH,thm:LBmaxcutMSH}.

\begin{proof}[Proof of \cref{thm:LBLHomEDKqSETH,thm:LBLHomEDKqMSH} for $q=2$]
	Consider an instance $\calI$ of the decision version of $\MaxCSP{$2$}{$r$}$ with $N$ variables $\calV$, clauses $\calC=\{C_1, \ldots, C_m\}$, and clause deletion budget $z\le m$ for the number of clauses that can be violated. For each clause $C_i$ in $\calC$ let $R_i$ be the corresponding relation of arity $r_i\le r$.
	
	\paragraph{An instance of \listcoloringED{2}.}
	We define an instance $(G,L)$ of \listcoloringED{2}. We consider its decision version and also define a corresponding edge deletion budget $z'$. For each variable $v\in \calV$, the graph $G$ contains a vertex $y(v)$ with $L(v) = \{1,2\}$. Let $Y=\{y(v) \mid v\in \calV\}$. Let us consider each $R_i$ as a relation in $[2]^{r_i}$. From \cref{cor:K2relations}, we know that for each relation $R_i$ there is a gadget $\calJ_i$ that $1$-realizes $R_i$. Note that by the definition of realization these gadgets may use lists. Let $\alpha_i$ be the required edge deletions within $\calJ_i$ for a state that satisfies $R_i$, i.e., for each $\boldd\in R_i$ we have $\edcount(\calJ_i, \boldd)=\alpha_i$. (By the definition of realizing, $\alpha_i$ does not depend on $\boldd$.) Since $\calJ_i$ is a $1$-realizer, each state outside of $R_i$ requires $\alpha_i+1$ edge deletions within $\calJ_i$. Also note that we can compute $\alpha_i$ in constant time. For each $i\in [m]$, the graph $G$ contains the gadget $J_i$, where the portals of $J_i$ are those vertices $y(v_1)$, $y(v_2)$, $y(v_3)$ for which $(v_1, v_2, v_3)$ is the scope of $R_i$.  We set $z'\coloneqq z +\sum_{i=1}^m \alpha_i$.
	
	\paragraph{Equivalence of instances.}
	Suppose there is a satisfying assignment $f\from \calV \to \{0,1\}$ that violates at most $z$ clauses from $\calC$. Consider the assignment $h$ that maps $y(v)$ to $1$ whenever $f(v)=1$, and that maps $y(v)$ to $2$ whenever $f(v)=0$. By construction, in order to ensure that $h$ can be extended to a list $2$-coloring, it suffices to make $\alpha_i+1$ edge deletions in each gadget $\calJ_i$ for which $C_i$ is violated by $f$, and $\alpha_i$ edge deletions for each of the remaining gadgets. Thus, it suffices to make a total of $z'$ edge deletions.
	The reverse direction is also straight-forward.

	\paragraph{Structure of the constructed instance.}
	The gadgets $\calJ_i$ as well as the integers $\alpha_i$ depend only on the arity $k$ of the relations.
	Consequently, for fixed $k$, the size of $G$ is $\abs{Y}+\bigO(1)=N^{\bigO(1)}$.
	Let $Q=Y$. Each connected component of $G - Q$ is a subgraph of a gadget $\calJ_i$ and has $k$ neighbors in $Q$ ($k$ portals of $\calJ_i$). Thus, the size of such a component depends only on $k$ and consequently $Q$ is a \core{\sigma}{k} of size $N+1$ for some $\sigma$ depending only on $k$.
	
	\paragraph{Runtime.}
	As the size of $G$ is polynomial in $N$, the hypothetical algorithm for \listcoloringED{2} would require $(2-\eps)^{(N+1)}\cdot N^{\bigO(1)} = (2-\eps)^{N}\cdot N^{\bigO(1)}$ time to solve $\MaxCSP{$2$}{$r$}$ on instance $\calI$.
	Note that $\MaxCSP{$2$}{$r$}$ with zero deletion budget clearly generalizes the decision problem $r$-\textsc{Sat}.
	Thus, the hypothetical algorithm contradicts the SETH for some $r$ depending on $\eps$ (and consequently for some $\sigma$ and $\delta$ depending on $\eps$). Moreover, according to \cref{thm:M3SHtoMaxCSP}, the hypothetical algorithm also contradicts the \msh even for $k=3$ (and consequently for some universal constant $\sigma$ and $\delta=3$) to prove \cref{thm:LBLHomEDKqSETH,thm:LBLHomEDKqMSH} for $q=2$. 
\end{proof}

We will now continue the previous reduction by removing the lists in the constructed instance of \listcoloringED{2} to obtain a result for \maxcut.

\begin{proof}[Proof of \cref{thm:LBmaxcutSETH,thm:LBmaxcutMSH}]
	
	As in the proof of \cref{thm:LBLHomEDKqSETH,thm:LBLHomEDKqMSH} for $q=2$, consider an instance $\calI$ of the decision version of $\MaxCSP{$2$}{$r$}$ with $N$ variables $\calV$, and then consider the equivalent instance $(G,L,z')$ of the decision version of \listcoloringED{2}.
	
	\paragraph{Constructing the instance of \maxcut.}
	Let us adjust the instance $(G,L,z')$ to an equivalent instance $(G^*, \abs{E(G^*)}-z')$ of the decision version of \maxcut. Since all vertices in $Y$ have list $\{1,2\}$, each vertex with list $\{1\}$ or list $\{2\}$ is part of some gadget $\calJ_i$.
	To construct the graph $G^*$, we add a new vertex $A$ to the graph $G$.
	For each vertex $v$ of $G$ that has a list $L(v)=\{1\}$, we introduce a set of $\alpha_i+\omega+1$ (parallel) $3$-vertex paths from $v$ to $A$ (each such path contains the vertex $v$, one inner vertex, and the vertex $A$). Similarly, for each vertex $v$ of $G$ in $\calJ_i$ that has a list $L(v)=\{2\}$, we introduce $\alpha_i+2$ (parallel) $4$-vertex paths from $v$ to $A$.
	This finishes the definition of $G^*$.
	
\paragraph{Equivalence of instances.}	
	To show that the two instances are equivalent, first assume that there is a set $X$ of at most $z'$ edges such that there is a $2$-coloring $\phi$ of $G\setminus X$ that respects the lists in $L$. Then $\phi$ can be extended to a $2$-coloring of $G^*\setminus X$ with $\phi(A)=1$. This ensures that all of the introduced $3$-vertex paths have endpoints that are mapped to $1$, and all $4$-vertex paths have endpoints that are mapped to different colors --- and consequently all of these paths can be $2$-colored without further edge deletions. If we interpret the two color classes as the two different parts of a cut, this gives a cut of size $\abs{E(G^*)}-\abs{X}\ge \abs{E(G^*)}-z'$.
	
	In the opposite direction, suppose that there is a cut of size at least $\abs{E(G^*)}-z'$. Let $X$ be the at most $z'$ non-cut edges.  Then there exists a $2$-coloring $\phi$ of $G^*\setminus X$ (mapping vertices on one side of the cut to $1$, and the other side to $2$).
	Without loss of generality, by renaming the colors, we can assume that $h(A)=1$.
	Suppose that for some vertex $v$, $\phi(v)\notin L(v)$, say $\phi(v)=2$ but $L(v)=\{1\}$ (the other case is analogous). As $L(v)\neq \{1,2\}$ the vertex $v$ is part of some gadget $\calJ_i$. Consequently, $X$ must contain at least one edge from each of the $3$-vertex-paths connecting $v$ and $A$. This gives a total of $\alpha_i+2$ such edge deletions.  
	However, it only requires a set $F$ of at most $\alpha_i+1$ edge deletions on $\calJ_i$ to ensure that the mapping $\phi(Y)$ can be extended to a $2$-coloring of $\calJ_i\setminus F$. Moreover, in this case $\phi$ satisfies the lists $L$ on the vertices of $\calJ_i$, which means that any $4$-vertex paths that connects some vertex $v$ of $\calJ_i$ to $A$ can be $2$-colored without further edge deletions. Summarizing, this shows that there exists a $2$-coloring that respects the lists $L$ and requires at most $z'$ edge deletions.
	
	\paragraph{Structure of the constructed instance.}
	Recall that the gadgets $\calJ_i$ as well as the integers $\alpha_i$ depend only on the arity $k$ of the relations of $\calI$.
	Consequently, for fixed $k$, the size of $G^*$ is $\abs{Y}+\bigO(1)=N^{\bigO(1)}$.
	Let $Q=Y\cup \{A\}$. Each connected component of $G^* - Q$ is a subgraph of a gadget $\calJ_i$ together with at most $\alpha_i+2$ pending paths per vertex of $\calJ_i$. Note that each of these components has at most $k+1$ neighbors in $Q$ ($k$ portals of $\calJ_i$ plus the vertex $A$). Thus, the size of such a component depends only on $k$ and consequently $Q$ is a \core{\sigma}{k+1} of size $N+1$ for some $\sigma$ depending only on $k$.
	
	\paragraph{Runtime.}
	As the size of $G^*$ is polynomial in $N$, the hypothetical algorithm for \maxcut would require $(2-\eps)^{(N+1)}\cdot N^{\bigO(1)} = (2-\eps)^{N}\cdot N^{\bigO(1)}$ time to solve $\MaxCSP{$2$}{$r$}$ on instance $\calI$.
	Recall that $\MaxCSP{$2$}{$r$}$ with zero deletion budget generalizes the decision problem $r$-\textsc{Sat}.
	Thus, the hypothetical algorithm contradicts the SETH for some $r$ depending on $\eps$ (and consequently for some $\sigma$ and $\delta$ depending on $\eps$). Moreover, according to \cref{thm:M3SHtoMaxCSP}, the hypothetical algorithm also contradicts the \msh even for $k=3$, which implies that $Q$ as defined in the previous paragraph is a \core{\sigma}{4} for constant $\sigma$.
	This proves \cref{thm:LBmaxcutSETH,thm:LBmaxcutMSH}.
\end{proof}

\subsection{\boldmath Hardness for \coloringED{q}}\label{sec:colED}\label{sec:ed-col}
Now we consider the problem \coloringED{q} for $q\ge 3$, and our goal is to prove \cref{thm:LBcolEDMSH}.
Again, we first show respective results for the list coloring, that is, now we show \cref{thm:LBLHomEDKqSETH,thm:LBLHomEDKqMSH} for the case $q\ge 3$, which also finishes the proofs of these two results.
As before (but using different gadgets), we show how to remove the lists to obtain \cref{thm:LBcolEDMSH}, as desired.

\begin{proof}[Proof of \cref{thm:LBLHomEDKqSETH,thm:LBLHomEDKqMSH} for $q\ge3$]
	 
	With \cref{thm:M3SHtoMaxCSP} in mind, for some fixed $r\ge 3$, let $\calI=(\calV, \calC)$ be an instance of the decision version of \MaxCSP{$q$}{$r$} with $N$ variables and deletion budget $z\le \abs{\calC}$. 
	
	\paragraph{\boldmath An instance of \listcoloringED{q}.}
	We now define an instance $(G,L, z')$ of the decision version of \listcoloringED{q}.
	For each variable $v\in \calV$ we introduce a variable vertex $y(v)$. Let $Y=\{y(v) \mid v\in \calV\}$.
	For a constraint in $\calC$, let $R\subseteq [q]^\ell$ be the corresponding relation of arity $\ell\le r$, and let $(v_1, v_2, \ldots, v_\ell)\in \calV^\ell$ be the scope of $R$.

	According to \cref{prop:realize-coloring}, for every $\boldd = (d_1,d_2,\ldots,d_\ell)\in [q]^\ell$,
	 there is an $\ell$-ary gadget $\calJ_\boldd=(J,L,\{z_1,z_2,\ldots, z_\ell\})$ with lists $L$ contained in $[q]$ such that any $q$-coloring $\psi$ of the vertices $z_1, z_2, \ldots, z_\ell$ can be extended to a list $q$-coloring of $\calJ_\boldd$ if and only if $(\psi(z_1), \psi(z_2), \ldots, \psi(z_\ell))\neq \boldd$. Let $\gamma_\boldd\ge 1$ be the minimum number of edge deletions in $\calJ_\boldd$
	 that are required to extend $\boldd$ to a proper list coloring of this gadget; note that $\gamma_\boldd$ can be computed in constant time as the number of vertices in $J$ is bounded by a function of $\ell$ and $q$, i.e., a constant.

	Let $P\coloneqq \prod_{\ell=1}^{r}\prod_{\boldp\in [q]^\ell} \gamma_\boldp$. Note that $P$ depends only on $q$ and $r$.
	For each vector $\boldd$ of $\ell\le r$ elements from $[q]$ with $\boldd\notin R$, we introduce $P/\gamma_\boldd$ copies of $\calJ_\boldd$ to balance out the different violation costs for different $\boldd$ (and different $R$), and we identify the respective portals $z_1, z_2, \ldots, z_\ell$ of each of these copies with the variable vertices $y(v_1),y(v_2),\ldots,y(v_\ell)$. Let $\calJ_R$ be the union of the copies of all the gadgets $\calJ_\boldd$ that we introduced for this constraint (over all vectors $\boldd$).
	We repeat this process for each constraint in $\calC$. This forms the graph $G$ together with the lists $L$.
	Finally, we set $z'=P\cdot z$.
	
	\paragraph{Equivalence of instances.}
	Suppose there is a set of constraints $\calC'\subseteq \calC$ with $\abs{\calC'}\le z$ such that there is a satisfying assignment $f\from \calV \to [q]$ for $(\calV, \calC\setminus \calC')$. Consider the coloring $\phi$ of $Y$ with $\phi(y(v))=f(v)$ for each $v\in \calV$.
	Consider a constraint in $\calC$ with the corresponding relation $R$ and scope $(v_1, v_2, \ldots, v_\ell)\in \calV^\ell$. If $f$ satisfies $\calC$ then $\phi$ can be extended to a proper list coloring of all the gadgets in $\calJ_R$ (with zero required additional edge deletions). However, if $f$ violates $\calC$ then $\boldd\coloneqq (\phi(v_1), \phi(v_2),\ldots, \phi(v_\ell))$ is not in $R$ and consequently extending $\phi$ to a proper list coloring of $\calJ_\boldd$ requires $\gamma_\boldd$ edge deletions for each copy of $\calJ_\boldd$, for a total of $P$ required edge deletions.
	Thus, $\phi$ can be extended to a proper list coloring of $G\setminus X$ where $X$ is a set of $P\cdot \abs{\calC'}\le P\cdot z= z'$ edges (where we use the fact that $P$ is independent of $R$).
	
	Now suppose there is a minimum-size set of edges $X$ with $\abs{X}\le z'$ such that $G\setminus X$ has a proper list $q$-coloring. Assume that $\phi$ is such a coloring. Consider some gadget $\calJ_\boldd$ that is part of $G$. Let $\boldd=(d_1, d_2,\ldots, d_\ell)$ and without loss of generality let $(y(v_1), y(v_2), \ldots, y(v_\ell))$ be the portals of $\calJ_\boldd$. If $(\phi(y(v_1)), \phi(y(v_2)),\ldots, \phi(y(v_\ell)))\neq \boldd$ then $X$ does not contain any edges of $\calJ_\boldd$ (because of the minimality of $X$), whereas otherwise $X$ has to contain precisely $\gamma_\boldd$ edges from $\calJ_\boldd$, and this holds for each of the $P/\gamma_\boldd$ copies of $\calJ_\boldd$. Thus, $\abs{X}$ is a multiple of $P$ and the assignment $f$ with $f(v)\coloneqq \phi(y(v))$ satisfies all but $\abs{X}/P\le z$ constraints in $\calC$.

	\paragraph{Structure of the constructed instance.}
	The size of the gadgets $\calJ_\boldd$ and the violation costs $\gamma_\boldd$ are upper-bounded by some function in $q$ and $r$. Consequently, the size of at most $P$ copies of the gadget $\calJ_\boldd$ is upper-bounded by some function $h(q,r)$. We can assume that no two constraints have scopes of variables such that one scope completely contains the other. So we can assume that there are at most $(N+1)^r$ different constraints in $\calC$. Consequently, the size of $G$ is at most $\abs{Y}+ (N+1)^r\cdot h(q,r)=N^{\bigO(1)}$.
	Let $Q=Y$. Then each component of $G-Q$ is a subgraph of a gadget of the form $\calJ_\boldd$. Hence the size of such a component depends only on $q$ and $r$, and each component has at most $r$ neighbors in $Q$ (the portals of $\calJ_\boldd$).
	Therefore, the set $S$ is a \core{\sigma}{r} of size $N$ for some $\sigma$ depending only on $q$ and $r$.

	\paragraph{Runtime.}
	As the size of $G$ is polynomial in $N$, the hypothetical algorithm for \listcoloringED{q} would require $(q-\eps)^{N}\cdot N^{\bigO(1)}$ time to solve $\MaxCSP{$q$}{$r$}$ on instance $\calI$.	
	This contradicts the SETH for some $r$ depending on $\eps$ according to \cref{thm:csp}. It also contradicts the \msh even for $r=3$ according to \cref{thm:M3SHtoMaxCSP}, which means that the set $Q$ as defined in the previous paragraph is a \core{\sigma}{3}. 
	This proves \cref{thm:LBLHomEDKqSETH,thm:LBLHomEDKqMSH} for $q\ge 3$.
\end{proof}

We will now continue the previous reduction by removing the lists in the constructed instance of \listcoloringED{q}.

\begin{proof}[Proof of \cref{thm:LBcolEDMSH}]
For $q\ge 3$, consider an instance $\calI=(\calV, \calC)$ of the decision version of \MaxCSP{$q$}{$r$} with $N$ variables $\calV$ and deletion budget $z$, and then consider the equivalent instance $(G,L,z')$ of \listcoloringED{q} constructed in the proof of \cref{thm:LBLHomEDKqSETH,thm:LBLHomEDKqMSH} for $q\ge 3$.
	
\paragraph{Constructing the instance of \coloringED{q}.}
We modify the instance $(G,L,z')$ into an equivalent instance $(G^*,z')$ of $\coloringED{q}$.
Note that for each vertex in $Y$ the corresponding list is $[q]$, so we only have to take care of lists of vertices that are non-portal vertices of some gadget of the form $\calJ_\boldd$, for some $\boldd=(d_1, d_2,\ldots, d_\ell)\in [q]^\ell$. Let $\gamma_\boldd\ge 1$ be the minimum number of edge deletions in $\calJ_\boldd$ that are required to extend the state $\boldd$ of the portals to a proper list coloring of this gadget. 
	
The high-level idea is to simulate lists using a copy of the $q$-clique, similarly as in the proof of \cref{thm:vd-coloring-intro}.
However, again we need to be careful as deleting edges incident to this clique (either inside the clique or joining the clique with vertices of $G$) could destroy the structure of solutions.
	
We first introduce $q$ new vertices $k_1, \ldots, k_q$. They will be used as a ``global copy of a $q$-clique''
Now we would like to turn $\{k_1,\ldots,k_q\}$ into a clique whose edges cannot be deleted.
As we cannot do it directly, instead of edges we will introduce many copies of gadgets that behave ``like an edge.''
More specifically, for each distinct $i,j \in [q]$, where $i < j$, we introduce $z'+1 $ copies of the following \emph{inequality gadget}.
We introduce a new $(q-1)$-clique $K_{i,j}$ and make it fully adjacent to $k_i$ (i.e., the set $V(K_{i,j}) \cup \{k_i\}$ induces a $q$-clique in $G^*$). Next, we introduce a new vertex $x_{i,j}$ and make if fully adjacent to $K_{i,j}$ (i.e., the set $V(K_{i,j}) \cup \{x_{i,j}\}$ induces a $q$-clique in $G^*$). Note that this enforces that in any proper $q$-coloring, the color of $x_{i,j}$ is the same as the color of $k_i$, and no further constraints are introduced. Finally, we make $x_{i,j}$ adjacent to $k_j$, so that the color of $k_j$ must be different than the color of $k_i$.

\smallskip
Now consider some $\boldd=(d_1, d_2,\ldots, d_\ell)\in [q]^\ell$ and a copy of a gadget $\calJ_\boldd$.
We introduce $\gamma_\boldd+1$ ``private'' $q$-cliques $C^p_\boldd$ with vertex labels $c^p_1, \ldots, c^p_q$ (for $p\in [\gamma_\boldd+1]$) . For $i\in [\ell]$, and for each $p\in [\gamma_\boldd +1]$, we introduce an edge between $k_{d_i}$ and each vertex of $C^p_\boldd$ with the exception of $c^p_{d_i}$.
	Then, for each vertex $v$ of $\calJ_\boldd$ and $j\notin L(v)$ and for each $p\in [\gamma_\boldd+1]$, we introduce an edge between $v$ and $c^p_j$. This completes the construction of $G^*$.
	
\paragraph{Equivalence of instances.}	
Now let us show that the two instances are indeed equivalent. Suppose there is a proper list coloring $\psi$ of $G\setminus X$ for some set of edges $X$ from $G$.
We claim that there is an extension $\phi$ of $\psi$ that is a proper $q$-coloring of $G^*\setminus X$.
To this end, we need to show how to extend $\psi$ to the cliques of the form $C^p_\boldd$, to the vertices $k_1, \ldots, k_q$, and to inequality gadgets.
For a clique of the form $C^p_\boldd$ we color the corresponding vertices $c^p_1, \ldots, c^p_q$ by setting $\phi(c^p_i)=i$.
For each $i\in [q]$, we set $\phi(k_i)=i$.
For each $i,j\in [q]$ and each copy of the inequality gadget, we color $K_{i,j}$ arbitrarily using colors $[q]\setminus \{i\}$ and $x_{i,j}$ receives the color $i$.
Observe that a vertex $k_i$ is never adjacent to a vertex $c^p_i$ with the same index $i$. Moreover, a vertex $c^p_i$ is only adjacent to some vertex $v$ in $\calJ_\boldd$ if $i\notin L(v)$.
Hence $\phi$ is a proper coloring of $G^* \setminus X$ if $\psi$ is a proper list coloring of $G\setminus X$.

\smallskip	
For the other direction let $\phi$ be a proper $q$-coloring of $G^*\setminus X$ for some minimum-size set of edges $X$ from $G^*$ with $\abs{X}\le z'$.
Note that since for each pair of $i,j \in [q]$ we introduced $z'+1$ copies of the inequality gadget between $k_i$ and $k_j$,
there is always one copy that is not affected by $X$. Thus we can assume that vertices $k_1,\ldots,k_q$ receive pairwise distinct colors.
By permuting colors we can assume that for each $i \in [q]$ we have $\phi(k_i)=i$.

Now consider some $\boldd = (d_1,\ldots,d_\ell)$ and a copy of $J_\boldd=(J,L, \{z_1, \ldots, z_\ell\})$.
Let $X'$ denote those edges from $X$ that have at least one endvertex in $J$ or in some clique of the form $C^p_{\boldd}$ introduced for $J_\boldd$. We aim to understand where $X'$ lies.

\begin{clm}\label{clm:recoloringJd}
Let $\pi$ be some permutation of $[q]$ and let $\calJ_\boldd^\pi=(J,L^\pi, \{z_1, \ldots, z_\ell\})$, where for each vertex $v$ of $J$, $L^\pi(v)\coloneqq \{\pi(i) \mid i\in L(v)\}$. Let $\psi \from \{z_1, \ldots, z_\ell\} \to [q]$ be a coloring of $\boldz=\{z_1, \ldots, z_\ell\}$. Then 
\begin{myitemize}
	\item If $\psi(\boldz) \neq \pi(\boldd)$ then $\psi$ can be extended to a proper list coloring of $\calJ_\boldd^\pi$.
	\item  If $\psi(\boldz) = \pi(\boldd)$ then $\gamma_\boldd$ edge deletions are required and sufficient to extend $\psi$ to a proper list coloring of $\calJ_\boldd^\pi$.
\end{myitemize}
\end{clm}
\begin{claimproof}
This follows directly from the properties of $\calJ_\boldd$ according to \cref{prop:realize-coloring}, if to each coloring of $\calJ_\boldd$ we apply the permutation $\pi$.
\end{claimproof}

Note that \cref{clm:recoloringJd} implies that regardless of the coloring of $z_1,\ldots,z_\ell$, we can make the copy of $J_\boldd$ $q$-colorable by removing at most $\gamma_\boldd$ edges from this gadget (i.e., with both endvertices in $J$).
In particular, this means that we can safely assume that $|X'| \leq \gamma_\boldd$ (otherwise we can obtain another, better solution).
Since we introduced $\gamma_\boldd+1$ local cliques $C^p_{\boldd}$ for $J_{\boldd}$, there is always at least one such clique whose vertices are not incident with any edge from $X'$.
Since all local cliques have the same neighborhood, we can safely assume that the edges from $X'$ are not incident to any of these cliques. In particular, all edges from $X'$ have both endvertices in $J$, i.e., they are present in $G$.

Furthermore, again, without loss of generality, we can assume that the corresponding vertices of each clique $C^p_{\boldd}$ receive the same color in $\phi$, i.e., for all $p,p' \in [\gamma_\boldd]$ and $i \in [q]$ we have $\phi(c^p_i) = \phi(c^{p'}_i)$.
Since the colors of vertices of $C^p_{\boldd}$ are pairwise distinct, we conclude that there is a permutation $\pi$ of $[q]$ such that $\phi(c^p_i) = \pi(i)$, for every $i,p$.
Summing up, we observe that the lists enforced by the local cliques on the vertices of $J_\boldd$ are $L^\pi$, where the definition of $L^\pi$ is as in \cref{clm:recoloringJd}. 
Note that this implies that without loss of generality we may assume that $X$ contains either 0 (if $\phi(\boldz) \neq \pi(\boldd)$) or $\gamma_\boldd$ (if $\phi(\boldz) = \pi(\boldd)$) edges from $J$ -- otherwise we could obtain a better solution.
Furthermore, observe that the edges joining $\{k_1,\ldots,k_q\}$ with vertices from the local cliques imply that for each $d_i$ we have $\pi(d_i)=d_i$. 
Consequently, $|X \cap E(J)| = 0$ if $\phi(\boldz) = \boldd$ and  $|X \cap E(J)| =\gamma_\boldd$ if $\phi(\boldz) = \boldd$.

But now, by the definition of $\calJ_\boldd$ and $\gamma_\boldd$,
we can remove $|X \cap E(J)|$ edges from $J$ and properly color all non-portal vertices in a way that this coloring respects lists $L$.
Repeating this argument for each copy of $J_{\boldd}$, we conclude that $G\setminus X$ admits a proper $q$-coloring that respects lists $L$.

\paragraph{Structure of the constructed instance.}
The size of at most $P$ copies of the gadget $\calJ_\boldd$ together with $P$ copies of a clique of size $q$ is upper-bounded by some function $h(q,r)$. We can assume that there are at most $(N+1)^r$ different constraints in $\calC$. Further, note that $w=z'+1=P\cdot z+1\le P\cdot (N+1)^r\in N^{\bigO(1)}$. Consequently, the size of $G^*$ is at most $\abs{Y}+ (N+1)^r\cdot h(q,r) + (w+1)\cdot q^3=N^{\bigO(1)}$.
Let $Q=Y\cup \{k_1, \ldots, k_q\}$. Then each component of $G^*-Q$ is of one of two forms.
The first possibility is that it is a subgraph of a gadget of the form $\calJ_\boldd$ together with the corresponding cliques $C^1_\boldd, \ldots, C^{\gamma_\boldd+1}_\boldd$, in which case the size of such a component depends only on $q$ and $r$, and the component has at most $3+r$ neighbors in $Q$,  three in $Y$, and $r$ in $\{k_1, \ldots, k_q\}$.
The second possibility is that such a component is an inequality gadget, i.e., it consists of $q$ vertices and attaches to two vertices of $Q$.
Therefore, the set $Q$ is a \core{\sigma}{3+r} of size $N+q$ for some $\sigma$ depending only on $q$.

\paragraph{Runtime.}
As the size of $G^*$ is polynomial in $N$, the hypothetical algorithm for $\coloringED{q}$ would require $(q-\eps)^{(N+q)}\cdot N^{\bigO(1)} = (q-\eps)^{N}\cdot N^{\bigO(1)}$ time to solve $\MaxCSP{$q$}{$r$}$ on instance $\calI$. This contradicts the \msh even for $r=3$ according to \cref{thm:M3SHtoMaxCSP}, which implies that the set $Q$ defined in th eprevious paragraph is a \core{\sigma}{6}.  This proves \cref{thm:LBcolEDMSH}.	
\end{proof}

\subsection{\boldmath Simple Algorithm for \coloringED{q}}\label{sec:ed-basic-algo}
The following algorithm solves the \coloringED{q} problem for $q \geq 2$, which also includes the \textsc{Max Cut} problem for $q = 2$. It is worth noting that this algorithm can be modified to use treewidth as the parameter.

\begin{thm}
	For all integers $q \geq 2$ and $\sigma, \delta \geq 1$, every $n$-vertex instance of \coloringED{q} can be solved in time $q^{p} \cdot n^{\bigO(1)}$ if a \core{\sigma}{\delta} of size $p$ is given in the input.
\end{thm}

\begin{proof}
	As input to \coloringED{q}, consider a graph $G$ with $n$ vertices, given along with a \core{\sigma}{\delta} $Q$ of size $p$.

	We start by guessing the coloring of $Q$. Let $f \colon Q \to [q]$ denote this coloring. Consider now a component $G'$ of $G \setminus Q$, and let $O$ denote the graph induced by the vertices of $G'$ and their neighbors in $Q$. Since $G'$ has constant size and is adjacent to constantly many vertices in $Q$, we can determine in constant time the minimum number of edges to delete from $O$, such that the remaining graph admits a $q$-coloring.

	The number of components in $G \setminus Q$ is bounded above by a polynomial in $n$, and therefore, in polynomial time, we can find the optimum solution that agrees with $X$ and $f$. Moreover, since there is always an $(X,f)$ corresponding to the optimal solution, an algorithm that tries all possible $X$ and $f$ gives the correct output. The running time of such an algorithm is upper bounded by $q^{p} \cdot n^{\bigO(1)}$.

\end{proof}

\section{\boldmath \listcoloring{q} and \listcoloringVD{q} with gadgets of constant degree faster than brute force}\label{sec:vd-beatbruteforce}
Recall that, assuming the SETH, we proved lower bounds for \coloring{q}, \coloringVD{q}, and \coloringED{q},
parameterized by the size of a \core{\sigma}{\delta} of the instance graph $G$, where $\sigma$ and $\delta$ are constants depending on $\epsilon>0$.
In all cases, we were able to exclude an algorithm with running time $(f(q)-\epsilon)^\cpar \cdot n^{\bigO(1)}$, for any $\epsilon>0$,  where $f(q)$ is an function of $q$ (i.e., either $q$ or $q+1$), $\cpar$ is the size of a \core{\sigma}{\delta}, and $n$ is the number of vertices of $G$.
Furthermore, assuming the \msh, we have a somewhat stronger lower bound for \coloringED{q} where $\sigma$ and $\delta$ are absolute (and small) constants independent of $\epsilon$. For \coloring{q} and \coloringVD{q}, we have only a weaker form of the lower bound, where the values of $\sigma$ and $\delta$ depend on $\epsilon$.

In this section, we show that these latter lower bounds cannot be strengthened so that $\sigma$ and $\delta$ are absolute constants; in fact, we prove that already $\delta$ on its own cannot be an absolute constant. 
Specifically, we prove the following algorithmic results that hold even for the respective coloring problems with lists.

\begin{thm} \label{thm:lcoloring-better-alg}
For every $q \geq 3$ and every constant $\delta$ there exists $\epsilon$ with the following property: 
For every constant $\sigma$, every instance $(G,L)$ of \listcoloring{q} with $n$ vertices, given with a \core{\sigma}{\delta} of size $\cpar$, can be solved in time $(q-\epsilon)^\cpar \cdot n^{\bigO(1)}$.
\end{thm}
\begin{thm} \label{thm:vd-better-alg}
For every $q \geq 3$ and every constant $\delta$ there exists $\epsilon$ with the following property: 
For every constant $\sigma$, every instance $(G,L)$ of \listcoloringVD{q} with $n$ vertices, given with a \core{\sigma}{\delta} of size $\cpar$, can be solved in time $(q+1-\epsilon)^\cpar \cdot n^{\bigO(1)}$.
\end{thm}

\subsection{Faster algorithm for \listcoloring{q}}

As a warm-up, let us start with \cref{thm:lcoloring-better-alg}, whose proof is quite straightforward.

\begin{proof}[Proof of \cref{thm:lcoloring-better-alg}.]
Let $X$ be a \core{\sigma}{\delta} of $G$.
If $X$ is a constant-size set, we can guess its coloring exhaustively, and then adjust the lists of the neighbors of vertices in $X$ as follows.
If $x \in X$ is guessed to be colored with color $i$, then $i$ can be removed from lists of all vertices in the neighborhood of $X$. After that we can safely remove $X$ from the graph, obtaining an equivalent instance, which can be solved in polynomial time as its every component is of constant size. Thus from now on assume that $|X|=p$ is sufficiently large.

We first check if there is a vertex with an empty list; if so, we immediately reject the instance.
Then we check whether there exists a component $A$ of $G-X$ with no neighbor in $X$.
If such an $A$ exist, we solve the instance $G[A]$ of \listcoloring{q}: it can be done in constant time as $|A| \leq \sigma$.
If it is a yes-instance, then we can proceed with the equivalent instance $G-A$; it is a no-instance, then $(G,L)$ is a no-instance and we reject it.

Consider a component $A$ of $G - X$, and let $\nh(A)$ denote its neighborhood in $X$; recall that we can assume that $\nh(A) \neq \emptyset$.
Let $r = |\nh(A)|$ and recall that $1 \leq r \leq \delta$.

Let $\mathcal{F}$ be the set of all colorings of $\nh(A)$ with $q$ colors, respecting lists $L$, that can be extended to a proper list coloring of $G[A \cup \nh(A)]$. Note that $|\mathcal{F}| \leq q^r$ and $\mathcal{F}$ can be computed in constant time as $|A \cup \nh(A)| \leq \delta	+ \sigma$. If $|\mathcal{F}| = q^r$, then \emph{every} coloring of $\nh(A)$ can be extended to the vertices of $A$. Thus we can remove $A$ from $G$, obtaining an equivalent instance.
So assume that $|\mathcal{F}|  \leq  q^r -1$.

We exhaustively guess a coloring of $\nh(A)$ that belongs to $\mathcal{F}$. 
In each branch we adjust the lists of the neighbors of vertices in $\nh(A)$ as previously, and then remove $A \cup \nh(A)$ from the graph, obtaining an equivalent instance.

The number of leaves in the recursion tree, measured as the function of $p$, is upper-bounded by the recursive formula
\[
	T(p) \leq (q^r-1) \cdot T(p-r).
\]
Proceeding by induction we can show that
\[
T(p) \leq (q^r-1)^{p/r} \cdot \left(q^\delta-1 \right)^{(p-r)/\delta} \leq \left(q^\delta-1 \right)^{p/\delta},
\]
where the last inequality can be verified by a standard yet quite tedious calculation.

As every node in the recursion tree is processed in polynomial time, the total running time is bounded by $T(p) \cdot n^{\Oh(1)} \leq \left(q^\delta-1 \right)^{p/\delta} \cdot n^{\bigO(1)}$.
Now \cref{thm:lcoloring-better-alg} follows by observing that $\left(q^\delta-1 \right)^{1/\delta}< q$.
\end{proof}

\subsection{Faster algorithm for \listcoloringVD{q}}

Now let us proceed to the \listcoloringVD{q} problem. 
It will be more convenient to reduce it to a certain auxiliary variant of \textsc{CSP}.
Before we formally define it, let us introduce some more notation.

Let $X$ be a set and let $q$ be a positive integer.
By $\setf{X}$ we denote the set of all functions $f \from X \to [q] \cup \{\del\}$, where $\del$ is a special symbol.
For a function $f \in \setf{X}$, we let $||f||$ denote $|f^{-1}(\del)|$.

We define a binary relation $\preceq$ on $\setf{X}$ as follows:
for all $f, f' \in \setf{X}$, we say that $f' \preceq f$ if for all $x \in X$ either $f(x)=f'(x)$, or $f(x)=\del$.
In other words, $f$ was obtained from $f'$ by mapping some (possibly empty) subset of elements of $X$ to $\del$.
Note that $(\setf{X}, \preceq)$ is a partially ordered set and its unique maximum element is $f^{\del}_{X}$, where $f^{\del}_{X} \in \setf{X}$ is the function that maps every element of $X$ to $\del$.
We write $f' \prec f$ if $f' \preceq f$ and $f' \neq f$.

For constants $q,r$, an instance of $(q,r)$-\textsc{CSP-with-Wildcard} is a triple $(\calV,\calC,\cost)$, where
\begin{myitemize}
\item $\calV$ is the set of variables, each with domain $[q] \cup \{\del\}$,
\item $\calC$ is a multiset of subsets of $\calV$, each of size at most $r$,
\item $\cost$ is a function that maps each pair $(C,f)$ where $C \in \calC, f \in \setf{C}$ to a non-negative integer and satisfies the following:

\smallskip
\begin{tabular}{ll}
	\textbf{(wildcard property)} & For all $C \in \calC$ and all $f,f' \in \setf{C}$, if $f' \preceq f$,\\
& then it holds that $\cost(C,f) \leq \cost(C,f')$.
\end{tabular}
\smallskip

\end{myitemize}

The aim of $(q,r)$-\textsc{CSP-with-Wildcard} is to find an assignment $\varphi \from \calV \to [q] \cup \{\del\}$,
minimizing the total cost defined as follows:
\[
\mathsf{total\text{-}cost}(\varphi)=||\varphi|| + \sum_{C \in \calC} \cost(C,\varphi|_C).
\]
It is fairly straightforward to reduce an instance of \listcoloringVD{q} with a \core{\sigma}{\delta} of size $\cpar$ 
to an instance of $(q,\delta)$-\textsc{CSP-with-Wildcard} with $\cpar$ variables.
\begin{lem}\label{lem:reduce-to-csp}
Let $\sigma,\delta$ be constants.
For an instance $(G,L)$ of \listcoloringVD{q}, given with a \core{\sigma}{\delta} of size $\cpar$,
in polynomial time we can compute an instance $(\calV,\calC,\cost)$ of $(q,\delta)$-\textsc{CSP-with-Wildcard}
with $\cpar$ variables, such that the values of optimum solutions of both instances are equal.
\end{lem}
\begin{proof}
Without loss of generality assume that $G$ is connected.
Let $X$ be a \core{\sigma}{\delta} of $G$, and let us enumerate the vertices of $X$ as $\{x_1,x_2,\ldots,x_\cpar\}$.
Recall that we can safely assume that for each $x_j$, the set $L(x_j)$ is non-empty, and clearly $|L(x_j)|\leq q$.
For each $j \in [\cpar]$, let $\rho_j$ be an arbitrary surjective function from $[q]$ to $L(x_j)$.

The set $\calV$ has $\cpar$ elements $v_1,\ldots,v_\cpar$, each with domain $[q] \cup \{\del\}$.
For each $i \in [\cpar]$, the variable $v_i$ represents $x_i$.
Assigning the value $j \in [q]$ to $v_i$ corresponds to mapping $x_i$ to $\rho_i(j)$, while assigning $\del$ to $v_i$ corresponds to deleting $x_i$.

Now consider a component $A$ of $G - X$, and let $\nh(A)$ denote its neighborhood in $X$.
Recall that $|\nh(A)| \leq \delta$.
We add the set $C = \{v_i ~|~ x_i \in \nh(A)\}$ to $\calC$.
Finally, for every $f \in \setf{C}$, we set $\cost(C,f)$ to be the minimum number of vertices from $A$
that need to be deleted so that the mapping of $\nh(A)$ corresponding to $f$ (i.e., if $f(v_i) = j \in [q]$, then $x_i$ is mapped to $\rho_i(j)$, and if $f(v_i) = \del$, then $x_i$ is deleted) can be extended to a proper list coloring of the subgraph of $G$ induced by the non-deleted vertices from $A \cup \nh(A)$.

It is possible that there is more than one component of $G-X$ with exactly the same neighborhood in $X$.
In this situation $\calC$ contains multiple copies of $C$ (recall that it is defined to be a multiset) and the values of the cost function of each copy correspond to a distinct component of $G-X$.

It is clear that the created instance of $(q,\delta)$-\textsc{CSP-with-Wildcard} is equivalent to the original instance $(G,L)$ of \listcoloringVD{q},
i.e., the optimal solutions have the same value.

Finally, let us argue that $(\calV, \calC, \cost)$ can indeed be computed in polynomial time.
Observe that the number of components of $G-X$ is $\bigO(|V(G)|)$.
Furthermore, for each such component $A$ there are at most $(q+1)^\delta$ assignments of variables corresponding the vertices of $\nh(A)$.
Finally, for each component $A$ and each assignment $f$ we can compute the minimum cost of extending $f$ to $G[A \cup \nh(A)]$ in constant time,
as the size of $A$ is at most $\sigma$, i.e., a constant.
\end{proof}

Clearly a brute-force approach to solving an instance $(q,r)$-\textsc{CSP-with-Wildcard} with $n$ variables requires time $(q+1)^n \cdot n^{\bigO(1)}$. Indeed, we enumerate all possible assignments $\varphi \from \calV \to [q] \cup \{\del\}$, for each of them we compute $\mathsf{total\text{-}cost}(\varphi)$, and we return the assignment whose total cost is minimum.
In the next theorem we show that actually we can solve the problem faster.

\begin{thm}\label{thm:solve-csp}
Let $q,r$ be constants.
Any $n$-variable instance of $(q,r)$-\textsc{CSP-with-Wildcard} can be solved in time $\left((q+1)^r-1 \right)^{n/r} \cdot n^{\bigO(1)}$.
\end{thm}
\begin{proof}
Let $(\calV,\calC,\cost)$ denote an instance of $(q,r)$-\textsc{CSP-with-Wildcard}.
Consider $C \in \calC$. Recall that $f^{\del}_{C} \in \setf{C}$ is the function mapping every element of $C$ to $\del$.
By the wildcard property we observe that $\cost(C,f_{C}^{\del}) = \min_{f \in \setf{C}} \cost(C,f)$.
For brevity, we denote $\cost(C,f_{C}^{\del})$ by $\cdel(C)$.

\begin{clm}\label{clm:constr-oneof}
For each $C \in \calC$ one of the following holds:
\begin{myenumerate}[a)]
\item for all $f \in \setf{C}$ it holds that $\cost(C,f) = \cdel(C)$, or
\item there exist $f,f' \in \setf{C}$ such that $f' \prec f$ and $\cost(C,f) + ||f|| \leq \cost(C,f') + ||f'||$.
\end{myenumerate}
\end{clm}
\begin{claimproof}
	Suppose that a) does not hold and let $f' \in \setf{C}$ be such that $\cost(C,f') \geq \cdel(C)+1$
and $||f'||$ is maximum possible.
Note that by the definition of $\cdel(C)$ we know that $f' \neq f_C^{\del}$, i.e., there exists $v \in C$ such that $f'(v) \neq \del$.
Let $f$ be obtained from $f'$ by mapping $v$ to $\del$. Note that $f' \prec f$.
Clearly $||f|| = ||f'||+1$ and, by the maximality of $f'$, we have $\cost(C,f)=\cdel(C)$.
Thus we have 
\[
\cost(C,f) + ||f|| = \cdel(C) + ||f'||+1 \leq \cost(C,f') + ||f'||,
\]
which completes the proof of claim.
\end{claimproof}

\begin{clm}\label{clm:notfprime}
Let $C,f,f'$ satisfy property b) in \cref{clm:constr-oneof}.
Then there exists an optimum solution $\varphi$ for $(\calV,\calC,\cost)$ such that $\varphi|_C \neq f'$.
\end{clm}
\begin{claimproof}
Let $\varphi'$ be an optimum solution for $(\calV,\calC,\cost)$.
If $\varphi'|_C \neq f'$, we are done, so suppose otherwise.
Define $\varphi \from \calV \to [q] \cup \{\del\}$ as follows:
\[
\varphi(v) = \begin{cases}
\varphi'(v) & \text{ if } v \notin C,\\
f(v) 		& \text{ if } v \in C.
\end{cases}
\]
Note that as $\varphi'|_C = f' \prec f$, we have $\varphi' \prec \varphi$.
Let us compute the total cost of $\varphi$.
\begin{align*}
\mathsf{total\text{-}cost}(\varphi)= & \ ||\varphi|| + \sum_{C' \in \calC} \cost(C',\varphi|_{C'}) \\
= & \ |\varphi^{-1}(\del) \setminus C| + |\varphi^{-1}(\del) \cap C|  + \cost(C,\varphi|_{C}) + \sum_{C' \in \calC \setminus \{C\}} \cost(C',\varphi|_{C'})\\
\stackrel{\mathclap{\normalfont\mbox{(1)}}}{=} & \ |\varphi'^{-1}(\del) \setminus C| + ||f||  + \cost(C,f) + \sum_{C' \in \calC \setminus \{C\}} \cost(C',\varphi|_{C'})\\
\stackrel{\mathclap{\normalfont\mbox{(2)}}}{\leq} & \ |\varphi'^{-1}(\del) \setminus C| + ||f||  + \cost(C,f) + \sum_{C' \in \calC \setminus \{C\}} \cost(C',\varphi'|_{C'})\\
\stackrel{\mathclap{\normalfont\mbox{(3)}}}{\leq} & \ |\varphi'^{-1}(\del) \setminus C| + ||f'||  + \cost(C,f') + \sum_{C' \in \calC \setminus \{C\}} \cost(C',\varphi'|_{C'})\\
= & \ ||\varphi'|| + \sum_{C' \in \calC} \cost(C',\varphi'|_{C'}) = \mathsf{total\text{-}cost}(\varphi').
\end{align*}
where step (1) follows as $\varphi'|_{\calV \setminus C} = \varphi|_{\calV \setminus C}$, step (2) follows from the wildcard property, and step (3) follows from the properties of $C,f,f'$ given by \cref{clm:constr-oneof} b).
Since $\varphi'$ is an optimal solution and thus its total cost is minimum possible,
we conclude that $\mathsf{total\text{-}cost}(\varphi)=\mathsf{total\text{-}cost}(\varphi')$ and thus $\varphi$ is also an optimal solution.
\end{claimproof}

In the beginning of (each recursive call of) our algorithm, we exhaustively apply the following two reduction rules.
First, if there is some $v \in \calV \setminus \bigcup_{C \in \calC}$, then we remove it from $\calV$.
Note that an optimal solution for the reduced instance can be extended to an optimal solution for the original one by mapping $v$ to any element of $[q]$.
Second, if there are $C,C' \in \calC$ such that $C' \subseteq C$, we can obtain an equivalent instance by removing $C'$ from $\calC$ and, for each $g \in \setf{C}$, increasing the value of $\cost(C,g)$ by $\cost(C',g|_{C'})$.
After applying the second reduction rule exhaustively we can assume that the sets in $\calC$ are pairwise incomparable.
It is clear that the reduction rules can be applied in polynomial time.

Now we are ready to describe our algorithm.
If $n = |\calV| \leq r$, we solve the problem in constant time by brute force in constant time, as $q$ and $r$ are constants.
Otherwise, pick any $C \in \calC$ and denote $|C|$ by $r'$. Clearly $r' \leq r$. Consider two possibilities given by \cref{clm:constr-oneof}.
If possiblity a) applies to $C$, then notice that we can safely remove $C$ from $\calC$.
Indeed, an optimal solution of $(\calV,\calC \setminus \{C\},\cost|_{\calC \setminus \{C\}})$ is also an optimal solution of $(\calV,\calC,\cost)$
and the total cost of this solution in the latter instance is equal to its total cost in the former instance plus $\cdel(C)$.
So assume that possibility b) applies and let $f,f'$ be as in \cref{clm:constr-oneof}~b).

We guess the valuation of variables from $C$, considering all possibilities except for $f'$.
By \cref{clm:notfprime} we know that it is safe to ignore $f'$ as there is always an optimum solution whose valuation restricted to $C$ is not $f'$.

This results in $(q+1)^{r'}-1$ branches.  Consider one such branch corresponding to a valuation $\widetilde{f} \in \setf{C}$.
The instance considered in this branch is $(\calV',\calC',\cost')$, defined as follows.
The set $\calV'$ of variables is $\calV \setminus C$; note that it is non-empty as $|C| \leq r$ and $n > r$.
The set $\calC'$ is defined as $\calC' \coloneqq \{ \widetilde{C} \cap \calV' ~|~ \widetilde{C} \in \calC\}$. Note that $\emptyset \notin \calC'$, since otherwise the second reduction rule could have been applied.
For each $C' \in \calC'$ and for each $g \in \setf{C'}$ we set 
\[
\cost'(C',g) = \sum_{\substack{\widetilde{C} \in \calC\\ \widetilde{C} \cap \calV' = C'}} \cost(\widetilde{C}, (g \cup \widetilde{f})|_{\widetilde{C}}),
\]
where by $g \cup \widetilde{f}$ we denote the valuation obtained by assigning the values of variables from $C$ according to $\widetilde{f}$, and the values of variables from $C'$ according to $g$; note that these sets are disjoint.
Since we aim to solve $(\calV,\calC,\cost)$, the value of an optimum solution found for $(\calV',\calC',\cost')$ must be increased by $||\widetilde{f}||$.
It is straightforward to verify that the above recursive procedure is correct. Now let us argue about the running time.

Since the local computation at each recursive call can be performed in polynomial time,
in order to bound the time complexity we need to estimate the number of leaves in the recursion tree.
Let $T(n)$ denote this value for an instance with $n$ variables. We aim to show that $T(n) \leq \left((q+1)^r-1 \right)^{n/r}$.

If $n \leq r$, then $T(n)=1 \leq \left((q+1)^r-1 \right)^{n/r}$.
Suppose now that $n > r$ and for every $n' < n$ we have $T(n') \leq \left((q+1)^r-1 \right)^{n'/r}$.
We obtain the following recursive inequality:
\[
T(n) \leq \left((q+1)^{r'}-1 \right) \cdot T(n-r').
\]
Applying the inductive assumption, we get that 
\[
T(n) \leq \left( (q+1)^{r'}-1 \right) \cdot \left((q+1)^r-1 \right)^{(n-r')/r} \leq \left((q+1)^r-1 \right)^{n/r},
\]
similarly as in the proof of \cref{thm:lcoloring-better-alg}.
Summing up, the total running time is bounded by $T(n) \cdot n^{\bigO(1)} \leq  \left((q+1)^r-1 \right)^{n/r} \cdot n^{\bigO(1)}$.
\end{proof}

Now \cref{thm:vd-better-alg} follows by combining \cref{lem:reduce-to-csp} with \cref{thm:solve-csp} and observing that $\left((q+1)^r-1 \right)^{1/r}< q+1$.

\section{Equivalent Hypotheses: Set Covering/Packing/Partitioning}\label{sec:equivHypo}
In this section we prove \cref{thm:SCCequivalent}.

Recall that the Set Cover Conjecture (SCC) asserts that for all $\varepsilon > 0$, there exists $d \geq 1$ such that there is no algorithm that solves every $\leqdsetcover$ instance $(U,\calF)$ in time $(2 - \varepsilon)^{|U|} \cdot |U|^{\bigO(1)}$~\cite{cyganProblemsHardCNFSAT2016}.

We show that, for a whole list of problems, similar lower bound statements are actually all equivalent to the SCC. The respective problems were informally introduced in the introduction, \cref{sec:intro}. Formally, they are defined in \cref{sec:prelims}. At a high level, we consider covering, packing and partition problems and distinguish between problem variants, where every set of the family has size exactly $d$ or at most $d$, and we also consider (whenever it makes sense) whether the number of selected sets or the size of their union is maximized.

All discussed problems can be solved in time $2^n \cdot n^{\bigO(1)}$, where $n$ is the size of the universe, by quite straightforward dynamic programming.
We now show that these problems are all equivalent to the SCC in the following sense: a \emph{faster than standard dynamic programming} algorithm for any of those problems violates the SCC. Similarly, if the SCC fails, then there is a \emph{faster than standard dynamic programming} algorithm for each of the problems in the list.

\SCCequivalent*

\cref{fig:reductionoverview} gives an overview of the reductions that we prove in order to show \cref{thm:SCCequivalent}.

\subsection{Basic Reductions}	
Recall from \cref{sec:overview_SCCequiv} that we say that a problem class $\calA=\{A_d\mid d\ge 1\}$ (such as \leqsetcoverclass) is $\hard$ if the following statement holds.
\begin{quote}
	For each $\eps > 0$, there is some $d\geq 1$ such that no algorithm solves $A_d$ on all $n$-element instances in time $(2 - \eps)^{n} \cdot n^{\Oh(1)}$.
\end{quote} 
Let $\calB=\{B_d\mid d\ge 1\}$ be another problem class. In the following sections we prove many reductions of the form ``If $\calA$ is $\hard$ then $\calB$ is $\hard$''. We will usually prove the contrapositive statement:

\begin{quote}
	Suppose for some $\eps > 0$ it holds that for every $d \geq 1$, there is an algorithm that solves $B_d$ on $n$-element instances in time $(2 - \eps)^{n} \cdot n^{\Oh(1)}$. \\[.5em]
	Then there exists $\eps' > 0$ such that for every $d \geq 1$, there is an algorithm that solves $A_d$ on $n$-element instances in time $(2 - \eps')^{n} \cdot n^{\Oh(1)}$.
\end{quote}

As \leqdsetcover (resp., \leqdsetpackingsets) is more general than \eqdsetcover (resp., \eqdsetpackingsets), we immediately observe the following.

\begin{obs}\label{obs:eqSCtoleqSC} \label{obs:eqPCtoleqPC2}
If \eqsetcoverclass is $\hard$, then \leqsetcoverclass is $\hard$; and 
if \eqsetpackingsetsclass is $\hard$, then \leqsetpackingsetsclass is $\hard$.
\end{obs}

Similarly, since $t$ disjoint sets of size $d$ cover exactly $t \cdot d$ elements, $\leqdsetpackingunion$ problem is a generalization of $\eqdsetpackingsets$.
\begin{obs}\label{obs:eqPCtoleqPC1}
If \eqsetpackingsetsclass is $\hard$, then \leqsetpackingunionclass is $\hard$.			
\end{obs}

Next, we observe that \eqdsetpartition can be easily reduced to  \eqdsetcover.
Now consider a set family $\calF$ over the universe $U$ such that each set in $\calF$ has size $d$. We observe that each set partition of $(U,\calF)$ is a set cover of size $n/d$, and vice versa each set cover of size $n/d$ can only contain disjoint sets, and is therefore a set partition.
Thus, \eqdsetcover can be used to solve \eqdsetpartition.

\begin{obs}\label{lem:eqPTtoeqSC}
If \eqsetpartitionclass is $\hard$, then \eqsetcoverclass is $\hard$.		
\end{obs}

Another simple argument gives the following.
\begin{lem}\label{lem:leqSCtoleqPTs}
If \leqsetcoverclass is $\hard$, then \leqsetpartitionsetsclass is $\hard$..	
\end{lem}

\begin{proof}
	Given a \leqdsetcover instance $(U, \calF, t)$, create a \leqdsetpartitionsets instance $(U, \calF', t)$ where
	$
		\calF' \coloneqq \{A \mid A \neq \emptyset, A \subseteq S, S \in \calF\}. 
	$

	First, suppose $U$ has a cover $\calS = \{S_1, \ldots, S_r\}  \subseteq \calF$ of size $r \leq t$.
	Without loss of generality we can assume that it is inclusion-wise minimal, i.e., for every $i \in [r]$ there exists $u \in S_i$ that is not covered by $\calS \setminus S_i$.

	Define sets $A_1,\ldots,A_r$ as follows. We set $A_1 = S_1$.
	Then, for $i \geq 2$, we set $A_i = S_i \setminus \bigcup_{i < i} A_i$.	
	It is straightforward to see that $\Bigl(\bigcup_{1 \leq i \leq r} A_i\Bigr) = \Bigl(\bigcup_{1 \leq i \leq S_i} S_i\Bigr)$.
	Furthermore, the sets $A_1, \ldots, A_r$ are pairwise disjoint, and they are nonempty by minimality of $\calS$.
	Since for $1 \leq i \leq r$, $A_i$ is a subset of $S_i$ at the end of the algorithm described above,
	we have $\{A_1, \ldots, A_r\} \subseteq \calF'$. Consequently, this set is
	a partition of $U$ of size $r \leq t$.

	Second, suppose that $U$ has a partition $\calS' \subseteq \calF'$ of size $r \leq t$.
	Since each set in $\calF'$ is a subset of some set in $\calF$, there is a corresponding
	set cover $\calS$ of $U$ of size $r \leq t$.
\end{proof}

Let $U$ be some universe and $\calF$ a family of size-$d$ subsets of $U$.
Clearly, $(U, \calF)$ has a set partition if and only if it has a set packing of size (at least) $\abs{U}/d$.
So we obtain the following observation.

\begin{obs}\label{obs:eqPTtoeqPC}
If \eqsetpartitionclass is $\hard$, then \eqsetpackingsetsclass is $\hard$.	
\end{obs}

In order to prove the following reduction we again use the trick that we also employed to show \cref{lem:csp-structured}. Namely, we split the universe of the instance into blocks and then guess the number of unused elements for each block.

\begin{lem}\label{lem:leqPCtoleqPTs}
If \leqsetpackingunionclass is $\hard$, then \leqsetpartitionsetsclass is $\hard$.
\end{lem}
\begin{proof}
	Let $(U, \calF, t)$ be an instance of \leqdsetpackingunion. Let $n=\abs{U}$.
	Let $b$ be the smallest integer such that $(2b+2)^{1/b}<(1+\eps/2)$. This choice will become clear later on.
	Without loss of generality, we assume that $n$ is divisible by $b$ (as otherwise we can add at most $b-1$ dummy elements that are in none of the sets from $\calF$).
	We split the elements from $U$ into pairwise disjoint blocks $U_1, \ldots, U_{n/b}$ of size $b$ each.
	
	Suppose we are given an integer $x_i\in \{0, \ldots, b\}$ for each block $U_i$ that specifies how many elements of this block are \emph{not} covered in a hypothetical \leqdsetpackingunion{} solution. 
	We refer to the vector $\boldx=(x_1, \ldots, x_{n/b})$ as the signature of this solution. Note that each solution has precisely one signature.
	
	Consider the following algorithm: 
	\begin{enumerate}
		\item Define the set  $U'=U\cup \{a_i \mid i\in [n/b]\}$, where $a_i$ is a new element introduced for each block.
		\item Iterate over all possible signatures. For each such signature $\boldx$, build a new set system $\calF_\boldx$ that contains all sets from $\calF$ and, in addition, for each $i$ and each set $S\subseteq U_i$ of size $x_i$, the set $S\cup \{a_i\}$.
		\item For $d'=b+1$, run the assumed algorithm for \leqsetpartition{d'} on the instance $(U', \calF_\boldx, t+n/b)$.
	\end{enumerate}
	
	First note that by the fact that a set $S\cup \{a_i\}$ has size $x_i+1\le b+1=d'$, $(U', \calF_\boldx, t+n/b)$ is, in fact, an instance of \leqsetpartitionsets{d'}.
	The idea is that for each block $U_i$, a partition of $U'$ contains precisely one set that contains the new element $a_i$. This set holds all elements that are not covered by a corresponding packing of sets from $\calF$ with signature $\boldx$.
	Hence the partition corresponding to such a packing contains one additional set per block. So it is straight-forward that above algorithm solves \leqdsetpackingunion.
	
	Let $n'=\abs{U'}=n+n/b$.
	There are at most $(b+1)^{n/b}$ different signatures; and for each signature the algorithm for \leqsetpartition{d'}  takes time $(2 - \eps)^{n'} \cdot n'^{\bigO(1)}$. Thus the total runtime is bounded by
	\begin{align*}
		(b+1)^{n/b}\cdot (2 - \eps)^{n'} \cdot n'^{\bigO(1)}
		&\le (b+1)^{n/b}\cdot (2 - \eps)^{n+n/b} n^{\bigO(1)}\\
		&\le (2b+2)^{n/b} \cdot (2 - \eps)^n \cdot n^{\bigO(1)}\\
		&\le (1+\eps/2)^n \cdot (2 - \eps)^n \cdot n^{\bigO(1)}\\
		&=(2-\eps^2/2)^n \cdot n^{\bigO(1)}.
	\end{align*}
	For $\eps'=\eps^2/2$ this gives the correct runtime.
\end{proof}

\subsection{Structured Families of Sets}
Several subsequent reductions use similar ideas. Therefore we introduce some notation and prove a few
technical lemmas to avoid redundancy as much as possible.
Let $U$ be a set of $n$ elements and $\calF \subseteq 2^{U}$ be a set family where each set has size at most $d$.

Let $(\calF_1, \ldots, \calF_d)$ be the partition of $\calF$ according to set size, i.e., let $\calF_i$ contain all size-$i$ sets of $\calF$. We now define a vector that represents a guess of how many sets of each size are chosen (in some solution).
	Let $\boldr=(r_1, \ldots, r_d) \in \{0, \ldots, {n \choose d}\}^{d} $ satisfy
	\begin{enumerate}
		\item $\sum_{i = 1}^{d} r_i \leq n$, and
		\item $r_i  \leq \abs{\calF_i}$.
	\end{enumerate}
	We say that $\boldr$ is an \signature{\calF} and define $w(\boldr)=\sum_{i = 1}^{d} r_i$ to be the \emph{weight} of $\boldr$. For a packing/partition $\calS \subseteq \calF$ of $U$, we say that $\calS$ respects $\boldr$ if $\calS$ contains exactly $\boldr_i$ sets from $\calF_i$.

	Intuitively, our aim is to modify $(U,\calF$)
	such that all sets have the same size while preserving the original packings/partitions. To this end, for each
	\signature{\calF}  $\boldr$ and integer $c \geq 1$, we will define a new set system $\left( U_{\boldr}, \calF_{\boldr}
	\right)$ such that every set in $\calF_{\boldr}$
	has size $\left( c \cdot d\,! + 1 \right) $. 
	Moreover, we will show that there is a one-to-one relationship between
	the packings (partitions) of $(U,\calF)$ that respect $\boldr$,
	and packings (partitions) of $\left( U_{\boldr},
	\calF_{\boldr} \right)$. We say that $\left( U_{\boldr},
	\calF_{\boldr} \right)$ is the \emph{$\joininstance{c,\boldr}$ of $(U,\calF)$}.

	The idea is to merge, for each $i\in [d]$, selections of $a_i \coloneqq \frac{c \cdot d\,!}{i}$ disjoint sets from $\calF_i$ to obtain new
	sets of size $a_i \cdot i = \frac{c \cdot d\,!}{i} \cdot i = c \cdot
	d\,!$ (now the same size for all $i$). A selection of $r_i/a_i$ of these new sets then corresponds to a selection of $r_i$ of the original sets. In order to avoid divisibility issues, we introduce some dummy sets and work with the vector  $\bolds=(s_1,\ldots, s_d) \in \mathbb{Z}_{\geq 0}^{d}$ with
	\begin{equation*}
		\bolds_i \coloneqq \left\lceil \frac{r_i}{a_i} \right\rceil \cdot a_i \quad (\text{for }i\in [d]),
	\end{equation*}
	So, coming from $r_i$, $s_i$ is the next-largest integer that is divisible by $a_i$.

	For $I_{\boldr} \coloneqq \{i\in [d]  \mid r_i \neq 0\}$, and
	an index $i \in I_{\boldr}$, we define $[s_i-r_i]$ new sets of dummy elements of size $i$ each:
	\[
		A^{\boldr, i}_j \coloneqq \{a^{i,j}_1, \ldots, a^{i,j}_{i}\} \quad \text{for } j \in [s_i-r_i].
	\]
	We assume that the sets $\calA^{\boldr,i}_j$ are all
	disjoint sets of new elements (outside of $U$).
	For a given signature $\boldr$, define the set of dummy elements
	\[
			N_{\boldr} \coloneqq \bigcup_{i \in I_{\boldr}}\bigcup_{j\in [s_i-r_i]} A^{\boldr,i}_j.
	\]
	Note that
	\begin{equation}\label{eq:size_U_r}
		\abs{N_{\boldr}} =  \abs{\bigcup_{i \in I_{\boldr}}\bigcup_{j\in [s_i-r_i]} A^{\boldr,i}_j} \leq  d \cdot c \cdot d\,! =  \Oh(1)
	\end{equation}
	since $\abs{I_{\boldr}} \leq d$ and $\bolds_i - \boldr_i \leq a_i$.
	
	For each $i \in I_{\boldr}$, the idea is to extend the part $\calF_i$ by the new size-$i$ sets
	\[
		\calF^{\boldr, i}_{\textrm{add} } \coloneqq \{ A^{\boldr,i}_j \mid j\in [s_i-r_i]\}.
	\]

	Then, for each $i \in I_{\boldr}$, we consider the sets that can be obtained as a union of $a_i$ pairwise disjoint sets from $\calF_i\cup \calF^{\boldr,i}_{\textrm{add}}$, formally,
	\begin{equation*}
		\calG^{\boldr}_i \coloneqq \bigg\{\bigcup_{j\in [a_i]} V_j \mid V_1, \ldots, V_{a_i} \text{ are pairwise disjoint sets from }\calF_i\cup \calF^{\boldr,i}_{\textrm{add}}\bigg\}.
	\end{equation*}

	Crucially, for each $i \in I_{\boldr}$ and set $X \in
	\calG^{\boldr}_i$, it holds that $\abs{X} = a_i \cdot i = \frac{c \cdot
	d\,!}{i} \cdot i = c \cdot d\,!$.  Ultimately, the idea is that all of
	these sets $X$ will form a new collection of sets whose elements now all have
	the same size.
	
	As previously mentioned, for each $i$, we are ultimately interested in selections of precisely $r'_i\coloneqq\frac{\boldr_i}{a_i}$ of the new size-$a_i$ sets. To ensure that no more than $r_i'$ sets are picked, we introduce $r_i'$ new elements $E_i \coloneqq \bigg\{e^{i}_1, \ldots, e^{i}_{r_i'}\bigg\}$, and extend each set in our new collection of sets $\calG^{\boldr}_i$ by one of these elements as follows:

	For each $i \in I_{\boldr}$, let
	\[
		\calZ^{\boldr}_i  \coloneqq \Big\{\{e\} \cup X \mid e\in E_i, X \in \calG^{\boldr}_i\Big\}.
	\]
	With this modification, each selection of disjoint sets from $\calZ^{\boldr}_i$ has size at most $\abs{E_i}=r_i'$.
	Let $\calE_\boldr\coloneqq \bigcup_{i \in [d]}E_i$ and let $\alpha_{\boldr} \coloneqq \sum_{i \in I_{\boldr}} r_i'$ be the target number of sets that we aim to select from the new collection of sets.

	 We finish the definition of the $\joininstance{c,\boldr}$ $\left( U_{\boldr},
	\calF_{\boldr} \right)$ by setting
	\[
	\calF_{\boldr} \coloneqq \bigcup_{i \in I_{\boldr}} \calZ^{\boldr}_i, \quad\text{and}\quad
	U_{\boldr} \coloneqq U \cup N_{\boldr} \cup \calE_\boldr.		
	\]

	For the bound on the size of this new instance, we observe that
	\begin{align*}
		\abs{\calE_\boldr} = \sum_{i \in I_{\boldr}} r_i'= \sum_{i = 1}^{d} \left\lceil \frac{\boldr_i}{a_i} \right\rceil \leq d + \sum_{i = 1}^{d} \frac{\boldr_i}{a_i}
		= d + \sum_{i = 1}^{d} \frac{\boldr_i \cdot i}{c \cdot d\,!}
		\leq d + \sum_{i = 1}^{d}  \frac{\boldr_i}{c (d-1)\,!}
		&\leq d + \frac{n}{c (d-1)\,!},
	\end{align*}
	where the last step follows because $\boldr$ is an $\signature{\calF}$.
	Moreover, each set in $\calF_{\boldr}$ has size $c
	\cdot d\,! + 1$ and consequently
	\begin{equation}\label{eq:U_bar_r_size}
		\abs{U_{\boldr} } = \abs{U} + \abs{N_{\boldr}} + \abs{\calE_\boldr} = n + \Oh(1) + d + \frac{n}{c \cdot (d-1)\,!} = n + \frac{n}{c \cdot (d-1)\,!} + \Oh(1).
	\end{equation}
	
	Next we show that these structured instances preserve packings and partitions. 
	\begin{lem}\label{lem:join_packing_new}
		Let $\calF$ be a collection of subsets of a universe $U$, each with size at most $d$. Let $c \geq 1$ be an integer. For each $\signature{\calF}$ $\boldr$, the following two statements are equivalent
		\begin{enumerate}
			\item $U$ has a packing $\calS \subseteq \calF$ such that $\abs{\calS \cap \calF_i} = \boldr_i$ for $i\in [d]$.\label{item:equiv_pack_1}
			\item $U_{\boldr}$ has a packing $\calS_{\boldr} \subseteq \calF_{\boldr}$ of size $\alpha_{\boldr}$ such that $(N_{\boldr} \cup \calE_\boldr) \subseteq \calS_\boldr$.\label{item:equiv_pack_2}
		\end{enumerate}
		Moreover, it holds that
		\begin{equation*}
			\left( \bigcup_{A \in \calS_{\boldr}} A \right)  \setminus \left( N_{\boldr} \cup \calE_\boldr \right)  = \left( \bigcup_{A \in \calS} A \right).
		\end{equation*}
	\end{lem}

	\begin{proof}
		Let $\boldr$ be an $\signature{\calF}$.
		
		\paragraph*{First direction: 1. implies 2.}
		Suppose $U$ has a packing $\calS \subseteq \calF$
		with $\abs{\calS \cap \calF_i} = \boldr_i$ for $i\in [d]$.
		For $i\in [d]$, we set $\calS_i \coloneqq \calS \cap \calF_i$
		and $\bolds_i \coloneqq \left\lceil \frac{\boldr_i}{a_i} \right\rceil \cdot a_i$.
		Let $I$ be the set of indices $i$ for which $\calS_i$ is nonempty. For $i \in I$, consider the set $\calS_i \cup
		\calF^{\boldr,i}_{\textrm{add}}$ and observe that its size is
		equal to $\bolds_i$ because
		$\abs{\calF^{\boldr,i}_{\textrm{add}}} = \bolds_i - \boldr_i$.
		Then, for each $i \in I$, partition $\calS_i \cup
		\calF^{\boldr,i}_{\textrm{add}}$ into groups of size $a_i$ to define new sets of size $a_i\cdot i=c\cdot d\,!$. Here we use the fact that the sets in $\calS_i \cup
		\calF^{\boldr,i}_{\textrm{add}}$ are pairwise disjoint.
		Let $D^{i}_1\ldots, D^{i}_{\bolds_i}$ be an enumeration of the sets in $\calS_i \cup
		\calF^{\boldr,i}_{\textrm{add}}$. To group the sets, we set
		\begin{equation*}
			\calS'_i \coloneqq \bigg\{D^{i}_{(j-1) \cdot a_i + 1} \cup \ldots \cup D^{i}_{j \cdot a_i} \mid j\in[\bolds_i/a_i] \bigg\}.
		\end{equation*}
		Observe that each set in $\calS'_i$ has size $a_i \cdot i = c \cdot d\,!$, $\abs{\calS'_i}= \frac{\bolds_i}{a_i} = \left\lceil \frac{\boldr_i}{a_i} \right\rceil$, $\calS'_i \subseteq \calG^{\boldr}_i$, and importantly that sets in $\calS'_i$ are pairwise disjoint. Moreover,
		\begin{equation}\label{eq:union_S_prime}
			\bigcup_{A \in \calS'_i} A = \bigcup_{A \in \left(\calS_i \cup \calF^{\boldr,i}_{\textrm{add}}\right)} A .
		\end{equation}
		Then we add to each set in $\calS'_i$ some distinct element from $E_i$. Note that this operation is well-defined since $\abs{E_i} = \abs{\calS'_i} = \left\lceil \frac{\boldr_i}{a_i} \right\rceil$. Explicitly, for an enumeration $T_1, \ldots, T_{\left\lceil \frac{\boldr_i}{a_i} \right\rceil} $ of the sets in
		$\calS'_i$, let
		\begin{equation*}
			\calT_i \coloneqq \bigg\{\{e^{i}_j\} \cup T_j \mid 1 \leq j \leq \left\lceil \frac{\boldr_i}{a_i} \right\rceil \bigg\}.
		\end{equation*}
		It follows that each element of $\calT_i$ has size $c \cdot d\,! + 1$, $\abs{\calT_i} = \left\lceil \frac{\boldr_i}{a_i} \right\rceil$, $\calT_i \subseteq \calZ^{(\boldr)}_i$, and importantly that sets in $\calT_i$ are pairwise disjoint. Moreover,
		\begin{equation}\label{eq:union_T}
			\bigcup_{A \in \calT_i} A = E_i\, \cup \left( \bigcup_{A \in \calS'_i} A\right).
		\end{equation}
		By \eqref{eq:union_S_prime} and \eqref{eq:union_T}, it follows that $\calS_{\boldr} \coloneqq \bigcup_{i \in I} \calT_i$ satisfies

		\begin{enumerate}
			\item $\calS_\boldr \subseteq \calF_{\boldr}$,
			\item sets in $\calS_{\boldr}$ are pairwise disjoint,
			\item the union of the sets in $\calS_\boldr$ is
				\begin{equation*}\label{eq:union_S_r}
					\bigcup_{A \in \calS_\boldr} A = \bigcup_{i \in I}  \bigcup_{A \in \calT_i} A = \left( \bigcup_{A \in \calS} A \right) \cup N_{\boldr} \cup \calE_{\boldr}.
				\end{equation*}
		\end{enumerate}
		This implies that $\calS_\boldr \subseteq \calF_{\boldr}$ is a packing of $U_\boldr$.
		Finally, the size of $\calS_\boldr$ is 
		\begin{equation*}
			\abs{\calS_\boldr}=\sum_{i \in I} \abs{\calT_i} = \sum_{i \in I} \left\lceil \frac{\boldr_i}{a_i} \right\rceil = \sum_{i = 1}^{d} \left\lceil \frac{\boldr_i}{a_i} \right\rceil = \alpha_{\boldr}.
		\end{equation*}
		
		\paragraph*{Second direction: 2. implies 1.}
		Now suppose that  $U_{\boldr}$ has a packing
		$\calS_{\boldr} \subseteq
		\calF_{\boldr}$ of size
		$\alpha_{\boldr}$ such that $(N_{\boldr} \cup \calE_\boldr) \subseteq \calS_\boldr$.
		Let $I_{\boldr} \coloneqq \{i\in [d]  \mid r_i \neq 0\}$.
		Observe that for $i \in I_{\boldr}$, $\abs{\calS_{\boldr} \cap
		\calZ^{\boldr}_i} \leq \left\lceil \frac{\boldr_i}{a_i}
		\right\rceil$ because each set in $\calZ^{\boldr}_i$ contains an element
		of $E_i$.
		However, this implies that $\calS_{\boldr} \cap
		\calZ^{\boldr}_i$ has exactly $r_i'\coloneqq\left\lceil \frac{\boldr_i}{a_i}
		\right\rceil$ elements as $\abs{\calS_{\boldr}}=\alpha_{\boldr}$.

		So, for each $i \in I_{\boldr}$, there are pairwise disjoint sets $X^{i}_1, \ldots, X^{i}_{r_i'} \in \calG^{\boldr}_i$ such that
		\begin{equation*}
			\calS_{\boldr} \cap \calZ^{\boldr}_i = \bigg\{\{e^{i}_{j}\} \cup X^{i}_{j} \mid j\in [r_i'] \bigg\}.
		\end{equation*}
		Furthermore, by definition of the set $\calG^{\boldr}_i$, each $X^{i}_j$ is a union of $a_i$ disjoint sets
		of size $i$ that belong to the set $\calF_i \cup
		\calF^{\boldr,i}_{\textrm{add}}$. Let $\kappa(X^{i}_j)$ denote these $a_i$ sets. For each $i \in I_{\boldr}$, let $\calS_i=\{A \mid A\in \kappa(X^{i}_j), j\in [r_i']\}$.
		Then we have $\calS_i \subseteq \calF_i \cup
		\calF^{\boldr,i}_{\textrm{add}}$, $\abs{\calS_i} =
		\left\lceil \frac{\boldr_i}{a_i} \right\rceil \cdot a_i =
		\bolds_i$, and the sets in $\calS_i$ are pairwise disjoint. Recall that $\calS_\boldr$ has the property that $(N_\boldr \cup \calE_\boldr) \subset \calS_\boldr$. Therefore, for each $i \in I_\boldr$ it holds that $\calF^{\boldr,i}_{\textrm{add}} \subseteq \calS_i$.
		Since $\abs{\calF^{\boldr,i}_{\textrm{add}}} = \bolds_i -
		\boldr_i$, it also holds that
		\begin{equation*}\label{eq:size_S_minus_F}
			\abs{\calS_i \setminus \calF^{\boldr,i}_{\textrm{add}}}
			= \abs{\calS_i} - \abs{\calF^{\boldr,i}_{\textrm{add}}}
			= \bolds_i - (\bolds_i - \boldr_i) = \boldr_i.
		\end{equation*}

		Finally, for $\calS \coloneqq \bigcup_{i \in I_{\boldr}} \left( \calS_i \setminus \calF^{\boldr,i}_{\textrm{add}} \right) $, it follows that
		\begin{equation*}
			\abs{\calS \cap \calF_i} = \abs{\calS_i \setminus \calF^{\boldr,i}_{\textrm{add}}} = \boldr_i.
		\end{equation*}
		Furthermore, it is straightforward to verify that $\left( \bigcup_{A \in \calS} A \right) = \left( \bigcup_{A \in \calS_{\boldr}} A \right)  \setminus \left( N_{\boldr} \cup \calE_\boldr \right).$ The combination of this observation and the fact that $\calS_\boldr$ is a packing of $U_\boldr$ implies that $\calS$ is a packing of $U$.
	\end{proof}

	Note that the assertion in \cref{lem:join_packing_new} remains valid when considering partitions instead of packings. This is a consequence of $U_{\boldr} = U \cup N_{\boldr} \cup \calE_\boldr$ and the second part of \cref{lem:join_packing_new}.
	\begin{cor}\label{cor:join_partition_new}
		Let $\calF$ be a collection of subsets of a universe $U$, each with size at most $d$. Let $c \geq 1$ be an integer. For each $\signature{\calF}$ $\boldr$, the following two statements are equivalent
		\begin{enumerate}
			\item $U$ has a partition $\calS \subseteq \calF$ such that $\abs{\calS \cap \calF_i} = \boldr_i$ for $i\in [d]$.\label{item:equiv_part_1}
			\item $U_{\boldr}$ has a partition $\calS_{\boldr} \subseteq \calF_{\boldr}$.\label{item:equiv_part_2}
		\end{enumerate}
	\end{cor}

\subsection{Reductions Based on Structured Families}\label{sec:furtherred}
	
\begin{lem}\label{lem:leqPTstoPT}
If \leqsetpartitionsetsclass is $\hard$, then \leqsetpartitionclass is $\hard$.
\end{lem}

\begin{proof}
	Let $(U,\calF,t)$ be an $n$-element instance of \leqdsetpartitionsets for some $d \geq 1$ and
	let $c$ be the smallest integer such that $(2-\eps)^{\frac{1}{c \cdot (d-1)\,!}} < (1 + \frac{\eps}{2})$.
	Moreover, let $\calB$ be an algorithm that solves \leqsetpartition{h} for $h \coloneqq (c \cdot d\,! + 1)$ in time $(2 - \eps)^{n} \cdot n^{\Oh(1)}$. Finally define $\eps' \coloneqq \frac{\eps^{2}}{2}$.

	\paragraph{Algorithm for \leqdsetpartitionsets and its correctness.}
	Given an instance $(U,\calF,t)$ of \leqdsetpartitionsets, the algorithm
	starts by guessing an $\signature{F}$ $\boldr$ such that $w(\boldr)
	\leq t$. For each such $\boldr$, the algorithm constructs the
	$\joininstance{c,\boldr}$ of $(U, \calF)$, denoted by $(U_\boldr,
	\calF_\boldr)$. Note that each set in $\calF_\boldr$ has size $(c \cdot
	d\,! + 1) = h$, hence $(U_\boldr, \calF_\boldr)$ is an instance of
	\leqsetpartition{h}. Finally, the algorithm returns \yes if any of the
	calls $\calB(U_{\boldr}, \calF_\boldr)$ returns \yes, otherwise it
	returns \no.

	To establish the correctness of the algorithm, assume that
	$(U,\calF,t)$ is a \yes-instance. Then $U$ has a partition $\calS
	\subseteq \calF$ of size at most $t$. Construct the $\signature{\calF}$
	$\boldr$ by defining the coordinate $\boldr_i \coloneqq \abs{\calS \cap
	\calF_i}$ for $1 \leq i \leq d$ and observe that $w(\boldr) =
	\abs{\calS} \leq t$. By definition, $\calS$ satisfies
	\cref{item:equiv_part_1} in \cref{cor:join_partition_new} for this
	specific $\signature{\calF}$ $\boldr$, therefore it holds that
	$U_\boldr$ has a partition $\calS_\boldr \subseteq \calF_\boldr$.
	Therefore, $\calB(U_\boldr, \calF_\boldr)$ and in turn the algorithm
	defined above returns \yes.

	On the other hand, if $(U,\calF,t)$ is a \no-instance, then $U$ has no
	partition $\calS \subseteq \calF$ of size at most $t$. Therefore for
	any $\signature{\calF}$ $\boldr$ guessed by the algorithm,
	\cref{item:equiv_part_1} in \cref{cor:join_partition_new} will not
	hold. This implies that \cref{item:equiv_part_2} does not hold as well.
	As a result $\calB(U_\boldr, \calF_\boldr)$ returns \no. All in all,
	the algorithm described above returns $\no$ as well.

	\paragraph{Running time.} The number of $\signatures{\calF}$ is at most
	${n \choose d}^{d} \leq n^{d^{2}} = n^{\Oh(1)}$. Constructing the
	instance $(U_\boldr, \calF_\boldr)$ takes polynomial time in $n$. Therefore, the total running time is
	\begin{align*}
		n^{\Oh(1)} \cdot \max_{\signature{\calF}\; \boldr} (2 - \eps)^{\abs{U_{\boldr}}} &= n^{\Oh(1)} \cdot (2 - \eps)^{n + \frac{n}{c \cdot (d-1)\,!}}\\
		&= n^{\Oh(1)}\cdot (2 - \eps)^{n} \cdot \left( 1 + \frac{\eps}{2} \right) ^{n}\\
		&= n^{\Oh(1)} \cdot \left( 2 - \frac{\eps^{2}}{2} \right)^{n}\\
		&= n^{\Oh(1)} \cdot \left( 2 - \eps' \right)^{n}		
	\end{align*}
	where the first equality holds by \eqref{eq:U_bar_r_size} and the
	second equality holds because of the definition of $c$.
\end{proof}

\begin{lem}\label{lem:leqPTtoeqPT}
If \leqsetpartitionclass is $\hard$, then \eqsetpartitionclass is $\hard$.
\end{lem}

\begin{proof}
	Let $d \geq 1$ and $(U,\calF)$ be an instance of \leqdsetpartition.
	Moreover, let $c$ be the smallest integer such that $(2-\eps)^{\frac{1}{c \cdot (d-1)\,!}} < (1 + \frac{\eps}{2})$ and $\calB$ be an algorithm that solves \eqsetpartition{h} for $h \coloneqq (c \cdot d\,! +1)$ in time $(2 - \eps)^{n} \cdot n^{\Oh(1)}$. Finally define $\eps' \coloneqq \frac{\eps^{2}}{2}$.

	\paragraph{Algorithm for \leqdsetpartition and its correctness.} Given an instance $(U,\calF)$ of \leqdsetpartition, the algorithm starts with guessing an $\signature{F}$ $\boldr$. For each such $\boldr$, the algorithm constructs the $\joininstance{c,\boldr}$ of $(U, \calF)$, denoted by $(U_\boldr, \calF_\boldr)$. Note that each set in $\calF_\boldr$ has size $( c \cdot d\,! + 1) = h$, hence $(U_\boldr, \calF_\boldr)$ is an instance of \eqsetpartition{h}. Finally, the algorithm returns \yes if any of the calls $\calB(U_{\boldr}, \calF_\boldr)$ returns \yes, otherwise it returns \no.
	
	To show that the algorithm is correct, assume that $(U,\calF)$ is a \yes-instance. Then $U$ has a partition
	$\calS \subseteq \calF$. Construct the $\signature{\calF}$ $\boldr$ by defining $\boldr_i \coloneqq \abs{\calS \cap \calF_i}$ for $1 \leq i \leq d$. By definition, $\calS$ satisfies \cref{item:equiv_part_1} in \cref{cor:join_partition_new} for this specific $\signature{\calF}$ $\boldr$, therefore it holds that $U_\boldr$ has a partition $\calS_\boldr \subseteq \calF_\boldr$. Therefore $\calB(U_\boldr, \calF_\boldr)$ and in turn the algorithm defined above returns \yes.

	On the other hand, if $(U,\calF)$ is a \no-instance, then $U$ has no partition $\calS$ such that $\calS \subseteq \calF$. Therefore for any $\signature{\calF}$ $\boldr$ the algorithm tries, \cref{item:equiv_part_1} in \cref{cor:join_partition_new} will not hold. This implies that \cref{item:equiv_part_2} does not hold as well, consequently $\calB(U_\boldr, \calF_\boldr)$ returns \no. Hence, the algorithm described above returns $\no$.

	\paragraph{Running time.} The number of $\signatures{\calF}$ is at most ${n \choose d}^{d} \leq n^{d^{2}} = n^{\Oh(1)}$. Constructing the instance $(U_\boldr, \calF_\boldr)$ takes polynomial time in $n$ as well.
	Therefore, the whole running time becomes
	\begin{align*}
		n^{\Oh(1)} \cdot \max_{\signature{\calF}\; \boldr} (2 - \eps)^{\abs{U_{\boldr}}} &= n^{\Oh(1)} \cdot (2 - \eps)^{n + \frac{n}{c \cdot (d-1)\,!}}\\
		&= n^{\Oh(1)}\cdot (2 - \eps)^{n} \cdot \left( 1 + \frac{\eps}{2} \right) ^{n}\\
		&= n^{\Oh(1)} \cdot \left( 2 - \frac{\eps^{2}}{2} \right)^{n}\\
		&= n^{\Oh(1)} \cdot \left( 2 - \eps' \right)^{n}		
	\end{align*}
	where the first equality holds by \eqref{eq:U_bar_r_size} and the second equality holds because of the definition of $c$.	
\end{proof}

\begin{lem}\label{lem:leqPCtoeqPC}
 If \leqsetpackingsetsclass is $\hard$, then \eqsetpackingsetsclass is $\hard$.
\end{lem}

\begin{proof}
	Let $(U,\calF,t)$ be an instance of \leqdsetpackingsets for some $d
	\geq 1$ and let $c$ be the smallest integer such that
	$(2-\eps)^{\frac{1}{c \cdot (d-1)\,!}} < (1 + \frac{\eps}{2})$.
	Moreover, let $\calB$ be an algorithm that solves \eqsetpackingsets{h}
	for $h \coloneqq (c \cdot d\,! + 1)$ in time $(2 - \eps)^{n} \cdot
	n^{\Oh(1)}$. Finally define $\eps' \coloneqq \frac{\eps^{2}}{2}$.
	
	\paragraph{Algorithm for \leqdsetpackingsets and its correctness.}
	Given an instance $(U,\calF,t)$ of \leqdsetpackingsets, the algorithm
	starts with guessing an $\signature{F}$ $\boldr$ of weight at least
	$t$. For each such $\boldr$, the algorithm constructs the
	$\joininstance{c,\boldr}$ of $(U, \calF)$, denoted by $(U_\boldr,
	\calF_\boldr)$. Note that each set in $\calF_\boldr$ has size $( c
	\cdot d\,! + 1) = h$, hence $(U_\boldr, \calF_\boldr, \alpha_\boldr)$
	is an instance of \eqsetpackingsets{h}. Finally, the algorithm returns
	\yes if any of the calls $\calB(U_{\boldr}, \calF_\boldr,
	\alpha_\boldr)$ returns \yes, otherwise it returns \no.
	
	To show that the algorithm is correct, assume that $(U,\calF,t)$ is a
	\yes-instance. Then $U$ has a packing $\calS \subseteq \calF$ of size
	at least $t$. Construct the $\signature{\calF}$ $\boldr$ by defining
	$\boldr_i \coloneqq \abs{\calS
	\cap \calF_i}$ for $1 \leq i \leq d$ and observe that $w(\boldr)
	=\abs{\calS} \geq t$. By definition, $\calS$ satisfies
	\cref{item:equiv_pack_1} in \cref{lem:join_packing_new} for this
	specific $\signature{\calF}$ $\boldr$, therefore it holds that
	$U_\boldr$ has a packing $\calS_\boldr \subseteq \calF_\boldr$ of size
	$\alpha_\boldr$. Therefore $\calB(U_\boldr,
	\calF_\boldr,\alpha_\boldr)$ and in turn the algorithm defined above
	returns \yes.

	On the other hand, if $(U,\calF,t)$ is a \no-instance, then $U$ has no
	packing $\calS \subseteq \calF$ of size at least $t$.Therefore for any
	$\signature{\calF}$ $\boldr$ the algorithm tries,
	\cref{item:equiv_pack_1} in \cref{lem:join_packing_new} will not hold.
	This implies that \cref{item:equiv_pack_1} does not hold as well,
	consequently $\calB(U_\boldr, \calF_\boldr, \alpha_\boldr)$ returns
	\no. Hence, the algorithm described above returns $\no$.

	\paragraph{Running time.} The number of $\signatures{\calF}$ is at most
	${n \choose d}^{d} \leq n^{d^{2}} = n^{\Oh(1)}$. Constructing the
	instance $(U_{\boldr}, \calF_{\boldr}, \alpha_{\boldr})$ takes
	polynomial time in $n$ as well. Therefore, the whole running time
	becomes
	\begin{align*}
		n^{\Oh(1)} \cdot \max_{\signature{\calF}\; \boldr} (2 - \eps)^{\abs{U_{\boldr}}} &= n^{\Oh(1)} \cdot (2 - \eps)^{n + \frac{n}{c \cdot (d-1)\,!}}\\
		&= n^{\Oh(1)}\cdot (2 - \eps)^{n} \cdot \left( 1 + \frac{\eps}{2} \right) ^{n}\\
		&= n^{\Oh(1)} \cdot \left( 2 - \frac{\eps^{2}}{2} \right)^{n}\\
		&= n^{\Oh(1)} \cdot \left( 2 - \eps' \right)^{n}		
	\end{align*}
	where the first equality holds by \eqref{eq:U_bar_r_size} and the
	second equality holds because of the definition of $c$.
\end{proof}

\section{\boldmath \packing{\triangle} and \partition{\triangle}}\label{sec:clique-packing-partitioning}
In this section we give the proof of \cref{thm:mainpacking}. In \cref{thm:tri_part_SCC} we show that a \emph{fast} algorithm for $\partition{\triangle}$ violates the SCC, where the technical details will be discussed in \cref{sec:tri_part_SCC}. Similarly, in \cref{sec:triangletoSCC} we prove \cref{thm:triangletoSCC}, which shows that if the SCC fails, then there exists a \emph{fast} algorithm for $\packing{\triangle}$ problem. All in all, since $\packing{\triangle}$ problem is a generalization of $\partition{\triangle}$, \cref{thm:tri_part_SCC,thm:triangletoSCC} together imply that \cref{thm:mainpacking} holds.
\subsection{Reducing \textsc{Set Partition} to \partition{\triangle}}\label{sec:tri_part_SCC}

To proceed, we require a technical result asserting that we can assume the value of $d$ in \eqdsetpartition to be divisible by 3.

\begin{lem}\label{lem:set_part_mod_3}
	Suppose for some $\varepsilon > 0$ it holds that, for every $d \geq 1$
	there is an algorithm that solves every $n$-element instance of $\eqsetpartition{(3 \cdot d)}$ in
	time $(2 - \varepsilon)^{n} \cdot n^{\Oh(1)}$. Then for every $d \geq 1$ there is an
	algorithm that solves every $n$-element instance of $\eqdsetpartition$ in time $(2 -
	\varepsilon)^{n} \cdot n^{\Oh(1)}$.
\end{lem}

\begin{proof}
	Let $(U,\calF)$ be an $n$-element instance of \eqdsetpartition for some $d \geq 1$. Note that we can assume that $d$ divides $n$. Moreover, let $\calB$ be an algorithm that solves
	\eqsetpartition{h} for $h \coloneqq 3 \cdot d$ in time $(2 -
	\eps)^{n} \cdot n^{\Oh(1)}$.

	The	algorithm is very similar to the ones given in
	\cref{lem:leqPTstoPT,lem:leqPTtoeqPT,lem:leqPCtoeqPC}. Given an
	instance $(U,\calF)$ of \eqdsetpartition, let $ 0 \leq s \leq 2$ denote
	the integer $\frac{n}{d}$ modulo 3. The algorithm constructs $\tilde{\calF}$
	by incorporating $\left( 3 - s \right)$ pairwise disjoint sets of size
	$d$, each of which is also disjoint from $U$, into the existing set
	$\calF$.  Let $\calA$ denote the newly introduced sets and define
	$R \coloneqq \bigcup \calA$. The new universe $U'$ is equal to $U \cup
	R$, i.e., it has $n + d(3-s) \leq n+3d$ elements. Finally, the
	algorithm constructs $\calF'$ by adding the union of all pairwise
	disjoint sets $A,B$, and $C$ where $A,B,C \in \tilde{\calF}$. Observe that
	each set in $\calF'$ has size equal to $3 \cdot d$. In the end, the
	algorithm returns \yes if $\calB(U', \calF')$ returns \yes.
	
	To show that the algorithm is correct, assume that $(U,\calF)$ is a
	\yes-instance. Then, $U$ has a partition $\calS \subseteq \calF$.
	Observe that the size of $\calS \cup \calA$ is equal to 0 modulo 3
	since $\abs{\calA} = 3 - s$. Group the sets in $\calS \cup \calA$ into
	groups of size 3 and take the union of the sets in each group. Let
	$\calS'$ be the resulting set. It is easy to verify that $\calS'
	\subseteq \calF'$ and the union of the sets in $\calS'$ is equal to $U
	\cup R = U'$. Hence $(U', \calF')$ is also a \yes-instance and the
	algorithm returns \yes.

	Similarly, if $(U, \calF)$ is a no instance, meaning that $U$ has no
	partition $\calS \subseteq \calF$, adding sets disjoint from those in
	$\calF$ does not create a set system where the universe can be
	partitioned. Hence, since $\calB(U', \calF')$ returns \no, 
	our algorithm returns \no as well.

	Constructing the instance $(U',
	\calF')$ takes polynomial time in $n$.
	Therefore, the whole running time becomes
	\begin{align*}
		n^{\Oh(1)} \cdot (2 - \eps)^{\abs{U'}} &= n^{\Oh(1)} \cdot (2 - \eps)^{n + 3 \cdot d}
		= n^{\Oh(1)} \cdot \left( 2 - \eps \right)^{n}.
	\end{align*}
	
\end{proof}

The following corollary can be easily deduced from \cref{thm:SCCequivalent} and \cref{lem:set_part_mod_3}.
\begin{cor}\label{cor:setpartition_hardness}
	For all $\varepsilon > 0$, there exists $d \geq 1$ such that there is no algorithm that solves every $n$-element instance of \eqsetpartition{(3 \cdot d)} in time $(2 - \varepsilon)^{n} \cdot n^{\bigO(1)}$, unless the SCC fails.
\end{cor}

\begin{thm}\label{thm:tri_part_SCC}
	For all $\varepsilon > 0$, there exists $\sigma, \delta \geq 1$  such that there is no algorithm that solves every $n$-vertex instance of \partition{\triangle} given with a \core{\sigma}{\delta} of size at most $p$ in time $(2- \varepsilon)^{p} \cdot n^{\bigO(1)}$, unless the SCC fails.
\end{thm}

\begin{proof}
Suppose  there exists $\eps > 0$ such that for all $\sigma, \delta \geq 1$ there exists an algorithm that solves every $n$-vertex $\partition{\triangle}$ instance given with a \core{\sigma}{\delta} of size at most $p$ in time $(2 - \eps)^{p} \cdot n^{\Oh(1)}$. 
Fix arbitrary $d \geq 1$ and denote $r=3 \times d$. Consider an instance $(U, \mathcal{F})$ of $\eqsetpartition{r}$, where $U = [n]$. In the following, we will describe an algorithm that solves $(U, \mathcal{F})$ in time $(2 - \eps)^{n} \cdot n^{\Oh(1)}$, which will contradict the SCC by \cref{cor:setpartition_hardness}.

\paragraph{Equality Gadget.}
A \emph{\trieq gadget} is a graph with a set of designated vertices called \emph{portals} that has exacly two triangle packings that cover all non-portal vertices:
\begin{myitemize}
\item one that also covers all portals (i.e., it is a triangle partition in the gadget), and
\item one that covers no portal.
\end{myitemize}
A \trieq gadget behaves similarly to gadgets that we introduced for variants of \coloring{q}, but this time for the \partition{\triangle} problem.
When constructing an instance $G$ of \partition{\triangle}, only the portal vertices of the gadget will have neighbors outside the gadget.
Consequently, in any triangle packing of $G$, all portals of the \trieq gadget will be in the same state: either they are all covered by triangles contained inside the gadget, or none of them is covered by such a triangle. This explains why we use the name \emph{equality gadget}.

Now let us show how to construct a \trieq gadget $Z$ with $r$ portals and $4r$ vertices in total; see \cref{fig:trieqgadget}.
Introduce four sets of vertices $P = \{p_0, \ldots,
p_{r-1}\} $, $Q = \{q_0, \ldots, q_{r-1}\}$, $A = \{a_0, \ldots, a_{r-1}\}$, and $B = \{b_0,\ldots, b_{r-1}\}$. All arithmetic iterations on indices of these vertices are performed modulo $r$.

The vertices from $A$ and $B$ form a cycle with concecutive vertices $a_1, b_1, a_2, b_2, \ldots, a_r, b_r$.
Then, for each $i \in [r]$, we create a triangle on vertices $a_i,b_i,p_i$;
let $\mathbb{P}_1$ denote the set of such triangles.
Similarly, for each $i \in [r]$, we create a triangle $q_i,b_i,a_{i+1}$; let $\mathbb{P}_2$ denote set of these triangles.
Finally, for each $i \in  [\frac{r}{3} -1]$, we create a triangle on vertices $q_{3i + 1}, q_{3i + 2}, q_{3i + 3}$; call the set of these triangles $\mathbb{P}_3$. 
This completes the construction of $Z$ and $P$ is the set of portal vertices.

\begin{figure}[ht]
    \centering
    \includegraphics[width=0.4\textwidth]{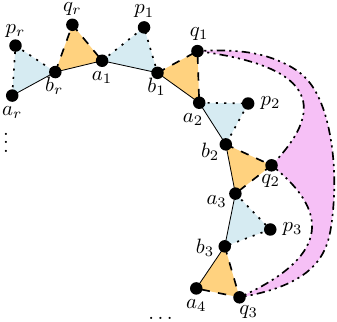}
    \caption{The construction of the \trieq gadget with $r$ portal vertices. The
	blue, orange, and violet triangles, with dotted, dashed, and
	dotted-dashed edges respectively, belong to $\mathbb{P}_1$,
        $\mathbb{P}_2$, and $\mathbb{P}_3$ respectively.}
    \label{fig:trieqgadget}
\end{figure}

Let us argue that $Z$ indeed has the property of a \trieq gadget.
Let $\mathcal{P}$ be a triangle packing in $Z$ that covers all non-portal vertices, i.e., $A \cup B \cup Q$.
Suppose $\mathcal{P}$ contains a triangle from $\mathbb{P}_1$, say, $p_i, a_i, b_i $ for $1 \leq i \leq r$.
In that case, since  $\mathcal{P}$ covers $a_{i + 1}$, it must also contain the triangle $p_{i + 1}, a_{i + 1}, , b_{i + 1}$.
By induction, it is easy to show that
$\mathcal{P}$ should include all triangles in $\mathbb{P}_1$.
The remaining vertices, i.e., $Q$ must be covered by the triangles in $\mathbb{P}_3$.
Therefore, in this case all the portal vertices of $Z$ are covered by $\mathcal{P}$.

Assume now that $\mathcal{P}$ has no triangle from $\mathbb{P}_1$.
In particular, no portal vertices are covered by $\mathcal{P}$.
We need to argue that such $\mathcal{P}$ exists and is unique.
Note that in order to cover all non-portal vertices, $\mathcal{P}$ must contain all the triangles in $\mathbb{P}_2$.
Hence the constructed graph is indeed a \trieq gadget.

\paragraph{Construction of the graph $G$.} Given a \eqsetpartition{r}
instance $(U,\mathcal{F})$, let us create an instance $G$ of
$\partition{\triangle}$ as follows. First, create $n$ vertices
$V_{U} = \{v_1, \ldots, v_n\} $, where each $v_i$ corresponds to $i \in U$. Then, for each $S \in \mathcal{F}$, create a \trieq
gadget, denoted by \trieqS, with $r$ portal vertices, and identify the
portal vertices of \trieqS with corresponding vertices of $V_{U}$. Let $G$ be the resulting graph obtained
by applying these operations.
Note that $G$ has $\Oh(n^r)$ vertices and $V_{U}$
is a \core{\sigma}{\delta} of $G$ of size $n$, where $\sigma = 4r$ and $\delta = r$. 

\paragraph{Equivalence of instances.} Suppose $U$ has a partition $\mathcal{S} \subseteq \mathcal{F}$.
For each $S \in \mathcal{S}$, define $V_S = \{ v_i ~|~ i \in S\}$.
Note that since $\mathcal{S}$ is a partition of $U$, sets $V_S$ form a  partition of the vertices $V_U$.
For $S \in \mathcal{S}$, let $\mathcal{T}_S$ be a
triangle partition of \trieqS which covers all its portal vertices, i.e., $V_S$; it exists by the definition of \trieq gadget.
For $S \in \mathcal{F} \setminus \mathcal{S}$, let $\mathcal{T}_S$ be a
triangle partition of \trieqS which covers no portal vertices; again, it exists by the definition of \trieq gadget.
Note that $\bigcup_{S \in \mathcal{F}} \mathcal{T}_{S}$ is a triangle partition of the whole graph $G$.

Now, suppose $G$ has a triangle partition $\mathcal{T}$. Let us call a copy \trieqS of the \trieq gadget in $G$ \emph{active} if there is no triangle that intersects the gadget, but it not contained in the gadget.
Recall that this means that each portal of \trieqS is a part of some triangle not contained in \trieqS.
Define $\mathcal{S} = \{S \in \mathcal{F} \mid \text{\trieqS is active}\}$. With a similar reasoning as in the previous paragraph, it is straightforward to verify that $\mathcal{S}$ is a partition of $U$. Summing up,  $(U,\mathcal{F})$ is a \yes-instance of $\eqsetpartition{(3 \cdot d)}$ if and only if $G$ is a \yes-instance of $\packing{\triangle}$.

\paragraph{Running Time.}
Building the graph $G$ takes time polynomial time in $n$. Moreover, as mentioned before, $G$ has $\Oh(n^r)$ vertices and  $V_U$ is a \core{4r}{r} of $G$ of size $n$. Therefore,  the assumed algorithm for $\eqsetpartition{r}$, one can solve the instance $(U,\mathcal{F})$ in time
\begin{equation*}
	(2 - \eps)^{\abs{V_U}} \cdot |V(G)|^{\Oh(1)} = (2 - \eps)^n \cdot n^{\Oh(1)}.
\end{equation*}
This completes the proof.
\end{proof}

\subsection{Reducing $\packing{\triangle}$ to \textsc{Set Packing}}\label{sec:triangletoSCC}

The goal of this section is to prove \cref{thm:triangletoSCC}.

\begin{thm}\label{thm:triangletoSCC}
	Suppose the SCC fails. Then there exists $\eps>0$ such that, for all $\sigma,\delta \ge 1$, there is an algorithm that solves every $n$-vertex instance of $\packing{\triangle}$ given with a $\sdhub$ of size $p$ in time $(2-\eps)^p\cdot n^{\bigO(1)}$.
\end{thm}

In order to simplify the description of the proof, we split it into two steps.
First, we show how we can use a fast algorithm for \leqdsetpackingsets to solve an auxiliary variant of the problem called $\prepacking$.
Then, we use some color-coding idea to show how an algorithm for $\prepacking$ can be used to solve $\packing{\triangle}$.

Let us start with some definitions and notation.
Let $G$ be a graph and $Q \subseteq V(G)$.
Let $\mathcal{C}$ denote the set of components of $G-Q$ and let $\mathcal{D}$ denote the set of triangles of $G$ that are contained in $Q$.
Let $\Pi$ be a collection of pairwise vertex-disjoint triangles in $G$.
We say that $C \in \mathcal{C}$ is \emph{active} (with respect to $\Pi$) if there is a triangle in $\Pi$ that intersects both $C$ and $Q$.
A triangle in $\mathcal{D}$ is \emph{active} if it belongs to $\Pi$.

For a function $\psi : \mathcal{C} \cup \mathcal{D} \to \N$, by $\mathcal{C}_i$ (resp., $\mathcal{D}_i$) we denote $\mathcal{C} \cap \psi^{-1}(i)$ (resp., $\mathcal{D} \cap \psi^{-1}(i)$).

For a constant $c$, an instance of $\prepacking$ is a quadruple $(G,t,Q,\psi)$, where $G$ is graph, $t$ is an integer, $Q$ is a subset of $V(G)$,
and $\psi$ is a function that maps elements of $\mathcal{C} \cup \mathcal{D}$ to numbers in $\left\{1, \ldots, \big\lceil |Q|/c \big\rceil \right\}$ (here $\mathcal{C}$ and $\mathcal{D}$ are defined as in the previous paragraph).
We call the numbers in the codomain of $\psi$ \emph{colors} and call $\psi$ a \emph{coloring}; we emphasize that this is an arbitrary coloring and has no correctness criterion.
We ask whether $G$ admits a packing of at least $t$ triangles such that, for each color $i$, at most $c$ elements of $\mathcal{C}_i \cup \mathcal{D}_i$ are active (the sets $\mathcal{C}_i$ and $\mathcal{D}_i$ are defined with respect to $\psi$).

\begin{lem}\label{lem:SetPackingToPrecolored}
Suppose there is $\eps'>0$ such that for all $d$, \leqdsetpackingsets with $n$-element universe can be solved in time $(2-\eps')^{n}\cdot n^{\bigO(1)}$.\\
Then there is $\eps''>0$ such that for every $\sigma,\delta \geq 1$ there is $c_0$ depending only on $\epsilon''$ and $\sigma$, so that for every $c \geq c_0$, every instance $(G,t,Q,\psi)$ of $\prepacking$, where $Q$ is a $\sdhub$ of size $p$, can be solved in time 
$(2-\eps'')^p \cdot |V(G)|^{\bigO(1)}$.
\end{lem}

\begin{proof}
Let $\sigma,\delta$ be fixed constants. Without loss of generality assume that $\sigma \geq 2$.
Let 
\begin{equation}
\epsilon'' = \epsilon'/4 \label{eq:chooseeps}
\end{equation}
and let $c_0$ be the smallest integer that satisfies both
\begin{align}
(c_0\sigma+1)^{1/c_0}\leq & (1+\epsilon'')  \label{eq:choosec1} \\
(2-\epsilon')^{1/c_0}\leq & (1+\epsilon''); \label{eq:choosec2}
\end{align}
and pick any $c \geq c_0$; clearly $c$ satisfies these inequalities too.
This choice will become clear later in the proof.
Note that $\epsilon''$ depends only on $\epsilon'$ and  $c_0$ depends only on $\eps'$ and $\sigma$.

Let $(G,t,Q,\psi)$ be an instance of $\prepacking$, where $Q$ is a $\sdhub$ of size $p$. By introducing at most $c-1$ isolated vertices to $Q$, we can assume that $c$ divides $p$; note that these dummy vertices contribute to the running time only by a constant factor.
We will use the assumed algorithm for  \leqdsetpackingsets for $d = 2c\sigma+1$ to solve $(G,t,Q,\psi)$ in time $(2-\eps'')^p \cdot |V(G)|^{\bigO(1)}$.

Let $\ell=p/c$ denote the number of colors used by $\psi$.
The sets $\calC, \calD$ and $\calC_i, \calD_i$ over all $i \in [\ell]$ are defined as previously.

Fix some (unknown) optimum solution $\Pi$, i.e., a largest triangle packing in $G$ that has at most $c$ active elements in each color.
A \emph{contribution} of a component $C \in \calC$ is the number of triangles in $\Pi$ that intersect $C$ (they can either be contained in $C$, or have some vertices in $C$ and some in $Q$). The contribution of a triangle in $\calD$ is 1 if this triangle is in $\Pi$ and 0 otherwise. A contribution of a color $i \in [\ell]$ is the total contribution of all elements in $\calC_i \cup \calD_i$.

In the algorithm, for each color $i$, we would like to exhaustively guess the contribution of this color.
However, the numbers involved might be very large, as $\calC_i$ might contain many components.
Thus we are going to do this indirectly.

Fix a color $i \in [\ell]$.
Let $X_i$ be the number of triangles of a maximum triangle packing of the graph induced by the components in $\calC_i$.
Note that $X_i$ can be computed in time linear in $|V(G)|$ as each component in $\calC_i$ has size at most $\sigma$ and there are at most $|V(G)|$ such components.
A maximum triangle packing can be obtained as the union of maximum triangle packings of the individual components.

Note that the contribution of the color $i$ is at least $X_i$ and at most $X_i + c\sigma$.
Indeed, on one hand, we can pick $X_i$ triangles that are contained in the components in $\calC_i$, without interfering with any object of any other color.
On the other hand, each active triangle in $\calD_i$ contributes 1 to the sum, while each of at most $\sigma$ vertices of each active component in $\calC_i$ can be present in at most one triangle intersecting $Q$. Since there at at most $c$ active objects, we get the upper bound.

So, instead of guessing the contribution of each color $i \in [\ell]$ directly, we guess the offset $q_i$ of this value against $X_i$. Formally, the algorithm iterates over all tuples of the form $\boldq=(q_1, \ldots, q_{\ell})\in \{0, \ldots, c\sigma\}^{\ell}$ that satisfy $\sum_{i\in [\ell]} (X_i+q_i) \ge t$. Note that if no such tuple exists, the discussion above yields that $(G,t,Q,\psi)$ is a \no-instance.
Observe that the number of choices for $\boldq$ is at most
\begin{equation}
	(c\sigma+1)^\ell = (c\sigma+1)^{p/c} \overset{\text{\eqref{eq:choosec1}}}{\leq} (1 + \epsilon'')^p. \label{eq:iterateoverq}
\end{equation}
	
Fix one such tuple $\boldq=(q_1, \ldots, q_{\ell})$ and consider a color $i \in [\ell]$.	
We say that a subset $S\subseteq Q$ of vertices from the hub is \emph{$i$-valid} if it has size at most $2c\sigma$ and the graph induced by $S$ together with the components in $\calC_i$ has a triangle packing $\Pi_S$ with the following properties:
	\begin{itemize}
		\item at most $c$ elements of $\calD_i \cup \calC_i$ are active w.r.t.~$\Pi_S$,
		\item all triangles of $\Pi_S$ contained in $Q$ are in $\calD_i$,
		\item the number of triangles in $\Pi_S$ is at least $X_i+q_i$.
	\end{itemize}

Intuitively, a set $S$ is $i$-valid if it is compatible with the choice of $q_i$.
The size bound comes from the fact that each active triangle in $\calD_i$ uses three vertices from $Q$,
while the triangles intersecting each active component from $\calC_i$ intersect at most $2\sigma$ vertices of $Q$. As $\sigma \geq 2$ and there are at most $c$ active elements in $\calD_i \cup \calC_i$, the bound follows.
	
\begin{clm}\label{clm:ivaldsets}
For each color $i$ and offset $q_i$, the $i$-valid sets can be enumerated in time polynomial in $|V(G)|$.
\end{clm} 
\begin{claimproof}
There are at most $\Oh(p^{2c\sigma})=|V(G)|^{\bigO(1)}$ subsets of $Q$ of size at most $2c\sigma$.
As each such candidate set $S$ is of constant size, in polynomial time we can enumerate all possible triangle packings in which every triangle intersects $S$. Furthermore, in polynomial time, each such packing can be extended in an optimal way by picking triangles contained in $\calC_i$; this is possible as each component in $\calC_i$ is of constant size.
It is straightforward to verify that $S$ is $i$-valid if and only if any of the packings obtained that way satisfies the conditions for $\Pi_S$.
\end{claimproof}
	
For each tuple $\boldq$, we can now define an instance $(U,\calF_\boldq,\ell)$ of \leqdsetpackingsets as follows.
The universe $U$ of the instance consists of all vertices in the hub $Q$ together with a distinct element $a_i$ per each color $i \in [\ell]$, so in total $|U| = p + \ell = p + p/c$.
The set system $\calF_\boldq$ contains, for each color $i \in [\ell]$ and for each $i$-valid set $S\subseteq Q$,
the set $S\cup \{a_i\}$.
Note that each set in $\calF_\boldq$ is of size at most $d=2c\sigma+1$, and the instance can be constructed in time polynomial in $|V(G)|$. 
The algorithm iterates over all such instances and executes the previously-mentioned algorithm for \leqdsetpackingsets on each of these instances.  If in any of the calls we obtain a positive answer, we report that $(G,t,Q,\psi)$ is a \yes-instance, and otherwise we reject it.

\paragraph*{Correctness}
The correctness of the algorithm follows from the following two claims.	
	
\begin{clm}\label{clm:triangletoSCCcorrect1}
If there is $\boldq$ for which $(U,\calF_\boldq,\ell)$ is a \yes-instance of  \leqdsetpackingsets, then $(G,t,Q,\psi)$ is a \yes-instance of \prepacking.
\end{clm}
\begin{claimproof}
Let $\calF^* \subseteq \calF_\boldq$ be a family $\ell$ pairwise disjoint sets.
Since each set in $\calF_\boldq$ contains one of the elements $\{a_1, \ldots, a_{\ell}\}$, the solution $\calF^*$ contains, for each color $i$, precisely one set of the form $S_i\cup \{a_i\}$.
Note that $S_i$ may be empty. For each such set, let $\Pi_{S_i}$ be a triangle packing defined as above.
We claim that $\bigcup_{i \in [\ell]} \Pi_{S_i}$ is a triangle packing that witnesses that $(G,t,Q,\psi)$ is a \yes-instance.

Clearly the triangles within one set $\Pi_{S_i}$ are pairwise disjoint.
Now consider two triangles from distinct sets, say $\Pi_i$ and $\Pi_j$.
As sets of components $\calC_i$ and $\calC_j$ are pairwise disjoint, these two triangles may overlap only on $Q$,
which means that $S_i \cap S_j \neq \emptyset$. However, this is not possible as sets in $\calF$ are pairwise disjoint.

Next, by the conditions in the definition of $i$-valid set, we observe that at most $c$ elements from $\calC_i \cup \calD_i$ are active. Finally, the total number of triangles is at least
$ \sum_{i \in [\ell]} (X_i + q_i)$ which is at least $t$ by the choice of $\boldq$.
\end{claimproof}
	
\begin{clm}\label{clm:triangletoSCCcorrect2}
If $(G,t,Q,\psi)$ is a \yes-instance of \prepacking, then there is  $\boldq$ for which $(U,\calF_\boldq,\ell)$ is a \yes-instance of  \leqdsetpackingsets.
\end{clm}
	\begin{claimproof}
Suppose there exists a triangle packing $\Pi$ of $G$ of size at least $t$,
that has at most $c$ active elements of each color. 
Fix color $i \in [\ell]$ and let $\Pi_i$ be the subfamily of $\Pi$ consisting of triangles that either are in $\calD_i$, or intersect $\calC_i$. Clearly the sets $\Pi$ over all $i \in [\ell]$ form a partition of $\Pi$.
Recall that we safely assume that $X_i \leq |\Pi_i| \leq X_i + c\sigma$
Define $q_i := |\Pi_i| - X_i$.
Note that $\boldq = (q_1,\ldots,q_\ell)$ is one of the tuples considered by our algorithm.

For $i \in [\ell]$, let $S_i \subseteq Q$ be the set consisting of vertices of triangles in $\Pi_i$.
The packing $\Pi_i$ is a witness that $S_i$ is an $i$-valid set (with respect to $\boldq$).
Furthermore, sets $S_i \cup \{a_i\}$ are pairwise disjoint. Thus,  $(U,\calF_\boldq,\ell)$ is a \yes-instance of  \leqdsetpackingsets.
	\end{claimproof}

\paragraph*{Runtime}
The algorithm iterates over all possible choices of $\boldq$, for each choice it builds an instance $(U,\calF_\boldq)$ in polynomial time, and then calls the assumed algorithm for \leqdsetpackingsets. As $|U|$ and $|\calF_\boldq|$ are bounded by a polynomial function of $|V(G)|$, and by \eqref{eq:iterateoverq}, the total running time is bounded by 
\begin{align*}
 &	(1 + \epsilon'')^p \cdot (2-\epsilon')^{|U|} \cdot |V(G)|^{\bigO(1)} \leq (1+\mu)^p \cdot (2-\epsilon')^{p(1+1/c)} \cdot |V(G)|^{\bigO(1)} \\
\leq   & \ \left( (1+\epsilon'') \cdot (2-\epsilon') \cdot (2-\epsilon')^{1/c} \right)^p \cdot |V(G)|^{\bigO(1)}   \overset{\text{\eqref{eq:choosec2}}}{\leq }  \left( (1+\epsilon'')^2 \cdot (2-\epsilon') \right)^p \cdot |V(G)|^{\bigO(1)}  \\
\overset{\text{\eqref{eq:chooseeps}}}{\leq }   & \ (2-\epsilon'')^p \cdot |V(G)|^{\bigO(1)},
\end{align*}
as claimed. This completes the proof.
\end{proof}
	
Now let us argue that we can reduce solving $\packing{\triangle}$ (where the instance is given with a hub) to $\prepacking$ for suitably chosen $c$.

In this reduction we use some color coding idea. We will use the so-called \emph{splitters}.
For integers $N$, $p$, and $\ell$ with $N\ge p$, an \emph{$(N,p,\ell)$-splitter} $\Psi$ is a family of functions from $[N]$ to $[\ell]$ such that, for each subset $S$ of $[N]$ with size $\abs{S}=p$, there is an $\psi\in \Psi$ that assigns colors from $[\ell]$ as evenly distributed as possible to the elements of $S$ --- formally, for $(p \bmod \ell)$ colors $j\in [\ell]$ it holds that $\abs{\psi^{-1}(j)}=\ceil{p/\ell}$, and for the remaining colors we have $\abs{\psi^{-1}(j)}=\lfloor{p/\ell}\rfloor$.
The methods of constructing splitters were introduced by Naor, Schulman, and Srinivasan~\cite{NaorSS95}.

\begin{thm}[{\cite[Theorem 3 (ii)]{NaorSS95}}]\label{thm:spitter}
Let $\ell=\omega(p)$ and let $f(p,\ell)=(2\pi p/\ell)^{\ell/2}e^{\ell^2/(12p)}$.
An $(N,p,\ell)$-splitter of size $s=\bigO(f(p,\ell)^{1+o(1)} \log N)$ can be computed in time polynomial in $s$ and $N$.
Here the $o(1)$ in the exponent is for $\ell/\sqrt{p}$ going to infinity.
\end{thm}

\begin{cor}\label{cor:splitter}
For every $\eps>0$ there is $c_1$ such that for each $N$, $p$ with $p\le N$, $c \geq c_1$, an $(N,p,\ceil{p/c})$-splitter can be computed (and iterated over) in time $(1+\eps)^p\cdot N^{\bigO(1)}$.
\end{cor}
\begin{proof}
Let $c_1$ be the smallest integer such that $(2\pi c_1)^{4/c_1}\le (1+\eps)$.
Choose $N,p \leq N$, and $c \geq c_1$.
Note that $\ell:=\ceil{p/c}=\omega(\sqrt{p})$.
	Thus, by \cref{thm:spitter}, there is an $(N,p,\ell)$-splitter of size $s=\bigO(f(p,\ell)^{1+o(1)}\log N)$ that can be computed in time polynomial in $N$ and $s$, where $f(p,\ell)= (2\pi p/\ell)^{\ell/2}e^{\ell^2/(12p)}$.
	The $o(1)$ in the exponent is for $\ell/\sqrt{p}=\ceil{p/c}/\sqrt{p}$ going to infinity, i.e., it is for $p$ going to infinity.
	If $p$ is smaller than some constant the size of the splitter is trivially $\bigO(\log N)$. So we can assume that $p$ is sufficiently large such that $f(p,\ell)^{1+o(1)}\le f(p,\ell)^2$ and $p>c$. Since $p>c$ we have $\ceil{p/c}\le 2p/c$. Then
	\[
	f(p,\ell)= (2\pi p/\ell)^{\ell/2}e^{\ell^2/(12p)}\le (2\pi c)^{p/c} e^{4p^2/(12p c^2)}\le (2\pi c)^{2p/c}.
	\]
	Using the fact that $(2\pi c)^{4/c}\le (1+\eps)$, it follows that
	\[
	f(p,\ell)^{1+o(1)}\le f(p,\ell)^2\le (2\pi c)^{4p/c}\le (1+\eps)^p.
	\]
	So we can compute an $(N,p,\ell)$-splitter and then iterate over its members in time $(1+\eps)^p \cdot N^{\bigO(1)}$. 
\end{proof}	
	
Let us proceed to the proof of the following result.

\begin{lem}\label{lem:PrecoloredToTriangles}
Suppose there is $\eps''>0$ such that for every $\sigma,\delta \geq 1$ there is $c_0$ depending only on $\epsilon''$ and $\sigma$, so that for every $c \geq c_0$, every instance $(G,t,Q,\psi)$ of $\prepacking$, where $Q$ is a $\sdhub$ of size $p$, can be solved in time 
$(2-\eps'')^p \cdot |V(G)|^{\bigO(1)}$.\\
Then there is $\eps>0$ such that for every $\sigma,\delta \geq 1$,
every $n$-vertex instance of $\packing{\triangle}$ given with a $\sdhub$ of size $p$ can be solved in time 
$(2-\eps)^p \cdot n^{\bigO(1)}$.
\end{lem}

\begin{proof}
Define $\epsilon = \epsilon''/3$.
Let $c_0$ be as in the assumption of the lemma, and let $c_1$ be as in \cref{cor:splitter} (for $\eps$).
Let $c = \max(c_0,c_1)$.

Fix $\sigma$ and $\delta$, and consider an $n$-vertex instance $(G,t)$ of $\packing{\triangle}$, given with a $\sdhub$ $Q$ of size $p$.
Again, by adding at most $c-1$ isolated vertices to the hub, we can assume that $c$ divides $p$. Define $\ell = p/c$.
Let $\calC$ be the set of components of $G-Q$, and let $\calD$ be the set of triangles contained in $Q$.
Define $N = |\calC + \calD| = \Oh(n^{3})$.

Let $\Psi$ be the $(N,p,\ell)$-splitter given by \cref{cor:splitter} for $N,p,c$.
We claim that $(G,t)$ is a \yes-instance of $\packing{\triangle}$ if and only if there exists $\psi \in \Psi$ such that $(G,t,Q,\psi)$ is a \yes-instance of $\prepacking$.

The backwards implication is trivial; let us show the forward one.
Suppose there is a triangle packing $\Pi$ in $G$, consisting of at least $t$ triangles.
Let $\calC^*$ and $\calD^*$ be the elements of $\calC$ and $\calD$, respectively, that are active with respect to $\Pi$.
Notice that $| \calC^*  \cup \calD^*| \leq p$. Indeed, each active element of $\calC  \cup \calD$ uses at least one vertex of $Q$ and these vertices are pairwise distinct. 
Let $S$ be a subset of $\calD \cup \calC$ of size exactly $p$ that contains $ \calC^*  \cup \calD^*$.

We interpret elements from $\Psi$ as functions from $\calC + \calD$ (recall that $|\calC + \calD|=N$) to $[\ell]= [p/c]$.
By the definition of the splitter, there is $\psi \in \Psi$ in which each color appears exactly $\frac{p}{p/c} = c$ times on elements of $S$.
Thus, for each color $i$ there are at most $c$ active elements of $\calD \cup \calC$.
Consequently, $(G,t,Q,\psi)$ is a \yes-instance of $\prepacking$.

Now the algorithm is simple. First, it uses \cref{cor:splitter} to compute $\Psi$.
Then it iterates over all $\psi \in \Psi$, and calls the assumed algorithm for the instance $(G,t,Q,\psi)$ of $\prepacking$. If any the calls succeeds, the algorithm reports a \yes-instance, and otherwise it reports a \no-instance.
The correctness follows from the reasoning above.

The running time is determined by the number of elements of $\Psi$ times the time needed to call the algorithm for each instance of $\prepacking$, thus it is bounded by 
\[
(1+\epsilon)^p \cdot (2-\epsilon'')^p \cdot n^{\bigO(1)} = (2 - \epsilon'' + 2\epsilon''/3 - \epsilon''^2/3)^p \cdot n^{\bigO(1)} \leq (2 - \epsilon)^p \cdot n^{\bigO(1)},
\]
as claimed. This completes the proof.
\end{proof}

Now \cref{thm:triangletoSCC} follows directly from the combination of \cref{thm:SCCequivalent}, \cref{lem:SetPackingToPrecolored}, and \cref{lem:PrecoloredToTriangles}.

\section{Dominating Set}\label{sec:dominating}
In this section we turn our attention to \textsc{Dominating Set} parameterized by the hub size.
We show two quite simple reductions that exclude an algorithm with running time $(2-\epsilon)^p \cdot |V(G)|^{\bigO(1)}$ (where the instance $G$ is given with a hub of size $p$) under two complexity assumptions: the SETH and the SCC.

\begin{thm}\label{thm:domsetSCC}
For every $\epsilon > 0$ there exists $\delta$ such that the \textsc{Dominating Set} problem on $n$-vertex instances given with a \core{3}{\delta}  of size $p$ cannot be solved in time $(2-\epsilon)^p \cdot |V(G)|^{\bigO(1)}$, unless the SCC fails.
\end{thm}
\begin{proof}
Assume the SCC. Given some $\eps >0$, let $d$ be the constant given by the SCC for this $\eps$.
We reduce from \leqdsetcover, let $(U,\mathcal{F})$ be an instance where $|U|=n$ and $|\mathcal{F}|=m$.
We start the construction of the instance $G$ of \textsc{Dominating Set}  with introducing a set $Y$ which contains a vertex $y_i$ for every $i \in U$.
Next, for every set $F \in \mathcal{F}$, we proceed as follows.
We introduce a three-vertex path with consecutive vertices $a_F,b_F,c_F$. Then we add an edge $a_Fy_i$ if and only if $i \in F$. Denote $A = \{a_F ~|~ F \in \mathcal{F}\}$, $B = \{b_F ~|~ F \in \mathcal{F}\}$, and $C = \{c_F ~|~ F \in \mathcal{F}\}$.
This completes the construction of $G$.

Note that this construction can be performed in time polynomial in $n$ (recall that $d$ is a constant and thus $m$ is polynomial in $n$) and the constructed graph has $n+3m = \Oh(n^d)$ vertices. Furthermore, the set $Y$ is a \core{3}{d} of size $n$.

We claim that $G$ has a dominating set of size at most $k + m$ if and only if $\mathcal{F}$ contains at most $k$ sets that cover $U$.

First, suppose that there is a subfamily $\mathcal{F'} \subseteq \mathcal{F}$ of size at most $k$ such that $U = \bigcup \mathcal{F}'$. Define $X = B \cup \{a_F ~|~ F \in \mathcal{F'} \}$. Clearly $X$ is of size  $k + m$. Let us show that it is a dominating set.

Notice that for each $F \in \mathcal{F}$, the vertices $a_F,b_F,c_F$ are dominated by $b_F \in X$.
Now consider any $y_i \in Y$. As $\mathcal{F}'$ covers $U$, there is $F \in \mathcal{F}'$ such that $i \in F$.
This means that $a_F \in X$ and $y_i$ is dominated by $a_F$.

For the other direction, suppose that $G$ has a dominating set $X$ of size at most  $k + m$.
We claim that we can assume that $B \subseteq X \subseteq A \cup B$.
Indeed, consider $F \in \mathcal{F}$. Notice that in order to dominate $c_F$, the set $X$ must contain at least one of $b_F,c_F$. However, if it contains $c_F$, we can obtain a solution of at most the same size by replacing $c_F$ with $b_F$ (if $b_F \notin X$) or removing $c_F$ (otherwise). Thus we may assume that $B \subseteq X$ and $C\cap X=\emptyset$.
Now, for contradiction, suppose that there is $y_i \in Y \cap X$. The vertex $y_i$ can only dominate itself or some set of vertices in $A$. However, $B\subseteq X$ already dominates all vertices in $A$. Hence, $X \setminus \{y_i\}$ dominates every vertex of $G$, except possibly $y_i$. So consider any $F \in \mathcal{F}$ such that $i \in F$; recall that such $F$ exists as $U = \bigcup \mathcal{F}$.
Then the set $X \setminus \{y_i\} \cup \{a_F\}$ is a dominating set of size at most $|X|$.

Define $\mathcal{F}' = \{ F ~|~ a_F \in X\}$; clearly $|\mathcal{F}'| \leq k$. We claim that $\bigcup \mathcal{F}' = U$.
Suppose there is $i \in U$ that is not in $\bigcup\mathcal{F}'$. However, this means that $y_i$ is not dominated by $X$, a contradiction.

Summing up, a hypothetical algorithm that solves \textsc{Dominating Set} on $G$ in time $(2-\epsilon)^n \cdot |V(G)|^{\bigO(1)}$ can be used to solve the instance $(U, \mathcal{F})$ of \leqdsetcover in time $(2-\epsilon)^n \cdot n^{\bigO(1)}$, contradicting the SCC.
\end{proof}

In the $\leq d$-\textsc{Hitting Set} problem the instance is a pair $(U, \mathcal{F})$, where $U$ is a set called the \emph{universe} and $\mathcal{F}$ is a family of subsets of $U$, each of size at most $d$.
We ask for a minimum \emph{hitting set}, i.e., a minimum-sized subset of $U$ that intersects every element of $\mathcal{F}$.

\begin{thm}[Cygan et al.~\cite{cyganProblemsHardCNFSAT2016}]\label{thm:hittingset}
For every $\epsilon > 0$ there exists $d$ such that the  $\leq d$-\textsc{Hitting Set} with universe of size $n$ cannot be solved in time $(2-\epsilon)^n \cdot n^{\Oh(1)}$, unless the SETH fails.
\end{thm}

With an argument similar to the proof of \cref{thm:domsetSCC}, we can show the following result.

\begin{thm}
For every $\epsilon > 0$ there exists $\delta$ such that the \textsc{Dominating Set} problem on $n$-vertex instances given with a \core{2}{\delta}  of size $p$ cannot be solved in time $(2-\epsilon)^p \cdot |V(G)|^{\bigO(1)}$, unless the SETH fails.
\end{thm}
\begin{proof}
Let $\epsilon >0$ and let $d$ be the constant given for that $\epsilon$ by \cref{thm:hittingset}.
We reduce from  $\leq d$-\textsc{Hitting Set}. Let $(U,\mathcal{F})$ be a corresponding instance with $|U|=n$.

We construct $G$ as follows. First, we introduce three sets $Y = \{y_i ~|~ i \in U\}$, $A = \{a_i ~|~ i \in U\}$, and $B = \{b_i ~|~ i \in U\}$. For each $i \in U$ we add edges $y_ia_i$ and $a_ib_i$.
Next, for each $F \in \mathcal{F}$, we add a vertex $z_F$. We add the edge $y_iz_F$ if and only if $i \in F$.
This completes the construction of $G$. Note that $|V(G)| = \Oh(n^d)$ and $Y$ is a \core{2}{d}.

We claim that $G$ has a dominating set of size at most $n+k$ if and only if $(U,\mathcal{F})$ admits a hitting set of size at most $k$.
Let $U' \subseteq U$ be a hitting set of size at most $k$.
We define $X = A \cup \{y_i ~|~ i \in U'\}$. Clearly $X$ dominates $X \cup A \cup B$.
Suppose that there is some $z_F$ which is not adjacent to any vertex in $X$. This means that $F$ does not intersect $U'$, a contradiction.

Now suppose that $G$ has a dominating set $X$ of size at most $n+k$.
Similarly as in the proof of \cref{thm:domsetSCC} we can assume that $A \subseteq X \subseteq A \cup Y$.
Define $U' = \{i ~|~y_i \in X\}$; clearly $|U'| \leq k$. We claim that $U'$ is a hitting set.
Indeed, if there is some $F \in \mathcal{F}$ which is not intersected by $U'$, its corresponding vertex $z_F$ is not dominated by $X$.

Thus a hypothetical algorithm that solves \textsc{Dominating Set} on $G$ in time $(2-\epsilon)^n \cdot |V(G)|^{\bigO(1)}$ can be used to solve the instance $(U, \mathcal{F})$ of $\leq d$-\textsc{Hitting Set} in time $(2-\epsilon)^n \cdot n^{\bigO(1)}$. By \cref{thm:hittingset}, this contradicts the SETH.
\end{proof}

\addcontentsline{toc}{section}{References}
\bibliography{main}

\end{document}